\documentclass[11pt]{amsart}
\usepackage{fullpage}
\usepackage[foot]{amsaddr}
\usepackage{microtype}
\usepackage[OT1]{fontenc}
\usepackage{eulervm}
\usepackage[tt=false]{libertine} %% tt is ugly
\usepackage{bbold}
\usepackage{amsmath}
\usepackage{amssymb}
\usepackage{amsthm}
\usepackage{thmtools} % Adds cases* and dcases environments
\usepackage[linesnumbered,boxed,ruled,vlined]{algorithm2e}
\usepackage{algpseudocode}
\usepackage{enumitem}
\usepackage{bm}
\usepackage{xifthen}
\usepackage{xspace}
\usepackage{tikz}
\usetikzlibrary{arrows.meta,positioning,calc,fit,backgrounds}

\usepackage[margin=1cm]{caption} % Adds an addition margin on either side of figure captions. This must come before "\usepackage{subfig}" (otherwise there are errors).
\usepackage{subfig}

\usepackage[thinlines]{easytable}

\usepackage[bookmarks=true,hypertexnames=false,pagebackref]{hyperref}
\hypersetup{colorlinks=true, citecolor=blue, linkcolor=red, urlcolor=blue}

\usepackage{pgfplots}
\pgfplotsset{compat=1.16}
\usepackage{tikz}
\usetikzlibrary{arrows,arrows.meta,backgrounds,calc,fit,decorations.pathreplacing,decorations.markings,shapes.geometric}

\tikzstyle{internal} = [draw, fill, shape=circle]
\tikzstyle{external} = [shape=circle]
\tikzstyle{square}   = [draw, fill, rectangle]
\tikzstyle{triangle} = [draw, fill, regular polygon, regular polygon sides=3, inner sep=3pt]
\tikzstyle{pentagon} = [draw, fill, regular polygon, regular polygon sides=5, inner sep=2pt, minimum size=14pt]
\tikzset{every fit/.append style=text badly centered}

\usetikzlibrary{positioning,chains,fit,shapes,calc}
\usetikzlibrary{trees}
\usetikzlibrary{decorations.pathreplacing}
\usetikzlibrary{decorations.pathmorphing}
\usetikzlibrary{decorations.markings}
\tikzset{>=latex} % arrow tips
\usepackage{ifthen}

\usepackage{cleveref}

\usepackage[textsize=tiny]{todonotes}

\usepackage[normalem]{ulem}

\usepackage{mleftright}

\usepackage{cool}
\Style{DSymb={\mathrm d},DShorten=true,IntegrateDifferentialDSymb=\mathrm{d}}

\newcommand{\tp}[1]{{\left( #1 \right)}}

\newcommand{\Ex}{\mathop{\mathbb{{}E}}\nolimits}
\renewcommand{\Pr}{\mathop{\mathrm{Pr}}\nolimits}

\def\*#1{\mathbf{#1}}
\def\+#1{\mathcal{#1}}
\def\-#1{\mathrm{#1}}
\def\=#1{\mathbb{#1}}
\def\^#1{\mathbb{#1}}

\newcommand{\abs}[1]{\ensuremath{\left\vert#1\right\vert}}

\newcommand{\eps}{\varepsilon}

\newcommand{\Var}[2]{\ensuremath{\textnormal{Var}_{#1}\left(#2\right)}}

\newcommand{\cl}[1]{\operatorname{cl}\left(#1\right)}

\newcommand{\dist}{\operatorname{dist}}

\newcommand{\defeq}{:=}

\newcommand{\NP}{\textnormal{\textbf{NP}}}
\newcommand{\BIS}{\#\textnormal{\textbf{BIS}}}

\newcommand{\DTV}[2]{\-D_{\mathrm{TV}}\left({#1},{#2}\right)}

\newtheorem{theorem}{Theorem}

\newtheorem{lemma}[theorem]{Lemma}
\newtheorem{claim}[theorem]{Claim}
\newtheorem{observation}[theorem]{Observation}
\newtheorem{proposition}[theorem]{Proposition}
\newtheorem{corollary}[theorem]{Corollary}
\theoremstyle{definition}
\newtheorem{condition}[theorem]{Condition}
\newtheorem{definition}[theorem]{Definition}

\theoremstyle{remark}
\newtheorem*{remark}{Remark}

\crefname{theorem}{Theorem}{Theorems}
\crefname{observation}{Observation}{Observations}
\crefname{claim}{Claim}{Claims}
\crefname{condition}{Condition}{Conditions}
\crefname{algorithm}{Algorithm}{Algorithms}
\crefname{property}{Property}{Properties}
\crefname{example}{Example}{Examples}
\crefname{fact}{Fact}{Facts}
\crefname{lemma}{Lemma}{Lemmas}
\crefname{corollary}{Corollary}{Corollaries}
\crefname{definition}{Definition}{Definitions}
\crefname{remark}{Remark}{Remarks}
\crefname{proposition}{Proposition}{Propositions}
\crefname{equation}{equation}{equations}
\crefname{enumi}{Case}{Case}
\creflabelformat{enumi}{(#2#1#3)}

\definecolor{HGcolor}{RGB}{255,50,50}

\makeatletter
\def\prob#1#2#3{\goodbreak\begin{list}{}{\labelwidth\z@ \itemindent-\leftmargin
      \itemsep\z@  \topsep6\p@\@plus6\p@
      \let\makelabel\descriptionlabel}
  \item[\textbf{Name}]#1
  \item[\textbf{Instance}]#2
  \item[\textbf{Output}]#3
  \end{list}}
\makeatother

\makeatletter
\providecommand\@dotsep{5}
\def\listtodoname{Todo list}
\def\listoftodos{\@starttoc{tdo}\listtodoname}
\makeatother

\usepackage{nicefrac,comment}

\usepackage{silence}
\newboolean{doubleblind}
\setboolean{doubleblind}{false}

\begin{document}

\title{Fast counting and sampling for ferromagnetic two-spin systems}

\author{Weiming Feng}
\author{Heng Guo}
\author{Yichun Yang}
\address[Weiming Feng]{School of Computing and Data Science, The University of Hong Kong, Hong Kong, China}
\address[Heng Guo]{School of Informatics, University of Edinburgh, Informatics Forum, Edinburgh, EH8 9AB, United Kingdom}
\address[Yichun Yang]{School of Computer Science, Beijing Institute of Technology, Beijing, China}

\begin{abstract}
  We introduce two new models equivalent to ferromagnetic two-spin systems: a weighted subgraph model and a random cluster type model.
  Using these new connections, we obtain an efficient sampling algorithm and a new randomised algorithm that efficiently approximates the partition function of ferromagnetic two-spin systems in certain parameter regimes. 
  No efficient sampling algorithms are known before in this regime,
  and our new estimation algorithm runs in near-quadratic time for bounded degree graphs and in polynomial time for general graphs,
  improving upon the previous algorithm of Guo, Liu, and Lu (2020).
\end{abstract}

\maketitle

\section{Introduction}

Spin systems model nearest neighbour interactions.
They originate from statistical physics, and are now widely studied in other areas including probability theory, machine learning, and theoretical computer science, sometimes under different names such as Markov random fields.
The most important computation task in their study is to estimate the so-called partition function,
which is a central quantity linking to many of the system's macroscopic properties.

A particularly important case is when the number of spin is 2. 
These two-state spin systems (or 2-spin systems for short) show drastically different behaviours depending on whether the interaction is repulsive (anti-ferromagnetic) or attractive (ferromagnetic).
For anti-ferromagnetic systems, great progress has been made regarding partition function estimation --
it is \NP-hard beyond the tree-uniqueness threshold \cite{Sly10,SS14,GSV16}, and has an efficient algorithm otherwise \cite{weitz2006counting,SST14,LLY13,CLV23,CCYZ25}. 
The complexity landscape for the anti-ferromagnetic case is completely mapped out.

On the other hand, ferromagnetic systems have demonstrated a more involved complexity landscape.
Here the most well-known system is the Ising model, and the pioneer work of Jerrum and Sinclair \cite{JerrumS93} showed that efficient algorithms for the ferromagnetic Ising model exist regardless of the tree-uniqueness threshold.
Their approach is direct Markov chain Monte Carlo (MCMC), and this approach was taken further by Goldberg, Jerrum, and Paterson \cite{GoldbergJP03} for general ferromagnetic 2-spin systems under certain conditions.
Later, Guo and Lu \cite{GuoL18} applied Weitz's correlation decay method \cite{weitz2006counting} to push the algorithmic threshold even further.
In certain regimes, their algorithm is almost optimal, up to a threshold beyond which the problem becomes \BIS{}-hard \cite{LiuLZ14}.
Here \BIS{} stands for the problem of counting independent sets in bipartite graphs, which is conjectured to be intractable to approximate \cite{DGGJ04}.
However, Guo and Lu's algorithm does not extend to the case where both spins are attracting.
This is not a barrier to algorithms though, as subsequently Guo, Liu, and Lu \cite{GLL20} successfully adapted Barvinok's method \cite{Bar16,PR17} to this setting.
These works \cite{JerrumS93,GoldbergJP03,GuoL18,GLL20} constitute the current algorithmic frontier for ferromagnetic 2-spin systems,
but there is still a large gap between known algorithmic and hardness thresholds.

A main drawback of the correlation decay method and Barvinok's method is that the running time of the resulting algorithm is often a polynomial with a large exponent. 
In particular, the exponent for the algorithm in \cite{GLL20} scales logarithmically in the maximum degree, and it runs within polynomial-time only for bounded degree graphs.
In this paper, we address this issue by giving a new algorithm that runs in near-quadratic time for bounded degree graphs, and in polynomial-time when the degree is unbounded.
We note that a similar speedup for \cite{GuoL18} was obtained earlier this year \cite{feng2026rapid},
which relies on a new analysis for direct MCMC.
In contrast, our result is quite different.
We obtain an equivalent (weighted subgraph) formulation of the problem first, and then apply MCMC to the new setting. 
Our new method applies to the regime in which both spin types are attractive, a setting where the techniques of \cite{GuoL18,feng2026rapid} encounter some fundamental barriers.

Another drawback of \cite{GLL20} is that due to lack of self-reducibility, no efficient sampling algorithm is known.
We also address this via introducing an intermediate random cluster type model so that we can transform a sample from the weighted subgraph model to a sample of spin systems.

Next, we formally define the problems and state our result in more detail.

\subsection{Sampling and counting for ferromagnetic two-spin systems}

Let $G=(V,E)$ be a graph with $n=|V|$ vertices and $m=|E|$ edges. A two-spin system on $G$ is defined by an interaction matrix $A$ and an external field vector $B$ as follows:
\[
    A = \begin{pmatrix}
        \beta & 1 \\
        1 & \gamma
    \end{pmatrix}, \qquad
    B = \begin{pmatrix}
        \lambda \\
        1
    \end{pmatrix}.
\]
For any configuration $\sigma: V \to \{0,1\}$, its weight is defined by
\begin{align*}
  w^{\text{spin}}_{G}(\sigma; \beta, \gamma, \lambda) = \prod_{(u,v) \in E} A_{\sigma_u, \sigma_v} \prod_{v\in V} B_{\sigma_v} = \beta^{m_0(\sigma)} \gamma^{m_1(\sigma)} \lambda^{n_0(\sigma)},
\end{align*}
where $m_0(\sigma)$ and $m_1(\sigma)$ are the numbers of edges with both endpoints in state $0$ and in state $1$, respectively, and $n_0(\sigma)$ is the number of vertices in state $0$.
The associated Gibbs distribution is given by $\mu^{\text{spin}}(\sigma) \propto w^{\text{spin}}_{G}(\sigma; \beta, \gamma, \lambda)$.
The partition function, namely the normalising factor for the Gibbs distribution, is
\begin{align}  \label{eqn:Z-spin}
  Z^{\text{spin}}_G(\beta, \gamma, \lambda) \defeq \sum_{\sigma: V \to \{0,1\}} w^{\text{spin}}_{G}(\sigma; \beta, \gamma, \lambda).
\end{align}
By symmetry between the two states, we will assume without loss of generality throughout the paper that $\beta \leq \gamma$.
The 2-spin system is said to be ferromagnetic if $\beta \gamma > 1$ (and anti-ferromagnetic if $\beta\gamma<1$).

Consider a ferromagnetic two-spin system with parameters $\beta$, $\gamma$, and $\lambda$ such that $\beta\le\gamma$. 
Goldberg, Jerrum, and Paterson \cite{GoldbergJP03} first showed that a fully polynomial-time randomised approximation scheme (FPRAS) for $Z^{\text{spin}}_G$ exists if $\lambda\le\frac{\gamma}{\beta}$, by extending the MCMC approach of Jerrum and Sinclair \cite{JerrumS93}.
Later, Guo and Lu \cite{GuoL18} identified a threshold $\lambda_c=\lambda_c(\beta,\gamma)\defeq\left(\frac{\gamma}{\beta}\right)^{\frac{\sqrt{\beta\gamma}}{\sqrt{\beta\gamma}-1}}$,
and introduced an fully polynomial-time approximation scheme (FPTAS)\footnote{Note that the difference between FPTAS and FPRAS is that FPTAS is deterministic.} if $\lambda<\lambda_c$ with the further assumption of $\beta\le 1$,
based on establishing strong spatial mixing (SSM) and the correlation decay method \cite{weitz2006counting}.
This threshold is almost optimal as Liu, Lu, and Zhang \cite{LiuLZ14} showed that if $\lambda>\left(\frac{\gamma}{\beta}\right)^{\left\lfloor\frac{2\sqrt{\beta\gamma}}{\sqrt{\beta\gamma}-1}\right\rfloor/2+1}$, the problem becomes \BIS{}-hard,
where approximating \BIS{} is conjectured to be intractable \cite{DGGJ04}.

On the other hand, the regime where $\beta>1$ is still poorly understood.
Here SSM does not hold any more for any constant $\lambda>0$.\footnote{Roughly speaking, this is because there is no degree bound. In fact, even the so-called weak spatial mixing does not hold. One can verify this by considering a regular tree of sufficiently large degrees.} 
Thus the approach of \cite{GuoL18} (as well as the recent improvement of \cite{feng2026rapid}) cannot be extended to this case.
However, the lack of correlation decay does not block efficient algorithms, which makes the picture much less clear.
Guo, Liu, and Lu \cite{GLL20} used Barvinok's method \cite{Bar16} to identify yet another threshold:
\begin{align}\label{eq:lambda-star}
    \lambda^\star = \lambda^\star(\beta, \gamma) \defeq \left(\frac{\gamma}{\beta}\right)^{\frac{\pi}{2\tan^{-1}\sqrt{\beta\gamma-1}}},
\end{align}
and introduced an FPTAS for bounded degree graphs if $\lambda<\lambda^{\star}$.\footnote{It can be verified that $\lambda^{\star}<\lambda_c$ \cite{GLL20}.}

Among all ferromagnetic 2-spin systems, the most well-understood special case is the Ising model (namely $\beta=\gamma > 1$).
There are rich structures for ferromagnetic Ising models, leading to various equivalent formulations such as the even-subgraph model \cite{Wae41} (also known as high-temperature expansion) and the random cluster model \cite{FK72}.
Using these connections, polynomial-time sampling and approximate counting algorithms are obtained \cite{JerrumS93,GuoJ18}.
Furthermore, in the presence of external field $\lambda \neq 1$, there are a near-linear time sampler and a near-quadratic time FPRAS \cite{CZ23}.

Part of the reason that the parameter regime $\gamma > \beta > 1$ is not very well-understood is the lack of these equivalent formulations.
Our main conceptual contribution is to fill in this gap.
We introduced new weighted subgraph models and random cluster type models, and generalise previous equivalences to the whole $\gamma > \beta > 1$ regime.
These new connections help us find fast approximate counting and sampling algorithms.
We believe they will be useful in fully resolving the approximate counting and sampling complexity of ferromagnetic $2$-spin systems in the future.

Our first result is a fast sampling algorithm for $\lambda<\lambda^{\star}$.
In this regime, no efficient sampling algorithm is previously known.
The FPTAS of \cite{GLL20} cannot be transformed to a sampler due to lack of self-reducibility (basically, by the same reason why SSM fails).

\begin{theorem}  \label{thm:sampling}
  Let $\beta,\gamma,\lambda > 0$ be parameters such that $\gamma > \beta > 1$ and $\lambda < \lambda^\star$.
  There exists an algorithm that, given any $\epsilon > 0$ and any graph $G=(V,E)$, outputs a random sample $\sigma\in\{0,1\}^V$ in time $O\left(\Delta^C n \log \frac{n}{\epsilon}\right)$ that is at most $\epsilon$ away from $\mu^{\text{spin}}$ in total variation distance, where $n = |V|$ is the number of vertices, $\Delta$ is the maximum degree of the graph, and $C = C(\beta, \gamma, \lambda)$ is a constant.
\end{theorem}

To achieve \Cref{thm:sampling}, we use Markov chains. 
However, directly MCMC is not efficient here.
To see this, for any constants $\beta>1$, $\gamma>1$, and $\lambda>0$,
consider a sufficiently large complete graph.
It is easy to see that the all-$1$ and all-$0$ configurations contribute weights exponentially larger than mixed configurations,
and therefore the mixing time for any locally defined Markov chains such as the Glauber dynamics is exponentially large.
This is also why the recent work of \cite{feng2026rapid}, which is a direct MCMC, is restricted to the case of $\beta\le1<\gamma$.

Our solution to bypass this issue is via the aforementioned new connections.
We consider an equivalent weighted subgraph formulation, which we will explain in \Cref{sec:Holant} shortly.
We show that Glauber dynamics for the weighted subgraphs mixes within the time bound of \Cref{thm:sampling}.
However, there appears to be no easy way to transform this subgraph sample to a spin sample.
We instead introduce an intermediate random cluster type model.
For the three models, we not only establish equivalence among their partition functions, but also show new couplings among their Gibbs distributions.
With the couplings, we first transform the subgraph sample to a random cluster sample, and then to a spin sample, to establish \Cref{thm:sampling}.
Details can be found in \Cref{sec:RC}.

Using the sampler above and a simple $\lambda$-annealing reduction, we also obtain the following FPRAS.

\begin{theorem}\label{thm:fpras-spin}
Let $\beta,\gamma,\lambda > 0$ be parameters such that $\gamma > \beta > 1$ and $\lambda < \lambda^\star$.
There exists an FPRAS that, given any $\epsilon > 0$ and any graph $G=(V,E)$, outputs a random quantity $\widehat{Z}$ in time $O\left(\Delta^C \cdot \frac{n^2}{\epsilon^2} \cdot \log^2 \frac{n}{\epsilon}\right)$ such that $\widehat{Z} \in [Z(1-\epsilon), Z(1+\epsilon)]$ with probability at least $\frac{2}{3}$, where $Z = Z^{\text{spin}}_G(\beta, \gamma, \lambda)$ is the partition function, $n = |V|$ is the number of vertices, $\Delta$ is the maximum degree of the graph, and $C = C(\beta, \gamma, \lambda)$ is a constant.
\end{theorem}

The algorithm in \cite{GLL20} applies to parameters of the same regime as in \Cref{thm:fpras-spin}, 
but it runs in time $O\left(\left(\frac{n}{\eps}\right)^{O(\log\Delta)}\right)$, which is no longer a polynomial if the graph has unbounded degrees.
This type of running time is typical for FPTASes \cite{weitz2006counting,Bar16,PR17}.
Utilising randomness, the algorithm in \Cref{thm:fpras-spin} runs in polynomial-time for any graph, and in near-quadratic time for bounded degree graphs.
We note that near-quadratic time is the best known running time for partition function estimation in many settings \cite{SVV09,Kol18},
although this barrier has been broken very recently when the parameters are lower than the computational transition threshold \cite{AFFGW25,CCLZ26}.
Their techniques do not apply in our setting, and it is not clear if sub-quadratic time can be achieved here.

\Cref{thm:sampling} and \Cref{thm:fpras-spin} cover only the case of $\gamma > \beta > 1$.
When $\beta\le 1\le \gamma$ but $\beta\gamma>1$,
we refer the interested reader to a recent work by Feng, Guo, and Yang~\cite{feng2026rapid}.
They showed an $\widetilde{O}(n^2)$ mixing time bound for Glauber dynamics in the spin system from the all $1$ configuration, when $\lambda<\lambda_c$,
and an $\widetilde{O}(n)$ mixing time bound when $\lambda<\sqrt{\gamma/\beta}$.
For the Ising model ($\beta = \gamma > 1$), there are the classical sampling algorithm for $\lambda=1$ by Jerrum and Sinclair~\cite{JerrumS93} and the near-linear time sampler for $\lambda \neq 1$ by Chen and Zhang~\cite{CZ23}.
Approximate counting algorithms in all these settings can be obtained similarly via simulated annealing.

Since the weighted subgraph model is the key to our whole argument,
we explain it next.

\subsection{The subgraph world}
\label{sec:Holant}

We introduce a new weighted subgraph (WSG) model and relate its partition function to that of the two-spin system.
This is a generalisation to the classical subgraph world models \cite{JerrumS93}. 
Let $G = (V,E)$ be a graph. For any subset of edges $S \subseteq E$ and any vertex $v$, let $\deg_S(v)$ denote the number of edges in $S$ incident to $v$; in other words, $\deg_S(v)$ is the degree of $v$ in the subgraph $(V,S)$. 
The model involves a number of parameters. 
Each edge $e \in E$ is associated with a parameter $\theta_e > 0$,
and each vertex $v \in V$ is associated with $\lambda_v > 0$.
Denote $\boldsymbol{\theta} = (\theta_e)_{e \in E}$ and $\boldsymbol{\lambda} = (\lambda_v)_{v \in V}$.
Let $\rho > 0$ be another parameter.
In our applications later, we will have $\rho \lambda_v < 1$ for all $v \in V$.
For any subset of edges $S \subseteq E$, the weight of $S$ is given by
\begin{align}\label{eq:weight-wsg}
w^{\text{wsg}}_{G}(S; \boldsymbol{\theta}, \boldsymbol{\lambda}, \rho) = \prod_{e \in S} \theta_e \prod_{v \in V} (1 + \lambda_v (-\rho)^{\deg_S(v)}).
\end{align}
The partition function of the WSG model is then given by
\begin{align*}
    Z^{\text{wsg}}_G(\boldsymbol{\theta}, \boldsymbol{\lambda}, \rho) = \sum_{S \subseteq E} w^{\text{wsg}}_{G}(S; \boldsymbol{\theta}, \boldsymbol{\lambda}, \rho).
\end{align*}
Note that in the definitions above, we allow non-uniform weights for later convenience.

The following equivalence is shown via the holographic transformation~\cite{Valiant08} in~\Cref{sec:WSG}. 
\begin{theorem}\label{thm:equivalence}
  Consider a two-spin system on a graph $G = (V,E)$ with parameters $\beta, \gamma, \lambda$ such that $\gamma\neq 1$, $\beta+\gamma\neq 2$, and $\beta \gamma \neq 1$. %Let $Z_{\text{spin}}=Z^{\text{spin}}_G(\beta, \gamma, \lambda)$ be the partition function of the two-spin system. 
Let $\rho = \frac{\beta-1}{\gamma-1}$ and $\theta = \frac{(\gamma-1)^2}{\beta\gamma-1}$, and let $\boldsymbol{1}_V = (1)_{v \in V}$ and $\boldsymbol{1}_E = (1)_{e \in E}$ denote the $V$-dimensional and $E$-dimensional all-ones vectors, respectively. Then
\begin{align*}
    Z^{\text{spin}}_G(\beta, \gamma, \lambda) = (1 - \rho \theta)^{-|E|} \cdot Z^{\text{wsg}}_G(\theta \boldsymbol{1}_E, \lambda \boldsymbol{1}_V, \rho). 
\end{align*}
\end{theorem}

\Cref{thm:equivalence} applies to a wide range of parameters due to its algebraic nature.
It is easy to check that the conditions of \Cref{thm:equivalence} are satisfied in our application for $\gamma>\beta>1$ and $\beta \gamma > 1$.

\begin{remark}[Comparison with existing subgraph-world models]
It is well known that the ferromagnetic Ising model ($\beta=\gamma$) is equivalent to certain subgraph-world models~\cite{JerrumS93,FGW23}. Our result in \Cref{thm:equivalence} generalises this equivalence to non-Ising ferromagnetic two-spin systems with $\beta < \gamma$. The Ising model without fields, namely when $\beta=\gamma$ and $\lambda=1$, is equivalent to the \emph{even-subgraph} model, where the subgraph $(V,S)$ must satisfy that $\deg_S(v)$ is even for every $v \in V$. In this case, our parameter $\rho=1$, and only even subgraphs have positive weights in~\eqref{eq:weight-wsg}.
When the external field is not $1$, say $\lambda<1$, the Ising model is equivalent to a subgraph world in which each odd-degree vertex incurs a penalty; this is also captured by~\eqref{eq:weight-wsg} with $\rho=1$. For general ferromagnetic two-spin systems with $\beta < \gamma$, our model imposes a penalty based on the degree of the vertex in the subgraph. 
We believe this new subgraph-world model is of independent interest.
\end{remark}

The main ingredient of our algorithm is an efficient sampling algorithm for the Gibbs distribution of the WSG model.
%By the standard counting-to-sampling reduction~\cite{JVV86,SVV09}, this gives an FPRAS for the partition function of the WSG model, which, by \Cref{thm:equivalence}, also yields an FPRAS for the partition function of the 2-spin system. 
%
Consider a WSG model on $G=(V,E)$ with parameters $\boldsymbol{\theta}, \boldsymbol{\lambda}, \rho$ that correspond to a ferromagnetic 2-spin system.
When $\lambda<\lambda^{\star}$, the weight function $w^{\text{wsg}}_{G}(S; \boldsymbol{\theta}, \boldsymbol{\lambda}, \rho)$  in~\eqref{eq:weight-wsg} is positive for all $S \subseteq E$ (see \Cref{lem:lambda-v-tau}).
Let $\mu = \mu^{\text{wsg}}_{G,\boldsymbol{\theta}, \boldsymbol{\lambda}, \rho}$ over $2^E$ be the Gibbs distribution induced by the weighted subgraph model. For any subset of edges $S \subseteq E$, the probability of $S$ is given by
\begin{align}\label{eqn:wsg-dist}
  \mu(S) = \frac{w^{\text{wsg}}_{G}(S; \boldsymbol{\theta}, \boldsymbol{\lambda}, \rho)}{Z^{\text{wsg}}_G(\boldsymbol{\theta}, \boldsymbol{\lambda}, \rho)}.
\end{align}
We often view $\mu$ as a distribution over $E$-dimensional Boolean vectors. For any $\sigma \in \{0, 1\}^E$, let $S_\sigma \subseteq E$ denote the subset of edges corresponding to $\sigma$. Then we write 
\begin{align*}
\forall \sigma \in \{0, 1\}^E, \quad \mu(\sigma) = \mu(S_\sigma) = \frac{w^{\text{wsg}}_{G}(S_\sigma; \boldsymbol{\theta}, \boldsymbol{\lambda}, \rho)}{Z^{\text{wsg}}_G(\boldsymbol{\theta}, \boldsymbol{\lambda}, \rho)}.
\end{align*}

To generate a random sample from the distribution $\mu$, we employ the standard Glauber dynamics. Starting from an arbitrary initial configuration $X_0 \in \{0, 1\}^E$, in each step $t\ge 0$ we update $X_t$ as follows:
\begin{itemize}
    \item sample an edge $e \in E$ uniformly at random;
    \item sample $X_{e} \sim \mu_e^{X_{E - e}}$, where $\mu_e^{X_{E - e}}$ is the conditional distribution on $e$ given $X_{E - e}$, and update $X_t$ by changing the state of $e$ to $X_e$ to get $X_{t+1}$.
\end{itemize}
It is easy to simulate Glauber dynamics: since the weight in~\eqref{eq:weight-wsg} is a product of terms, each of which depends only on a single edge or a single vertex, the conditional distribution $\mu_e^{X_{E - e}}$ can be computed in time $O(\Delta)$, where $\Delta$ is the maximum degree of the graph.

It is well known that the stationary distribution of the Glauber dynamics is $\mu$. 
To analyse the efficiency of Glauber dynamics, we study its \emph{mixing time}.
Let $(X_t)_{t \geq 0}$ be the random configurations generated by Glauber dynamics over time.
Given any $\epsilon > 0$, the mixing time is the number of steps required for the chain to reach its stationary distribution up to total variation distance $\epsilon$. Formally,
\begin{align*}
    T_{\text{mix}}(\epsilon) = \max_{X_0 \in \{0, 1\}^E} \min \{t \geq 0 \mid \DTV{X_t}{\mu} < \epsilon\},
\end{align*}
where $\DTV{X_t}{\mu} = \frac{1}{2}\sum_{\sigma \in \{0, 1\}^E} |\Pr[X_t = \sigma] - \mu(\sigma)|$ is the \emph{total variation distance} between (the law of) $X_t$ and $\mu$.
Our main technical contribution is the following mixing time bound for the WSG model.
Recall that $\lambda^\star  = \lambda^\star(\beta, \gamma) = \left(\frac{\gamma}{\beta}\right)^{\frac{\pi}{2\tan^{-1}\sqrt{\beta\gamma-1}}}$ is the threshold in~\eqref{eq:lambda-star}.
\begin{theorem}\label{thm:mixing-time-sgw}
Let $\beta, \gamma, \lambda > 0$ be constants such that $\gamma > \beta > 1$ and $\lambda < \lambda^\star$.
Let $\theta = \frac{(\gamma-1)^2}{\beta\gamma-1}$ and $\rho = \frac{\beta-1}{\gamma-1}$.
Then, the weight function $w^{\text{wsg}}_{G}(S; \theta \boldsymbol{1}_E, \lambda \boldsymbol{1}_V, \rho) \geq 0$ for all $S \subseteq E$.
Let $\mu$ denote the Gibbs distribution induced by the weighted subgraph model on a graph $G = (V,E)$ with parameters $\theta \boldsymbol{1}_E$, $\lambda \boldsymbol{1}_V$, and $\rho$. Then the mixing time of the Glauber dynamics on $\mu$ is at most $O\left( \Delta^C \cdot n \log \frac{n}{\epsilon} \right)$, where $n = |V|$, $\Delta$ is the maximum degree of the graph, and $C = C(\beta, \gamma, \lambda)$ is a constant.
\end{theorem}

For comparison, in the Ising model, where $\beta=\gamma>1$, the analogous threshold is $\lambda^{\star}=1$.
When $\lambda<1$, Chen and Zhang \cite{CZ23} proved a fast mixing time bound for the WSG dynamics induced by the Ising model.
However, our method is based on stability of polynomials,
which is completely different from the coupling-based method in \cite{CZ23}.
Their coupling crucially relies on the fact that, in the weighted subgraph model induced by the Ising model, every vertex of odd degree receives the same \emph{uniform} penalty.
In contrast, in our setting with $\beta < \gamma$, vertices of different degrees receive different penalties in~\eqref{eq:weight-wsg}.
Chen and Gu~\cite{CG24} developed a more general coupling-based technique, but our setting does not meet their requirements either.

%On the other hand, for the Ising model with no field, namely $\beta=\gamma>1$ and $\lambda=1$,
%the subgraph world allows only even subgraphs, and Glauber dynamics is not even ergodic.
%Thus we cannot strengthen \Cref{thm:mixing-time-sgw} to allow $\lambda=\lambda^{\star}$.

%Note that our results only cover the case of $\gamma > \beta > 1$.
%When $\beta\le 1\le \gamma$ but $\beta\gamma>1$,
%we refer the interested reader to a recent work by Feng, Guo, and Yang~\cite{feng2026rapid}.
%They showed an $\widetilde{O}(n^2)$ mixing time bound for Glauber dynamics in the spin system from the all $1$ configuration, when $\lambda<\lambda_c$,
%and an $\widetilde{O}(n)$ mixing time bound when $\lambda<\sqrt{\gamma/\beta}$.
%For the Ising model when $\beta = \gamma > 1$, one can refer the classical sampling algorithm by Jerrum and Sinclair~\cite{JerrumS93} or near-linear time algorithm when $\lambda < 1$ by Chen and Zhang~\cite{CZ23}.

\section{Proof outline of the mixing time}

%\subsection{Weighted subgraph-world model}

%We now relate the partition function of the two spin system to the partition function of the weighted subgraph-world model. Consider a slightly generalized weighted ferromagnetic two spin system with interaction matrix $A_e$ for each edge $e \in E$ and external field vector $b_v$ for each vertex $v \in V$ as follows:
%\[
%    A_e = \begin{pmatrix}
%        \beta_e & 1 \\
 %       1 & \gamma_e
%    \end{pmatrix}, \qquad
%    b_v = \begin{pmatrix}
%        \lambda_v \\
%        1
%    \end{pmatrix}.
%\]
%The partition function and the weight of the configuration can be naturally generalized as follows:
%\begin{align*}
%Z_{\text{spin}} = \sum_{\sigma: V \to \{0,1\}} w_{\text{spin}}(\sigma) = \prod_{e = \{u,v\} \in E: \sigma_u = \sigma_v = 0} \beta_e \prod_{e = \{u,v\} \in E: \sigma_u = \sigma_v = 1} \gamma_e \prod_{v \in V: \sigma_v = 0} \lambda_v.
%\end{align*}

%\subsection{Counting-to-sampling reduction}

In this section we outline the proof of \Cref{thm:mixing-time-sgw}. Under the parameter assumption in the theorem, the positivity of the weight function is guaranteed can be verified by a straightforward calculation, which is deferred to \Cref{sec:validify-gibbs}. Assuming the distribution $\mu$ is well-defined, we outline the proof of the mixing time bound.

We use the following notations.
For any subset of edges $\Lambda \subseteq E$ and any pinning $\tau \in \{0,1\}^{E\setminus\Lambda}$, we write $\mu^{\tau}$ for the conditional distribution of $\mu$ given $\tau$. In other words, in $\mu^\tau$, the edges in $E \setminus \Lambda$ are pinned by $\tau$ and the edges in $\Lambda$ follow the conditional distribution.
Moreover, we write $\mu_{S}$ for the marginal distribution of $S\subseteq E$ under $\mu$.
When $S=\{e\}$, we may also write $\mu_e$ instead.
In particular, $\mu^{\tau}_e$ means the marginal distribution of $e$ under the conditional distribution $\mu^{\tau}$.

\subsection{Mixing of Glauber dynamics via spectral independence}

We prove the mixing time bound via the spectral independence technique introduced by Anari, Liu, and Oveis Gharan~\cite{ALO24}. Let $\mu$ be a distribution over $\{0,1\}^E$. For any two variables $i,j \in E$, the influence from $i$ to $j$ is defined as 
\begin{align*}
\Psi_\mu(i,j) = \Pr_{X \sim \mu}[X_j = 1 \mid X_i = 1] - \Pr_{X \sim \mu}[X_j = 1 \mid X_i = 0].
\end{align*}

\begin{definition}[spectral independence~\text{\cite{ALO24}}]
The distribution $\mu$ is said to be $\eta$-spectrally independent if the maximum eigenvalue of the influence matrix $\Psi_\mu$ is at most $\eta$.
\end{definition}

\begin{definition}[marginal lower bound~\text{\cite{CLV21}}]
A distribution $\mu$ over $\{0,1\}^E$ has a marginal lower bound $b > 0$ if, for any subset of edges $\Lambda \subseteq E$, any pinning $\tau \in \{0,1\}^{E \setminus \Lambda}$, and any edge $e \in \Lambda$, 
\begin{align*}
\forall c \in \{0,1\}, \quad \mu_e^\tau(c) = \Pr_{X \sim \mu}[X_e = c \mid X_{E \setminus \Lambda} = \tau] \geq b.
\end{align*}
\end{definition}

Using the definition of the weighted subgraph model in~\eqref{eq:weight-wsg}, the following lemma can be established by a straightforward calculation. The proof is deferred to \Cref{sec:lb}.

\begin{lemma}\label{lem:marginal-lower-bound}
    Let $\theta, \rho, \lambda > 0$ be three constants such that $\rho < 1$ and $\lambda \rho < 1$. The weighted subgraph model on any graph $G$ with parameters $\theta \boldsymbol{1}_E$, $\lambda \boldsymbol{1}_V$, and $\rho$ has a marginal lower bound $b = b(\theta, \lambda, \rho) > 0$.
\end{lemma}

In particular, Chen, Liu, and Vigoda \cite[Theorem 1.12]{CLV21} showed the following.
For a Gibbs distribution $\mu$ on a graph $G=(V,E)$ with maximum degree $\Delta$, if every conditional distribution induced by $\mu$ is $\eta$-spectrally independent, and $\mu$ has a constant marginal lower bound $b$, then the mixing time of the Glauber dynamics is bounded by $O\left(\Delta^C \cdot n \cdot (\log\log \frac{1}{\mu_{\min}}+ \log \frac{1}{\epsilon})\right)$ for some $C>0$ depending on~$\eta$ and~$b$, where $\mu_{\min} = \min_{\sigma \in \{0, 1\}^E} \mu(\sigma)$. 

We summarise and specialise this result for the WSG model in the next lemma.
In particular, the constant marginal lower bound is guaranteed by \Cref{lem:marginal-lower-bound}, which is independent from $\Delta$.
\begin{lemma}\label{lem:spectral-independence-to-mixing}
Let $\beta, \gamma, \lambda > 0$ be constants such that $\gamma > \beta > 1$ and $\lambda < \lambda^\star$.
Let $\theta = \frac{(\gamma-1)^2}{\beta\gamma-1}$ and $\rho = \frac{\beta-1}{\gamma-1}$.
Let $\mu$ denote the Gibbs distribution induced by the WSG model on a graph $G = (V,E)$ with parameters $\theta \boldsymbol{1}_E$, $\lambda \boldsymbol{1}_V$, and $\rho$.
Suppose there exists a constant $\eta > 0$ such that for any $\Lambda \subseteq E$ and any $\tau \in \{0,1\}^{E \setminus \Lambda}$, the conditional distribution $\mu^\tau$ is $\eta$-spectrally independent. Then the mixing time of the Glauber dynamics on $\mu$ is at most $O\left(\Delta^C \cdot n \log \frac{n}{\epsilon}\right)$, where $C = C(\beta, \gamma, \lambda,\eta)$ is a constant.
\end{lemma}

\subsection{Spectral independence via the stability}

By \Cref{lem:spectral-independence-to-mixing}, it remains to show that every conditional distribution $\mu^\tau$ is $\eta$-spectrally independent, for every subset $\Lambda \subseteq E$ and every pinning $\tau \in \{0,1\}^{E \setminus \Lambda}$. 
We show this by showing stability on the complex plane for a suitable polynomial, and then applying the techniques from \cite{CLV24}.

For any configuration $\sigma \in \{0,1\}^E$, let $S^\sigma$ denote the subset of edges selected by $\sigma$:
$S^\sigma = \{e \in E \mid \sigma_e = 1\}$.
By the definition of the weighted subgraph model in~\eqref{eq:weight-wsg}, the weight of the configuration $\sigma$ is
$
    \prod_{e \in S^\sigma} \theta_e
    \prod_{v \in V}\tp{1 + \lambda_v(-\rho)^{\deg_{S^\sigma}(v)}}.
$
%To show this, we need to define conditional partition function and write it as a polynomial. 
%Let $S^\tau$ denote
%\begin{align*}
%S^\tau = \{e \in E \setminus \Lambda \mid \tau_e = 1\}.
%\end{align*}
%For any subset of edges $S \subseteq E$ such that $S_\tau \subseteq S$, recall that the weight of the configuration $S$ in the weighted subgraph-world model is given by in \eqref{eq:weight-wsg} as $\prod_{e \in S \cup S^\tau} \theta_e \prod_{v \in V}\tp{1 + \lambda_v(-\rho)^{\deg_{S \cup S^\tau}(v)}}$.
We separate out the vertex contribution and define
\begin{align}\label{eq:w-weight-function}
\forall  \sigma \in \{0,1\}^E, \quad
w(\sigma) \defeq \prod_{v \in V}\tp{1 + \lambda_v(-\rho)^{\deg_{S^\sigma}(v)}}.
\end{align}
Given a pinning $\tau \in \{0,1\}^{E \setminus \Lambda}$, we treat the edge activities on the unpinned edges as complex variables $z_e$ and define the conditional partition function polynomial
\begin{align}\label{eq:conditional-partition-function-polynomial}
Z^\tau(\boldsymbol{z})
= \sum_{\substack{\sigma \in \{0,1\}^E: \sigma_{E \setminus \Lambda} = \tau}}
w(\sigma) \prod_{\substack{e \in \Lambda: \sigma_e = 1}} z_e.
\end{align}
The product only ranges over edges in $\Lambda$, since the edges in $E\setminus \Lambda$ have already been fixed by $\tau$.
Thus $Z^\tau(\boldsymbol z)$ is a polynomial in the variables $z_e$, $e\in \Lambda$. 
When the pinning is empty, we write $Z(\boldsymbol z)$ for $Z^\emptyset(\boldsymbol z)$.
Substituting $z_e=\theta_e$ recovers the conditional distribution on the unpinned variables:
%To emphasize the dependence on $w$ and $\boldsymbol{\theta}$, we write this distribution as $\mu^\tau_{w,\boldsymbol{\theta}}$:\HG{What's the point of the subscript here? It contradicts the marginal distribution notation.}
\begin{align*}
\forall \sigma \in \{0,1\}^E \text{ with } \sigma_{E\setminus\Lambda}=\tau,
\quad
\mu^\tau(\sigma)
\defeq
\frac{w(\sigma) \prod_{\substack{e \in \Lambda: \sigma_e = 1}} \theta_e}{Z^\tau(\boldsymbol{\theta})}.
\end{align*}
Again, when the pinning is empty, $\mu^\emptyset=\mu$.

The polynomial in~\eqref{eq:conditional-partition-function-polynomial} allows us to complexify the edge activities and prove spectral independence through polynomial stability, following~\cite{CLV24}. 
Throughout the paper, $\mathbb{C}$ denotes the complex plane.
%The following theorem states that the spectral independence of the conditional distribution $\mu^\tau_{\Lambda}$ can be established by the stability of the polynomial in~\eqref{eq:conditional-partition-function-polynomial}.

\begin{definition}[stability of polynomial]\label{def:stability-of-polynomial}
For regions $\Gamma_i\subseteq \mathbb{C}$, $i\in [n]$, a multivariate polynomial $p\in \mathbb{C}[x_1, \ldots, x_n]$ is $\left(\prod_{i=1}^{n}\Gamma_i\right)$-stable if $p(x_1, \ldots, x_n) \neq 0$ whenever $x_i \in \Gamma_i$ for all $i\in [n]$.
\end{definition}

We specialise \cite[Theorem 3.2]{CLV24} to our setting as follows.

\begin{proposition}\label{prop:stability-to-spectral-independence}
Let $E$ be a set of variables and $w: \{0,1\}^E \to \mathbb{R}$ be a weight function.
For any $e\in E$, let $\Gamma_e\subset \mathbb{C}$ be a non-empty open connected region, and let $\boldsymbol{\theta}:E\to \mathbb{R}$ satisfy $\theta_e\in \mathbb{R}_+\cap \Gamma_e$ for every $e\in E$. 
Suppose that, for every pinning $\tau$, the conditional partition function $Z^\tau(\boldsymbol{z})$ is $\left(\prod_{e\in \Lambda} \Gamma_e\right)$-stable. %\footnote{$Z^\tau$ is essentially a polynomial in the unpinned variables $z_e$, $e\in \Lambda$.}
Then the distribution $\mu^{\tau}$ is spectrally independent with constant
\begin{align*}
  \eta= \frac{2}{b\delta^2},
\end{align*}
where $b$ is the marginal lower bound of $\mu^{\tau}$, $\delta= \min_{e\in E} \frac{1}{\theta_e}\operatorname{dist}(\theta_e,\partial \Gamma_e)$,
$\partial \Gamma_e$ is the boundary of $\Gamma_e$, and $\operatorname{dist}(x, \partial \Gamma_e) \defeq \inf_{z \in \partial \Gamma_e} |x-z|$.
\end{proposition}

We apply \Cref{prop:stability-to-spectral-independence} to the WSG model. 
Once again, recall that our setting is $\gamma > \beta > 1$, $\rho = \frac{\beta-1}{\gamma-1}$, $\theta = \frac{(\gamma-1)^2}{\beta\gamma-1}$, and $\lambda < \lambda^\star(\beta,\gamma)$.
Let $\mu$ be the Gibbs distribution induced by the WSG model on a graph $G = (V,E)$ with parameters $\theta \boldsymbol{1}_E$, $\lambda \boldsymbol{1}_V$, and $\rho$. 
We will use \Cref{prop:stability-to-spectral-independence} to prove spectral independence for every conditional distribution $\mu^\tau$.

The following self-reduction lets us handle all pinnings uniformly. Fix a pinning $\tau \in \{0,1\}^{E \setminus \Lambda}$. Let $G^\tau=(V,\Lambda)$ be the subgraph obtained by deleting the pinned edges in $E\setminus \Lambda$. For each vertex $v\in V$, let $k_v$ be the number of pinned edges incident to $v$ that are fixed to $1$ by $\tau$. We absorb the contribution of these pinned edges into the external field by setting
\begin{align*}
\lambda_v^\tau = \lambda \cdot \left(-\rho\right)^{k_v}.
\end{align*}
The following technical lemma is verified in \Cref{sec:validify-gibbs}.
\begin{lemma}\label{lem:lambda-v-tau}
For any integer $k \geq 0$, it holds that
$ -1 < -\rho \lambda^\star \leq \lambda \cdot \left(-\rho\right)^k \leq \lambda.
$
\end{lemma}
Given the pinning $\tau$, for every $S\subseteq \Lambda$, the conditional weight of the configuration $S$ is, up to the non-zero contribution of the pinned edges, equal to
\[
    \theta^{|S|}
    \prod_{v\in V}\left(1+\lambda_v^\tau(-\rho)^{\deg_S(v)}\right).
\]
By \Cref{lem:lambda-v-tau}, the above weight is still non-negative.
Hence the conditional distribution $\mu^\tau$ projected onto $\Lambda$ is the Gibbs distribution of a weighted subgraph model on the subgraph $G^\tau$ with parameters $\theta \boldsymbol{1}_\Lambda$, $(\lambda_v^\tau)_{v \in V}$, and $\rho$. We will use this self-reduction to handle all the conditional distributions.

It remains to establish the required stability condition in \Cref{prop:stability-to-spectral-independence}. For $x \in \mathbb{C}$ and $\delta>0$, write $\mathcal{U}_\delta(x) \defeq \{z \in \mathbb{C} \mid |z-x| < \delta\}$ for the open disk of radius $\delta$ centered at $x$.

\begin{theorem}\label{thm:stability-of-partition-function}
    %Fix constants $\gamma > \beta > 1$. Let $\lambda^* = \lambda^*(\beta,\gamma)$. 
    Fix three constants $\gamma > \beta > 1$ and $0 < \lambda < \lambda^*(\beta, \gamma)$.
    Let $\rho = \frac{\beta-1}{\gamma-1}$ and $\theta = \frac{(\gamma-1)^2}{\beta\gamma-1}$.
    Let $G = (V,E)$ be a graph and $\boldsymbol{\lambda} = (\lambda_v)_{v \in V}$ be parameters such that $-\rho \lambda \leq \lambda_v \leq \lambda$ for all $v \in V$.
    Let $w: \{0,1\}^E \to \mathbb{R}$ be the weight function defined in \eqref{eq:w-weight-function} with respect to graph $G$ and parameters $\boldsymbol{\lambda},\rho$.
    There exists a constant $\delta = \delta(\beta,\gamma,\lambda) > 0$ such that the partition function $Z(\bm{z})$ defined in~\eqref{eq:conditional-partition-function-polynomial} is $\prod_{e\in E} \mathcal{U}_\delta(\theta)$-stable.
\end{theorem}

This theorem is enough for the stability hypothesis in \Cref{prop:stability-to-spectral-independence}. 
Indeed, after the self-reduction above, each polynomial $Z^\tau(\boldsymbol z)$ appearing in \Cref{prop:stability-to-spectral-independence} is exactly the partition function in \Cref{thm:stability-of-partition-function} for the graph $G^\tau$ and fields $(\lambda_v^\tau)_{v\in V}$. 
This is because, in the definition of $Z^\tau(\boldsymbol z)$ in~\eqref{eq:conditional-partition-function-polynomial}, 
only the edge activities of the unpinned edges in $\Lambda$ appear. 
Thus proving \Cref{thm:stability-of-partition-function} for arbitrary graphs and all fields in the interval $(-\rho\lambda^\star,\lambda^\star)$ establishes the required stability for all conditional distributions.
\Cref{thm:mixing-time-sgw} is a consequence of combining \Cref{lem:spectral-independence-to-mixing}, \Cref{prop:stability-to-spectral-independence}, and \Cref{thm:stability-of-partition-function}.

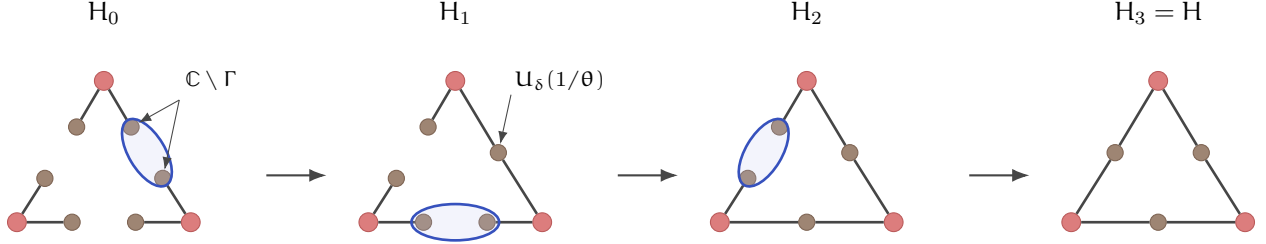
\begin{figure}[t]
\centering
\resizebox{\textwidth}{!}{%
\begin{tikzpicture}[>=Latex, scale=1]

% --------------------------------------------------
% Colors
% --------------------------------------------------
\definecolor{myred}{RGB}{205,70,70}
\definecolor{mybrown}{RGB}{120,85,60}
\definecolor{myblue}{RGB}{60,90,210}
\definecolor{myedge}{RGB}{70,70,70}

% --------------------------------------------------
% Styles
% --------------------------------------------------
\tikzset{
  root/.style={
    circle,
    draw=myred!85!black,
    fill=myred,
    draw opacity=0.85,
    fill opacity=0.70,
    inner sep=2.5pt
  },
  halfedge/.style={
    circle,
    draw=mybrown!85!black,
    fill=mybrown,
    draw opacity=0.85,
    fill opacity=0.72,
    inner sep=2.15pt
  },
  edge/.style={
    draw=myedge,
    line width=1.0pt
  },
  arr/.style={
    ->,
    thick,
    draw=myedge
  },
  hl/.style={
    draw=myblue!90!black,
    fill=myblue!35,
    fill opacity=0.18,
    line width=1.0pt
  },
  lab/.style={
    font=\small
  },
  annot/.style={
    font=\scriptsize,
    inner sep=1pt
  },
  annotarrow/.style={
    ->,
    thin,
    draw=myedge
  }
}

% ==================================================
% H0
% ==================================================
\begin{scope}[shift={(0,0)}]
  \node[lab] at (0,2.15) {$H_0$};

  % roots
  \node[root] (T0) at (0,1.25) {};
  \node[root] (R0) at (1.15,-0.62) {};
  \node[root] (L0) at (-1.15,-0.62) {};

  % half-edge vertices
  % pair for T-R
  \node[halfedge] (T0R) at (0.36,0.64) {};
  \node[halfedge] (R0T) at (0.78,-0.04) {};

  % pair for R-L
  \node[halfedge] (R0L) at (0.42,-0.62) {};
  \node[halfedge] (L0R) at (-0.42,-0.62) {};

  % pair for L-T
  \node[halfedge] (L0T) at (-0.78,-0.04) {};
  \node[halfedge] (T0L) at (-0.36,0.64) {};

  % edges
  \draw[edge] (T0)--(T0R);
  \draw[edge] (T0)--(T0L);
  \draw[edge] (R0)--(R0T);
  \draw[edge] (R0)--(R0L);
  \draw[edge] (L0)--(L0T);
  \draw[edge] (L0)--(L0R);

  % highlight first contraction pair
  \draw[hl, rotate around={-59:(0.57,0.30)}]
    (0.57,0.30) ellipse [x radius=0.50, y radius=0.24];

  % annotation
  \node[annot] (A0) at (1.42,1.25) {$\mathbb{C}\setminus \Gamma$};

% two arrows from the label to the two brown vertices
\draw[annotarrow] ([xshift=-1pt,yshift=-2pt]A0.south west) -- (T0R);
\draw[annotarrow] ([xshift=-1pt,yshift=-2pt]A0.south west) -- (R0T);
\end{scope}

% arrow H0 -> H1
\draw[arr] (2.15,0) -- (2.95,0);

% ==================================================
% H1
% ==================================================
\begin{scope}[shift={(4.65,0)}]
  \node[lab] at (0,2.15) {$H_1$};

  % roots
  \node[root] (T1) at (0,1.25) {};
  \node[root] (R1) at (1.15,-0.62) {};
  \node[root] (L1) at (-1.15,-0.62) {};

  % contracted midpoint for T-R
  \node[halfedge] (M1TR) at (0.57,0.30) {};

  % remaining half-edge vertices
  \node[halfedge] (R1L) at (0.42,-0.62) {};
  \node[halfedge] (L1R) at (-0.42,-0.62) {};
  \node[halfedge] (L1T) at (-0.78,-0.04) {};
  \node[halfedge] (T1L) at (-0.36,0.64) {};

  % edges
  \draw[edge] (T1)--(M1TR);
  \draw[edge] (R1)--(M1TR);
  \draw[edge] (R1)--(R1L);
  \draw[edge] (L1)--(L1R);
  \draw[edge] (L1)--(L1T);
  \draw[edge] (T1)--(T1L);

  % annotation
  \node[annot] (A1) at (1.38,1.25) {$U_\delta(1/\theta)$};
  \draw[annotarrow] (A1.south west) -- (M1TR);

  % highlight second contraction pair
  \draw[hl]
    (0,-0.62) ellipse [x radius=0.58, y radius=0.24];
\end{scope}

% arrow H1 -> H2
\draw[arr] (6.80,0) -- (7.60,0);

% ==================================================
% H2
% ==================================================
\begin{scope}[shift={(9.30,0)}]
  \node[lab] at (0,2.15) {$H_2$};

  % roots
  \node[root] (T2) at (0,1.25) {};
  \node[root] (R2) at (1.15,-0.62) {};
  \node[root] (L2) at (-1.15,-0.62) {};

  % contracted midpoints
  \node[halfedge] (M2TR) at (0.57,0.30) {};
  \node[halfedge] (M2RL) at (0,-0.62) {};

  % remaining half-edge vertices
  \node[halfedge] (L2T) at (-0.78,-0.04) {};
  \node[halfedge] (T2L) at (-0.36,0.64) {};

  % edges
  \draw[edge] (T2)--(M2TR);
  \draw[edge] (R2)--(M2TR);
  \draw[edge] (R2)--(M2RL);
  \draw[edge] (L2)--(M2RL);
  \draw[edge] (L2)--(L2T);
  \draw[edge] (T2)--(T2L);

  % highlight third contraction pair
  \draw[hl, rotate around={59:(-0.57,0.30)}]
    (-0.57,0.30) ellipse [x radius=0.50, y radius=0.24];
\end{scope}

% arrow H2 -> H3
\draw[arr] (11.45,0) -- (12.25,0);

% ==================================================
% H3 = H
% ==================================================
\begin{scope}[shift={(13.95,0)}]
  \node[lab] at (0,2.15) {$H_3 = H$};

  % roots
  \node[root] (T3) at (0,1.25) {};
  \node[root] (R3) at (1.15,-0.62) {};
  \node[root] (L3) at (-1.15,-0.62) {};

  % edge-vertices
  \node[halfedge] (M3TR) at (0.57,0.30) {};
  \node[halfedge] (M3RL) at (0,-0.62) {};
  \node[halfedge] (M3LT) at (-0.57,0.30) {};

  % final graph
  \draw[edge]
    (T3)--(M3TR)--(R3)--(M3RL)--(L3)--(M3LT)--(T3);
\end{scope}

\end{tikzpicture}%
}
\caption{Illustration of our contraction procedure for a $3$-cycle $C_3$. Red vertices correspond to original vertices of the graph, and brown vertices correspond to edge variables. We start from $H_0$, the disjoint union of three disjoint stars. At each step, the blue highlighted pair of brown vertices is contracted into a single vertex. The variables of original brown vertices are \textit{safe} in $\mathbb{C}\setminus \Gamma$, while the contracted brown vertex is \textit{safe} in $U_\delta(1/\theta)$. After three contractions, we obtain $H_3=H_{C_3}$. }
\label{fig:contraction-3cycle}
\end{figure}

\Cref{thm:stability-of-partition-function} is our most technically involved ingredient and is shown in \Cref{sec:stability-of-polynomial}. 
Here we give a brief overview of the proof.
The overall strategy is via the Asano contraction \cite{Asa70,Rue71}, similar to previous zero-freeness results for Holant problems \cite{GLLZ21}.
Let $G=(V,E)$ be a graph with $\abs{V}=n$ and $\abs{E}=m$. 
Note that the variables in $Z_w^\tau(\boldsymbol z)$ are on the edges.
Let $H_G$ be the (bipartite) vertex-edge incidence graph of $G$ (see \Cref{def:graph-HG}).
We construct a sequence of graphs $H_0,\dots,H_m$, where $H_m=H_G$.
The initial graph $H_0$ is a collection of disjoint star graphs.
For each $v\in V$, construct a star whose center is $v$ and the leaves correspond to all edges adjacent to $v$.
The union of all these stars is $H_0$.
Thus, in $H_0$, each edge $e=(u,v)\in E$ corresponds two leaves connecting to $u$ and $v$, respectively.
Call these leaves half-edges, and variables are attached to the half-edges.
Order edges in $E$ as $e_1,\dots,e_m$.
Then given $H_i$, contract the two half edges corresponding to $e_{i+1}$ to obtain $H_{i+1}$.
The two corresponding variables are also identified.
After processing all edges, we recover $H_m=H_G$.
\Cref{fig:contraction-3cycle} illustrates this procedure for a $3$-cycle.

To prove the desired stability, there are two main ingredients.
The first is the stability of the initial graph $H_0$.
Since $H_0$ is a disjoint union of stars, the stability of the partition function follows from handling each star separately. 
We carefully choose a region $\Gamma$ so that the partition function is non-zero whenever no variable lies in $\Gamma$.
The second ingredient is the stability-preserving property of the contraction step. 
Namely, assuming that the partition function of $H_i$ is $\left(\mathbb{C}\setminus \Gamma\right)\times\left(\mathbb{C}\setminus \Gamma\right)$-stable for the pair of contracted half edge variables,
we prove that the new partition function for $H_{i+1}$ after contraction remains stable when the contracted variable lies in a small neighborhood $U_\delta(1/\theta)$ around $1/\theta$.
%(Due to technical reasons, we change of variables \(q=1/z\), so we consider the neighborhood of \(1/\theta\) in the \(q\)-coordinate, instead of the original \(z\)-coordinate.)
Iterating this argument propagates stability from the initial disjoint union of stars $H_0$ all the way to the final graph $H_m$. 
Consequently, the WSG partition function for $G$ is $\prod_{e\in E} \mathcal{U}_\delta(\theta)$-stable.

The most essential part of this argument is the choice of $\Gamma$.
Our choice is actually inspired by the zero-freeness result of \cite{GLL20}.
Although their result concerns the spin model and we are working in the WSG model,
we noticed that after certain variable transformations, the initial stability for $H_0$ resembles the contraction step in \cite{GLL20}.
Then, with the same variable transformation, our contraction step resembles the initial stability in \cite{GLL20}.
This is an interesting duality regarding the contraction argument for zero-freeness under holographic transformations,
and is also why our \Cref{thm:fpras-spin} works up to the same $\lambda$ threshold as in \cite{GLL20}.
However, as we need stability for a constant neighbourhood around $\theta$, we also need to strengthen various intermediate results in \cite{GLL20}.

In addition, there is one more wrinkle to our proof -- the zero-freeness result of \cite{GLL20} is optimised and only works for degrees $\ge 2$.
The way they handle degree $1$ vertices (namely leaves) is to pre-process them first.
However, since pinning edges may produce new leaves,
this strategy does not work for us.
To this end, we identify \emph{bad} leaves, whose external fields are positive and cause trouble to an avoidance lemma (\Cref{lem:lambda_avoid}).
Observe that any graph can be decomposed into a degree $\ge 2$ core and trees appended to the core (see \Cref{fig:leaf-contraction-overview} for an illustration). 
Our strategy is to recursively prune bad leaves and reduce the original graph $G$ to an equivalent bad-leaf reduced graph,
whose external fields are updated to some $\boldsymbol{\lambda}'$.
Throughout this recursive process, we must ensure that the new fields $\boldsymbol{\lambda}'$ remain in a complex region where the stability argument applies.
To control this recursion, we notice that the external fields are updated in a similar fashion to the tree recursion of Weitz \cite{weitz2006counting}.
Thus, we use the potential function from~\cite{GuoL18} to design the admissible complex region for the external fields $\boldsymbol{\lambda}'$. 
Intuitively, we show that the perturbations due to the pruning process do not drift the external fields $\boldsymbol{\lambda}'$ too far away from the real axis, 
and that they remain inside the region required by the contraction argument. 
Once we obtain a graph without bad leaves, the zero-freeness argument outlined above applies, which completes the proof of \Cref{thm:stability-of-partition-function}.

\section{Equivalent formulations}\label{sec:3-models}

In this section we show equivalence between the two-spin system and the weighted subgraph model,
introduce the random cluster type model, and show its equivalence to the other two models.

\subsection{Weighted subgraphs}\label{sec:WSG}

We use the holographic transformation~\cite{Valiant08} to prove the equivalence between the two-spin system and the weighted subgraph model.
Let $d \geq 1$ be an integer, and let $f:\{0,1\}^d \to \mathbb{C}$ be a function, where $\mathbb{C}$ denotes the complex plane. We write $f = (f_{x_0},f_{x_1},\ldots,f_{x_{2^d-1}})$ as either a column or a row vector, where for each $0 \leq i < 2^d$, $x_i \in \{0,1\}^d$ is the binary representation of $i$. 
In the symmetric case, where $f$ is invariant under permutations of the bits of $x_i$, we write $f$ succinctly as $f = [f_0,f_1,\ldots,f_{d}]$, where for each $0 \leq j \leq d$, $f_j$ is the common value of $f_x$ over all binary strings $x \in \{0,1\}^d$ with Hamming weight $|x| = j$.
This is also called the signature of $f$.

Given a bipartite graph $H=(V_1 \uplus V_2, E)$, let $\+F = (f_v)_{v \in V_1}$ and $\+G = (g_u)_{u \in V_2}$ denote the functions on the vertices of $H$. The Holant is defined as 
\begin{align*}
    \operatorname{Holant}(H; \+F \mid \+G) = \sum_{\sigma: E \to \{0,1\}} \prod_{v \in V_1} f_v(\sigma_{E(v)}) \prod_{u \in V_2} g_u(\sigma_{E(u)}),
\end{align*}
where $E(v)$ and $E(u)$ are the edges incident to $v$ and $u$, respectively.

Let $\boldsymbol{T} \in \mathbb{C}^{2 \times 2}$ be an invertible matrix, and let $f: \{0,1\}^d \to \mathbb{C}$ be a signature function. We write $\boldsymbol{T}f = \boldsymbol{T}^{\otimes d} f$ (resp. $f\boldsymbol{T} = f \boldsymbol{T}^{\otimes d}$) for the transformed function when $f$ is viewed as a column (resp. row) vector. Given $\operatorname{Holant}(H; \+F \mid \+G)$ on a bipartite graph $H=(V_1 \uplus V_2, E)$, view all $f_v \in \+F$ as row vectors and all $g_u \in \+G$ as column vectors.
Define the transformed functions as $ \+F\boldsymbol{T} = \{f_v\boldsymbol{T}\}_{v \in V_1}$ and $\boldsymbol{T}^{-1}\+G = \{\boldsymbol{T}^{-1}g_u\}_{u \in V_2}$, where $\boldsymbol{T}^{-1}$ is the inverse of $\boldsymbol{T}$. The following equivalence holds.
\begin{theorem}[\text{Valiant's Holant theorem~\cite{Valiant08}}]\label{thm:holant-theorem}
    $\operatorname{Holant}(H; \+F \boldsymbol{T} \mid \boldsymbol{T}^{-1}\+G) = \operatorname{Holant}(H; \+F \mid \+G)$.
\end{theorem}

Next, we represent the two-spin system as a Holant instance via the \emph{half-edge} decomposition. 
Consider a ferromagnetic two-spin system on a graph $G=(V,E)$ defined by parameters $\beta, \gamma, \lambda > 0$.
We define the vertex-edge incidence bipartite graph $H_G$ of $G$ by the following.

\begin{definition}\label{def:graph-HG}
    Given $G=(V,E)$, the vertex-edge incidence graph $H_G = (V_1 \uplus V_2, E_H)$ is a bipartite graph whose left part $V_1 = V$ consists of all vertices of $G$ and whose right part $V_2 = E$ consists of all edges of $G$. For each edge $e = \{u,v\} \in E$ of $G$, we add two edges $e_u = \{u,e\}$ and $e_v = \{v,e\}$ to $E_H$.
    We call $e_u$ and $e_v$ the half-edges of $e$.
\end{definition}

 Note that the degree of each $v \in V_1$ in $H_G$ equals the degree of $v$ in $G$, which we denote by $d_v = \deg_G(v) = \deg_{H_G}(v)$, while the degree of each $e \in V_2$ in $H_G$ is exactly $2$. For each vertex $v \in V_1$, we define the row vector $f_{v}:\{0,1\}^{d_v} \to \mathbb{C}$ by 
\begin{align*}
f_v = \lambda\, (1,0)^{\otimes d_v} + (0,1)^{\otimes d_v} = [\lambda, \underbrace{0,\ldots,0}_{(d_v-1) \text{ zeros}}, 1],
\end{align*}
In words, if all half-edges incident to $v$ are in state $0$, then the value of $f_v$ is $\lambda$; if all half-edges incident to $v$ are in state $1$, then the value of $f_v$ is $1$; otherwise, the value of $f_v$ is $0$. For each edge $e = \{u,v\} \in E$ of $G$, we define the column vector $g_e:\{0,1\}^2 \to \mathbb{C}$ by 
\begin{align*}
g_e = (\beta, 1, 1, \gamma)^T = [\beta, 1, \gamma].
\end{align*}
In words, if the two half-edges of $e$ are both in state $0$, then the value of $g_e$ is $\beta$; if they are both in state $1$, then the value of $g_e$ is $\gamma$; otherwise, the value of $g_e$ is $1$. We define
\begin{align}\label{eq:F-G}
\+F = \{f_v\}_{v \in V_1} \quad \text{and} \quad \+G = \{g_e\}_{e \in V_2}.
\end{align}

This construction encodes the two-spin system as a Holant instance. Indeed, each spin configuration $\sigma \in \{0,1\}^V$ of $G$ corresponds to the assignment in $H_G$ that gives every half-edge incident to a vertex $v$ the spin value $\sigma_v$. Under this correspondence, the vertex function $f_v$ contributes the external field at $v$, while the edge function $g_e$ contributes the interaction along $e = \{u,v\}$. Comparing the weights term by term, the following proposition is immediate from the definitions of $f_v$ and $g_e$.
\begin{proposition}\label{prop:spin-holant}
    The partition function $Z_G^{\rm spin}(\beta, \gamma, \lambda)$ of the two-spin system on $G$ is given by
    \begin{align*}
        Z_G^{\rm spin}(\beta, \gamma, \lambda) = \operatorname{Holant}(H_G; \+F \mid \+G),
    \end{align*}
    where $\+F$ and $\+G$ are defined in \eqref{eq:F-G}.
\end{proposition}

We now apply the holographic transformation to the Holant instance of the two-spin system and compute the transformed signatures. 
Since $\gamma \neq 1$, define
\begin{align}\label{eq:H-matrix}
      \boldsymbol{T}\defeq\begin{pmatrix}1&-\rho\\ 1&1\end{pmatrix},
      \qquad \text{ where }
      \rho=\frac{\beta-1}{\gamma-1}.
\end{align}
Also, as $\beta+\gamma\neq 2$, $\boldsymbol{T}$ is invertible.
For any vertex $v \in V_1$ of $H_G$, the transformed vertex signature is
\begin{align}\label{eq:transformed-vertex-signature}
    f_v \boldsymbol{T}^{\otimes d_v}
    &= \lambda\, (1,-\rho)^{\otimes d_v} + (1,1)^{\otimes d_v}
    = [1+\lambda, 1-\lambda\rho, \ldots, 1+\lambda(-\rho)^{d_v}].
\end{align}
Thus the transformed vertex signature depends only on the number of incident transformed half-edges in state $1$: its value on assignments of Hamming weight $k$ is $1+\lambda(-\rho)^k$.
For the edge signature, recall that $\theta = \frac{(\gamma-1)^2}{\beta\gamma-1}$.
Using
\[
    \boldsymbol{T}^{-1}=\frac{1}{1+\rho}
    \begin{pmatrix}
        1 & \rho \\
        -1 & 1
    \end{pmatrix},
\]
a direct calculation gives
\begin{align}\label{eq:transformed-edge-signature}
    (\boldsymbol{T}^{-1})^{\otimes 2}g_e
    &=
    \frac{1}{(1+\rho)^2}
    \begin{pmatrix}
        \beta+2\rho+\gamma\rho^2 \\
        -(\beta-1)+\rho(\gamma-1) \\
        -(\beta-1)+\rho(\gamma-1) \\
        \beta+\gamma-2
    \end{pmatrix} = \frac{1}{1-\rho\theta}
    \begin{pmatrix}
        1 \\
        0 \\
        0 \\
        \theta
    \end{pmatrix}.
\end{align}
Thus the transformed edge signature is an \emph{equality} signature on the two transformed half-edges, with weight $1$ on state $00$ and weight $\theta$ on state $11$, multiplied by the common factor $\frac{1}{1-\rho\theta}$. This is the exact reason why we set the parameter $\rho$ in~\eqref{eq:H-matrix}.
Since every pair of half-edges are forced to be equal, by mapping the half-edges to the edges, this exactly gives the weighted subgraph model on the original graph $G$.
Therefore, by \Cref{thm:holant-theorem} and \Cref{prop:spin-holant},
\begin{align*}
    Z_G^{\rm spin}(\beta, \gamma, \lambda)
    &= \operatorname{Holant}(H_G; \+F  \mid \+G) = \operatorname{Holant}(H_G; \+F\boldsymbol{T} \mid \boldsymbol{T}^{-1}\+G) \\
    &= (1-\rho\theta)^{-|E|} \sum_{S \subseteq E} \theta^{|S|}
       \prod_{v \in V} \left(1+\lambda(-\rho)^{\deg_S(v)}\right) \\
    &= (1-\rho\theta)^{-|E|} Z^{\text{wsg}}_G(\theta\boldsymbol{1}_E, \lambda\boldsymbol{1}_V, \rho),
\end{align*}
where $\deg_S(v)$ is the degree of vertex $v$ in the subgraph $G'= (V,S) \subseteq G$.
This proves \Cref{thm:equivalence}.

\subsection{Random cluster type model}\label{sec:RC}

Next we introduce a random cluster (RC) type model \`a la Fortuin and Kasteleyn \cite{FK72},
and couple this model with the $2$-spin system and the WSG model, respectively. 
This way, we can transform a sample from the WSG model to a sample from the spin system, and establish \Cref{thm:sampling}.

The parameters are $\gamma > \beta > 1$ and $\lambda>0$.
As before, let $\rho=\frac{\beta-1}{\gamma-1}$ and $\theta=\frac{(\gamma-1)^2}{\beta\gamma-1}$.
Given a graph $G=(V,E)$, the RC distribution is defined over subsets $S$ of edges:
\begin{align}\label{eqn:mu-RC}
  \mu^{\text{rc}}(S)\propto w^{\text{rc}}(S)\defeq \prod_{C\in \kappa(V,S)}\left( \lambda^{\abs{V(C)}}(\beta-1)^{\abs{E(C)}}+(\gamma-1)^{\abs{E(C)}} \right),
\end{align}
where $\kappa(V,S)$ is the set of connected components in the subgraph $(V,S)$, $E(C)$ is the edge set of the component $C$, and $V(C)$ is the vertex set of the component $C$.
Define the corresponding partition function $Z^{\text{rc}}=Z^{\text{rc}}_{G}(\beta,\gamma,\lambda)\defeq\sum_{S\subseteq E}w^{\text{rc}}(S)$.
It is not hard to verify that for ferromagnetic Ising models ($\beta=\gamma>1$ and $\lambda=1$), \eqref{eqn:mu-RC} coincides with the classical Fortuin-Kasteleyn model \cite{FK72}.
For ferromagnetic Ising models with external fields ($\beta=\gamma>1$ and $\lambda\neq1$), \eqref{eqn:mu-RC} coincides with the weighted random cluster model \cite{FGW23}.

The next theorem is a generalisation of the equivalence by Fortuin and Kasteleyn \cite{FK72} and the coupling by Edwards and Sokal \cite{ES88} between ferromagnetic Ising models and random cluster models.

\begin{theorem}  \label{thm:spin-RC}
  For any graph $G$ and parameters $(\beta,\gamma,\lambda)$,
  $Z^{\text{rc}}_{G}(\beta,\gamma,\lambda) = Z^{\text{spin}}_{G}(\beta,\gamma,\lambda)$,
  where $Z^{\text{spin}}$ is defined in \eqref{eqn:Z-spin}.

  Moreover, given an RC sample $S\sim\mu^{\text{rc}}$,
  one can obtain a spin sample $\sigma$ obeying the law of $\mu^{\text{spin}}$ as follows:
  independently for each connected component $C\in\kappa(V,S)$, 
  assign spin $0$ to all vertices in $C$ with probability $\frac{\lambda^{\abs{V(C)}}(\beta-1)^{\abs{E(C)}}}{\lambda^{\abs{V(C)}}(\beta-1)^{\abs{E(C)}}+(\gamma-1)^{\abs{E(C)}}}$, 
  and spin $1$ otherwise.
\end{theorem}
\begin{proof}
  Recall that the ferromagnetic $2$-spin interaction can be represented as a Holant problem on the vertex-edge incidence graph $H_G$.
  A vertex $v$ is assigned the function $f_v=\lambda(1,0)^{\otimes d_v}+(0,1)^{\otimes d_v}$ and the edges assigned $g=(\beta, 1, 1, \gamma)^T$.
  Let $g_0=(1,1,1,1)^T$ and $g_1=(\beta-1,0,0,\gamma-1)^T$ so that $g=g_0+g_1$.
  We may also view $g$, $g_0$, and $g_1$ as symmetric functions on two variables.
  Then, we have
  \begin{align*}
    Z^{\text{spin}} & = \sum_{\sigma: V \to \{0,1\}} \lambda^{n_0(\sigma)} \prod_{(u,v) \in E} g(\sigma_u,\sigma_v)
    = \sum_{\sigma: V \to \{0,1\}} \lambda^{n_0(\sigma)} \prod_{(u,v) \in E} (g_0(\sigma_u,\sigma_v)+g_1(\sigma_u,\sigma_v))\\
    & = \sum_{\sigma: V \to \{0,1\}} \lambda^{n_0(\sigma)} \sum_{\tau: E\to \{0,1\}}\prod_{e=(u,v) \in E}g_{\tau(e)}(\sigma_u,\sigma_v)\\
    & = \sum_{\tau: E\to \{0,1\}} \sum_{\sigma: V \to \{0,1\}} \lambda^{n_0(\sigma)} \prod_{e=(u,v) \in E}g_{\tau(e)}(\sigma_u,\sigma_v).
  \end{align*}
  Given an edge assignment $\tau$, let $S$ be the set of edges that are assigned $1$. 
  In fact, since $g_0$ is the all-$1$ function, its contribution is the same as removing the edge.
  On the other hand, $g_1$ is non-zero only if both half edges are $0$ or $1$.
  This implies that, the only $\sigma$ with non-zero weight is to assign the same value for each connected component $C\in\kappa(V,S)$.
  The weight of assigning $0$ is $\lambda^{\abs{V(C)}}(\beta-1)^{\abs{E(C)}}$ and assigning $1$ is $(\gamma-1)^{\abs{E(C)}}$.
  The overall weight is exactly $w^{\text{rc}}(S)$, which proves the first part of the theorem.

  For the second part, consider a joint distribution over both edges and vertices.
  The only allowed configuration is to have the same value in each connected component induced by the edges,
  and the weight is as given above.
  It is not hard to see that this distribution projected on the vertices recovers the spin distribution,
  and on the edges recovers the RC distribution.
  The second part of the theorem follows.
\end{proof}

The reason we introduce the RC model is because we can also couple it with the WSG model,
which is a generalisation of the coupling used in \cite{GJ09,FGW23}.
Denote by $\mu^{\text{wsg}}$ the Gibbs distribution of the WSG model, as defined in \eqref{eqn:wsg-dist} with $\boldsymbol{\theta}=\theta\boldsymbol{1}$ and $\boldsymbol{\lambda}=\lambda\boldsymbol{1}$.
Also recall that $\rho\theta=\frac{(\beta-1)(\gamma-1)}{\beta\gamma-1}<1$.

\begin{lemma}  \label{lem:RC-WSG}
  For a graph $G=(V,E)$ and parameters $\gamma > \beta > 1$ and $\lambda>0$,
  let $S$ be a sample from $\mu^{\text{wsg}}$.
  Independently add each remaining edge in $E\setminus S$ with probability $\rho\theta$ to get $R$.
  Then $R$ follows the law of $\mu^{\text{rc}}$.
\end{lemma}

To prove \Cref{lem:RC-WSG}, we will need the following lemma.

\begin{lemma}  \label{lem:subgraph-weight}
  Let $G=(V,R)$ be a graph.
  Then,
  \begin{align}\label{eqn:subgraph-weight}
    \sum_{S\subseteq R}\rho^{-\abs{S}} w^{\text{wsg}}_{\text{vtx}}(S) = \left( \frac{\beta+\gamma-2}{(\beta-1)(\gamma-1)} \right)^{\abs{R}} w^{\text{rc}}(R),
  \end{align}
  where $w^{\text{wsg}}_{\text{vtx}}(S)\defeq\prod_{v\in V}\left( 1+\lambda(-\rho)^{\deg_S(v)} \right)$.
\end{lemma}
\begin{proof}
  This is a holographic transformation via \Cref{thm:holant-theorem}.
  Again, let $H=H_G$ be the vertex-edge incidence graph.
  Then, $w^{\text{rc}}(R)$ is the Holant by giving each $v\in V$ a function $f_v=\lambda(1,0)^{\otimes d_v}+(0,1)^{\otimes d_v}$, and each $e\in E$ a function $g=[(\beta-1),0,(\gamma-1)]$.
  This is because with these functions, the only configuration with non-zero weights have the same value in each connected component $C$,
  and the weight is $\lambda^{\abs{V(C)}}(\beta-1)^{\abs{E(C)}}$ if the value is $0$, and $(\gamma-1)^{\abs{E(C)}}$ if the value is $1$.

  Given this Holant, we apply the same transformation as in \Cref{thm:equivalence} by $\boldsymbol{T}\defeq\begin{pmatrix}1&-\rho\\ 1&1\end{pmatrix}$.
  Recall that $f_v \boldsymbol{T}^{\otimes d_v}$ is given in \eqref{eq:transformed-vertex-signature} already.
  On the edges, the resulting function is 
  \begin{align*}
    (\boldsymbol{T}^{-1})^{\otimes 2}g & = \frac{1}{(1+\rho)^2}
    \begin{pmatrix}
      \beta-1+(\gamma-1)\rho^2\\ 0\\ 0\\ \beta+\gamma-2
    \end{pmatrix}
    = \left( \frac{\gamma-1}{\beta+\gamma-2} \right)^2\begin{pmatrix}
      \frac{(\beta-1)(\beta+\gamma-2)}{\gamma-1}\\ 0\\ 0\\ \beta+\gamma-2
    \end{pmatrix}\\
    & = \frac{(\beta-1)(\gamma-1)}{\beta+\gamma-2} \begin{pmatrix}
      1\\ 0\\ 0\\ \rho^{-1}
    \end{pmatrix}.
  \end{align*}

  On the other hand, the left hand side of \eqref{eqn:subgraph-weight} can be represented as a Holant as well.
  The vertex function is exactly  $f_v \boldsymbol{T}^{\otimes d_v}$ for any $v\in V$,
  and the edge function is $[1,0,\rho^{-1}]$.
  This proves the lemma.
\end{proof}

Now we can prove \Cref{lem:RC-WSG}.

\begin{proof}[Proof of \Cref{lem:RC-WSG}]
  Note that $\mu^{\text{wsg}}(S)=\frac{\theta^{\abs{S}}w^{\text{wsg}}_{\text{vtx}}(S)}{Z^{\text{wsg}}}$.
  Then, the probability of getting $R$ is
  \begin{align*}
    &\;\sum_{S\subseteq R}\frac{\theta^{\abs{S}}w^{\text{wsg}}_{\text{vtx}}(S)}{Z^{\text{wsg}}}\left( \rho\theta \right)^{\abs{R}-\abs{S}}\left( 1-\rho\theta \right)^{\abs{E}-\abs{R}}\\
    =\;&\; \frac{(\rho\theta)^{\abs{R}}\left( 1-\rho\theta \right)^{\abs{E}-\abs{R}}}{Z^{\text{wsg}}} \sum_{S\subseteq R}\rho^{-\abs{S}}w^{\text{wsg}}_{\text{vtx}}(S) \\
    =\;&\; \frac{\left( 1-\rho\theta \right)^{\abs{E}}}{Z^{\text{wsg}}} \left( \frac{\rho\theta}{1-\rho\theta} \right)^{\abs{R}}\left( \frac{\beta+\gamma-2}{(\beta-1)(\gamma-1)} \right)^{\abs{R}} w^{\text{rc}}(R)
    \tag{by \protect\Cref{lem:subgraph-weight}}.
  \end{align*}
  By \Cref{thm:equivalence} and \Cref{thm:spin-RC}, $Z^{\text{wsg}}=\left( 1-\rho\theta \right)^{\abs{E}}Z^{\text{spin}}=\left( 1-\rho\theta \right)^{\abs{E}}Z^{\text{rc}}$.
  Moreover, $\frac{\rho\theta}{1-\rho\theta}=\frac{(\beta-1)(\gamma-1)}{\beta+\gamma-2}$.
  Thus, the probability above is $\frac{ w^{\text{rc}}(R)}{Z^{\text{rc}}}$, which is exactly $\mu^{\text{rc}}(R)$.
\end{proof}

We remark that one can also strengthen \Cref{thm:spin-RC} and  \Cref{lem:RC-WSG} into a grand model as in \cite{FGW23} in a straightforward way.
However, it does not seem to benefit our applications so the details are omitted here.
In any case, we conclude this section with a proof of \Cref{thm:sampling}.

\begin{proof}[Proof of \Cref{thm:sampling}]
  We run Glauber dynamics for the WSG model for $O\left( \Delta^C \cdot n \log \frac{n}{\epsilon} \right)$ steps to get a subgraph $S$,
  where $C$ is from \Cref{thm:mixing-time-sgw}.
  Since each step of Glauber dynamics can be implemented in time $O(\Delta)$, the total running time here is $O\left( \Delta^{C+1} \cdot n \log \frac{n}{\epsilon} \right)$.
  By \Cref{thm:mixing-time-sgw} and the coupling inequality,
  we may couple $S$ with a perfect sample from $\mu^{\text{wsg}}$ with probability at least $1-\epsilon$.
  Then, we use \Cref{lem:RC-WSG} to get an RC sample $R$, and subsequently use \Cref{thm:spin-RC} to get a spin sample $\sigma$.
  This $\sigma$ can be coupled with a perfect sample from $\mu^{\text{spin}}$ with probability at least $1-\epsilon$.
  Thus, the total variation distance of our sample from $\mu^{\text{spin}}$ is at most $\epsilon$.
\end{proof}

\section{Stability of the polynomial}\label{sec:stability-of-polynomial}

In this section, we prove \Cref{thm:stability-of-partition-function}.
We decompose the argument into two steps. First, we reduce an arbitrary graph $G$ to a graph $G'$ by contracting \textit{bad} leaves (the degree-one vertices with positive external fields). Second, we prove the stability result for the reduced graph $G'$. The latter step relies crucially on the leaf-contraction operation developed in the reduction step, so we present the reduction first. 

Before we proceed, we explain why this preliminary reduction is needed.
This is due to a technical reason: our proof in the second step cannot directly handle degree-one vertices with positive external fields.
(In particular, \Cref{lem:lambda_avoid} does not hold for leaves.)
A related obstacle also appears in~\cite{GLL20}.
One might try to avoid it by preprocessing the original graph $G$, deleting all degree-one vertices in advance, and then analyzing the weighted subgraph model only on graphs with minimum degree at least two.
This shortcut was used in~\cite{GLL20}, but it is not sufficient here.
Our stability result must hold not only for the original graph, but also after arbitrary pinnings of the edge variables.
Pinning an edge to $0$ is equivalent to deleting that edge.
Thus, even if the original graph has no degree-one vertices after preprocessing, 
the pinning may produce new degree-one vertices that require us to handle.
%We therefore need a reduction that can be applied uniformly to the instances arising after pinning, rather than only to the original graph.

We next recall some basic settings. 
Let $\gamma > \beta > 1$ and $0 < \lambda < \lambda^*(\beta, \gamma)$. Let $\rho = \frac{\beta-1}{\gamma-1}$ and $\theta = \frac{(\gamma-1)^2}{\beta\gamma-1}$.
Let $G = (V,E)$ be a graph and $\boldsymbol{\lambda} = (\lambda_v)_{v \in V}$ be parameters such that $-\rho \lambda \leq \lambda_v \leq \lambda$ for all $v \in V$.
We will call $\lambda_v$ the \emph{external field} of vertex $v$.
Let $w: \{0,1\}^E \to \mathbb{R}$ be the weight function defined in \eqref{eq:w-weight-function} with respect to graph $G$ and parameters $\boldsymbol{\lambda},\rho$. Formally, for any $S \subseteq E$,
\begin{align}\label{eq:w-weight-function-general}
w(S) = \prod_{v \in V} (1+\lambda_v(-\rho)^{\deg_{S}(v)}),
\end{align}
where $\deg_{S}(v)$ is the degree of vertex $v$ in the subgraph $G_S = (V,S)$.
%We may also view $w$ as a function from $\{0,1\}^E \to \mathbb{R}$ by setting $w(\sigma) = w(S^\sigma)=  w(\{e \in E \mid \sigma_e = 1\})$.
The partition function polynomial in~\eqref{eq:conditional-partition-function-polynomial} (without pinning) is defined as
\begin{align}\label{eq:partition-function-polynomial-general}
  Z(\bm{z}) = \sum_{S \subseteq E} w(S) \prod_{e \in S} z_e.
\end{align}

Note that the weight function $w$ in~\eqref{eq:w-weight-function-general} depends on the graph $G= (V,E)$ and the parameters $\boldsymbol{\lambda},\rho$.
Throughout this section, the parameter $\rho = \frac{\beta-1}{\gamma-1}$ is fixed. 
However, we may need to vary the graphs and external fields. 
We will use the following notation to highlight the dependence on the graph and the external fields:
\begin{align*}
 Z_G(\bm{z}; \bm{\lambda}) = \sum_{S \subseteq E} \prod_{v \in V} (1+\lambda_v(-\rho)^{\deg_{S}(v)}) \prod_{e \in S} z_e,
\end{align*}
where the parameter $\rho = \frac{\beta-1}{\gamma-1}$.

\Cref{thm:stability-of-partition-function} is proved in two steps. 
First, we reduce a general graph to a graph without \textit{bad} leaf vertex. Then we prove stability for such reduced graphs. 
To state the two lemmas precisely, we introduce some notation.
For a real interval $I$, define its $\delta$-strip by
$$\mathcal S_{\delta}(I):=\{a\in\mathbb C:\operatorname{dist}(a,I)<\delta\}.$$ 
Recall that 
$$U_\delta(x) := \{a \in \mathbb{C}:\operatorname{dist}(a,x)<\delta \}$$ 
is the disk of radius $\delta$ centered at $x$.

Throughout the next two lemmas, fix constants $\gamma > \beta > 1$ and $0 < \lambda < \lambda^\star(\beta, \gamma)$, and set
\[
\rho = \frac{\beta-1}{\gamma-1},
\qquad
\theta = \frac{(\gamma-1)^2}{\beta\gamma-1},
\qquad
I=[-\rho\lambda,\lambda], \quad I_-= [-\rho\lambda,0].
\]

\begin{condition}\label{condition:Glambda}
Let $\delta_1>0$ be a constant.
Let $G= (V,E)$ be a graph and $\boldsymbol{\lambda} = (\lambda_v)_{v \in V} \in \mathbb{C}^V$. It holds that for all vertices $v \in V$, $\lambda_v \in \mathcal{S}_{\delta_1}(I)$, furthermore, if $\deg_{G}(v)=1$, then $\lambda_v\in \mathcal{S}_{\delta_1}(I_-)$.
\end{condition}

\begin{lemma}[Reduction lemma]\label{lem:reduction-lemma}
  For any constant $\delta_1 > 0$, there exists a constant $\delta_2 = \delta_2(\beta,\gamma,\lambda,\delta_1) > 0$ such that the following holds.
  Let $G = (V,E)$ be a graph and $\boldsymbol{\lambda} = (\lambda_v)_{v \in V} \in \mathbb{R}^V$ be parameters such that $-\rho \lambda \leq \lambda_v \leq \lambda$ for all $v \in V$.
  For any vector  $\boldsymbol{z} = (z_e)_{e \in E}$ with $z_e \in U_{\delta_2}(\theta)$ for all $e \in E$, there exists a subgraph $G'=(V',E')$ of $G$, and complex-valued parameters $\boldsymbol{\lambda}' = (\lambda'_v)_{v \in V'} \in \mathbb{C}^{V'}$ such that $G'$ and $\boldsymbol{\lambda}'$ satisfy \Cref{condition:Glambda} with constant $\delta_1$ and
  \begin{align*}
      Z_{G'}(\boldsymbol{z}_{E'}; \boldsymbol{\lambda}') \ne 0 \textnormal{ if and only if } Z_{G}(\boldsymbol{z}; \boldsymbol{\lambda}) \ne 0.
    \end{align*}
\end{lemma}

\begin{lemma}[Stability of the bad-leaf-reduced core]\label{thm:core-zero-free}
  %Fix $\gamma>\beta>1$ and $0<\lambda<\lambda^\star$.  
  There exists $\delta_1=\delta_1(\beta,\gamma,\lambda)>0$ such that the following holds. Let $G=(V,E)$ and $\boldsymbol{\lambda} = (\lambda_v)_{v \in V} \in \mathbb{C}^V$ be a graph and a vector such that \Cref{condition:Glambda} holds with constant $\delta_1$.
  Then there exists a constants $\delta_3 = \delta_3(\beta,\gamma,\lambda,\delta_1)>0$ such that  for all $\boldsymbol{z} = (z_e)_{e \in E}$ with $z_e \in U_{\delta_3}(\theta)$ for all $e \in E$,
  \[
  Z_G(\bm{z}; \bm{\lambda}) \neq 0.
  \]
  \end{lemma}

Assuming the above two lemmas, we prove \Cref{thm:stability-of-partition-function}.
\begin{proof}[Proof of \Cref{thm:stability-of-partition-function}]
 Let $\delta_1 = \delta_1(\beta,\gamma,\lambda) > 0$ be the constant from \Cref{thm:core-zero-free}. Since \Cref{lem:reduction-lemma} holds for every $\delta_1>0$, we apply it with this value of $\delta_1$. Let $\delta_2 = \delta_2(\beta,\gamma,\lambda,\delta_1)$ be the constant from \Cref{lem:reduction-lemma}, and let $\delta_3 = \delta_3(\beta,\gamma,\lambda)$ be the constant from \Cref{thm:core-zero-free}. Define 
 \begin{align*}
 \delta=\min(\delta_2,\delta_3).
 \end{align*}
 Since $\delta_1$ depends only on $\beta,\gamma,\lambda$, so does $\delta$.

 Fix a graph $G = (V,E)$ and a real vector $\boldsymbol{\lambda} = (\lambda_v)_{v \in V} \in \mathbb{R}^V$ such that $-\rho \lambda \leq \lambda_v \leq \lambda$ for all $v \in V$.
 Let $\boldsymbol{z} = (z_e)_{e \in E}$ satisfy $z_e \in U_{\delta}(\theta)$ for all $e \in E$. Since $\delta\le \delta_2$, \Cref{lem:reduction-lemma} produces a reduced graph $G'=(V',E')$ and fields $\boldsymbol{\lambda}'$ with $\lambda'_v\in \mathcal S_{\delta_1}(I)$ for all $v\in V'$ and $\lambda'_v \in \mathcal S_{\delta_1}(I_-)$ for all $v\in V'$ with $\deg_{G'}(v)=1$, such that it is enough to prove $Z_{G'}(\boldsymbol{z}_{E'};\boldsymbol{\lambda}')\ne0$. Since $\delta\le \delta_3$ and $(G',\boldsymbol{\lambda}')$ satisfies \Cref{condition:Glambda} with constant $\delta_1$, this non-vanishing follows from \Cref{thm:core-zero-free}.

 Therefore $Z_G(\boldsymbol{z}; \boldsymbol{\lambda})\ne0$ for all $\boldsymbol{z}\in \prod_{e\in E}U_\delta(\theta)$, which proves the desired stability.
    %By \Cref{lemma:leaf-contraction} and our discussion above, we have that $Z_G^{\text{wsg}}(\bm{z},\bm{\lambda})=C(\bm{z},\bm{\lambda}) Z_{G'}^{\text{wsg}}(\bm{z},\hat{\bm{\lambda}})$ with $C(\bm{z},\bm{\lambda})$ stable for $z_e\in U_{\delta_2}(\theta)$. By \Cref{thm:core-zero-free}, we have that $Z_{G'}^{\text{wsg}}(\bm{z},\hat{\bm{\lambda}})$ is stable for $z_e\in U_{\delta_3}(\theta)$ for every $e\in E'$ for fixed $\hat{\bm{\lambda}}\in \prod_v \mathcal{S}_{\delta_1}(I)$. Therefore, $Z_G^{\text{wsg}}(\bm{z},\bm{\lambda})$ is stable for any $z_e\in U_{\delta_3}(\theta)$ and $\hat{\bm{\lambda}}\in \prod_v \mathcal{S}_{\delta_1}(I)$. Therefore, choosing $\delta=\min(\delta_2,\delta_3)$ proves the theorem.
  \end{proof}

  Next, we prove \Cref{lem:reduction-lemma} in \Cref{subsec:leaf-contraction-reduction} and \Cref{thm:core-zero-free} in \Cref{subsec:stability-of-local-polynomial}.

\begin{figure}[t]
  \centering
  \resizebox{0.8\linewidth}{!}{%
  \begin{tikzpicture}[
      x=0.82cm,y=0.82cm,
      line cap=round,line join=round,
      corev/.style={
          circle,
          draw=blue!45!black,
          fill=blue!30!black!65,
          line width=1pt,
          minimum size=8.5mm,
          inner sep=0pt
      },
      treev/.style={
          circle,          
          draw=orange!85!black,
          fill=orange!18,
          line width=0.9pt,
          minimum size=6.5mm,
          inner sep=0pt
      },
      badv/.style={
          circle,
          draw=green!50!black,
          fill=green!18,
          line width=0.9pt,
          minimum size=6.5mm,
          inner sep=0pt
      },
      halo/.style={
          circle,
          draw=red!70!black,
          fill=red!15,
          line width=0.9pt,
          minimum size=11.5mm,
          inner sep=0pt
      },
      coreedge/.style={draw=blue!20!black, line width=1.1pt},
      treeedge/.style={draw=green!55!black, line width=1pt},
      coreellipse/.style={draw=blue!20!black, line width=1pt},
      corelabel/.style={font=\bfseries\fontsize{20}{20}\selectfont, text=blue!20!black},
      treelabel/.style={font=\bfseries\fontsize{18}{18}\selectfont, text=green!55!black},
      sublab/.style={font=\huge}
  ]
  
  % =========================================================
  % Left panel : (a) General graph G
  % =========================================================
  \begin{scope}[shift={(0,0)}]
  
  % ----- core vertices -----
  \node[corev] (L1) at (0.0, 1.2) {};
  \node[corev] (L2) at (1.8, 1.9) {};
  \node[corev] (L3) at (3.8, 1.2) {};
  \node[corev] (L4) at (4.2,-0.8) {};
  \node[corev] (L5) at (0.1,-0.8) {};
  \node[corev] (L6) at (2.0,-2.0) {};
  
  % ----- core ellipse and labels -----
  % enlarged so that all core vertices are clearly enclosed
  \draw[coreellipse] (2.0,0.0) ellipse (3.55 and 2.65);
  \node[corelabel] at (2.0,3.85) {core};
  \node[treelabel] at (7.05,-4.45) {appended trees};
  
  % ----- core edges -----
  \draw[coreedge] (L1)--(L2);
  \draw[coreedge] (L2)--(L3);
  \draw[coreedge] (L3)--(L4);
  \draw[coreedge] (L4)--(L6);
  \draw[coreedge] (L6)--(L1);
  \draw[coreedge] (L1)--(L5);
  \draw[coreedge] (L5)--(L4);
  \draw[coreedge] (L2)--(L4);
  
  % ----- upper-left appended tree from L1 -----
  \node[treev] (Lul1) at (-1.2, 2.2) {};
  \node[treev] (Lul2) at (-2.3, 3.0) {};
  \node[treev] (Lul3) at (-3.3, 3.7) {};
  \node[badv]  (Lul4) at (-2.25, 1.75) {};
  \draw[treeedge] (L1)--(Lul1)--(Lul2)--(Lul3);
  \draw[treeedge] (L1)--(Lul4);
  
  % ----- lower-left appended tree from L5 -----
  \node[treev] (Lbl1) at (-1.4,-2.45) {};
  \node[treev] (Lbl2) at (-2.8,-1.85) {};
  \node[treev] (Lbl3) at (-4.15,-1.30) {};
  \node[badv]  (Lbl4) at (-2.75,-3.10) {};
  \node[treev] (Lbl5) at (-5.10,-3.65) {};
  \draw[treeedge] (L5)--(Lbl1);
  \draw[treeedge] (Lbl1)--(Lbl2)--(Lbl3);
  \draw[treeedge] (Lbl1)--(Lbl4)--(Lbl5);
  
  % ----- bottom appended tree from L6 -----
  \node[badv]  (Lbo1) at (1.85,-3.45) {};
  \node[treev] (Lbo2) at (0.55,-4.45) {};
  \node[treev] (Lbo3) at (-0.35,-5.20) {};
  \node[treev] (Lbo4) at (3.25,-4.40) {};
  \draw[treeedge] (L6)--(Lbo1);
  \draw[treeedge] (Lbo1)--(Lbo2)--(Lbo3);
  \draw[treeedge] (Lbo1)--(Lbo4);
  
  % ----- upper-right appended tree from L3 -----
  \node[treev] (Lur1) at (4.95, 2.15) {};
  \node[treev] (Lur2) at (6.25, 2.85) {};
  \node[treev] (Lur3) at (6.15, 1.70) {};
  \node[badv]  (Lur4) at (7.55, 2.05) {};
  \draw[treeedge] (L3)--(Lur1);
  \draw[treeedge] (Lur1)--(Lur2);
  \draw[treeedge] (Lur1)--(Lur3)--(Lur4);
  
  % ----- lower-right appended tree from L4 -----
  \node[treev] (Llr1) at (5.75,-1.55) {};
  \node[treev] (Llr2) at (7.10,-1.10) {};
  \node[treev] (Llr3) at (7.10,-2.55) {};
  \node[treev] (Llr4) at (8.45,-3.15) {};
  \draw[treeedge] (L4)--(Llr1);
  \draw[treeedge] (Llr1)--(Llr2);
  \draw[treeedge] (Llr1)--(Llr3)--(Llr4);
  
  \node[sublab] at (2.0,-6.8) {(a) General graph $G$.};
  
  \end{scope}
  
  % =========================================================
  % Right panel : (b) Graph G' after leaf contraction
  % =========================================================
  \begin{scope}[shift={(15.5,0)}]
  
  % ----- red halos FIRST (background), so blue nodes stay on top -----
  \node[halo] at (0.0, 1.2) {};
  \node[halo] at (4.2,-0.8) {};
  \node[halo] at (-1.55,-2.25) {};
  \node[halo] at (-3.10,-3.00) {};
  \node[halo] at (2.0,-3.45) {};
  \node[halo] at (5.45,2.35) {};
  
  % ----- core vertices on top of halo -----
  \node[corev] (R1) at (0.0, 1.2) {};
  \node[corev] (R2) at (1.8, 1.9) {};
  \node[corev] (R3) at (3.8, 1.2) {};
  \node[corev] (R4) at (4.2,-0.8) {};
  \node[corev] (R5) at (0.1,-0.8) {};
  \node[corev] (R6) at (2.0,-2.0) {};
  
  % ----- core ellipse and labels -----
  \draw[coreellipse] (2.0,0.0) ellipse (3.55 and 2.65);
  \node[corelabel] at (2.0,3.85) {core};
  \node[treelabel] at (6.9,-4.15) {appended trees};
  
  % ----- core edges -----
  \draw[coreedge] (R1)--(R2);
  \draw[coreedge] (R2)--(R3);
  \draw[coreedge] (R3)--(R4);
  \draw[coreedge] (R4)--(R6);
  \draw[coreedge] (R6)--(R1);
  \draw[coreedge] (R1)--(R5);
  \draw[coreedge] (R5)--(R4);
  \draw[coreedge] (R2)--(R4);
  
  % ----- left short branch -----
  \node[badv] (Ru0) at (-2.95,1.75) {};
  \draw[treeedge] (R1)--(Ru0);
  
  % ----- lower-left remaining branch -----
  \node[treev] (Rbl1) at (-1.55,-2.25) {};
  \node[badv]  (Rbl2) at (-3.10,-3.00) {};
  \draw[treeedge] (R5)--(Rbl1)--(Rbl2);
  
  % ----- bottom remaining branch -----
  \node[badv] (Rbo1) at (2.0,-3.45) {};
  \draw[treeedge] (R6)--(Rbo1);
  
  % ----- upper-right remaining branch -----
  \node[treev] (Rur1) at (5.45,2.35) {};
  \node[treev] (Rur2) at (6.85,1.75) {};
  \node[badv]  (Rur3) at (8.20,2.05) {};
  \draw[treeedge] (R3)--(Rur1)--(Rur2)--(Rur3);
  
  \node[sublab] at (2.0,-6.8) {(b) Graph $G'$ after bad-leaf deletion.};
  
  \end{scope}
  
  \end{tikzpicture}%
  }
  \caption{An illustration of the leaf-deletion reduction. (a) A general graph $G$ is decomposed into a core of minimum degree at least two and appended trees. In each appended tree, a green vertex $u$ denotes $\lambda_u<0$, while an orange vertex $u$ denotes $\lambda_u\ge 0$. 
  (b) After pruning bad-leaves in $G$, part of the appended trees is removed.
  Removed vertices are absorbed into effective external fields on the circled vertices.}
  \label{fig:leaf-contraction-overview}
  \end{figure}
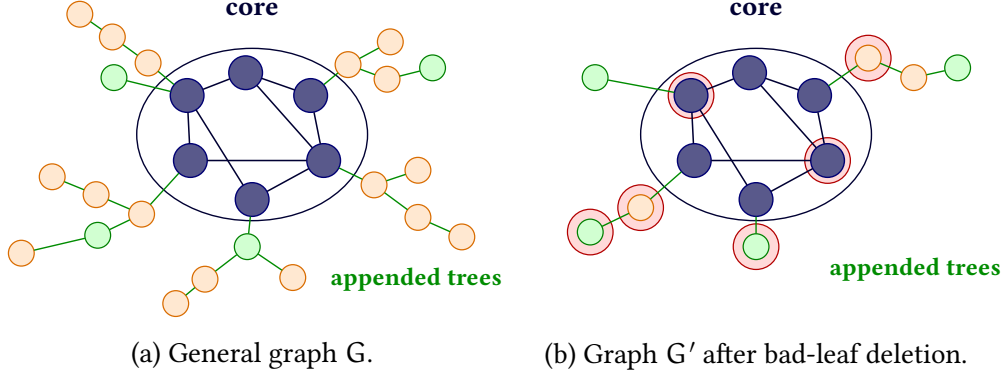

\subsection{Leaf-deletion reduction}\label{subsec:leaf-contraction-reduction} 

%The purpose of this subsection is to reduce a general graph to a graph of minimum degree at least two. Our approach is to recursively contract degree-one vertices and derive a equivalent new graph.  Each leaf contraction produces an equivalent instance, up to a nonzero multiplicative prefactor. Since all $z_e$ are complex-valued, the contraction operation may produce complex-valued external fields even when the initial field $\bm{\lambda}$ is real.  
%This brings some new technical challenges. The main solution to this challenge is that, we will design a reduction satisfying if the initial external fields lie in the real interval $[-\rho\lambda,\lambda]$, then all external fields generated by the leaf-contraction procedure remain inside a fixed $\delta$-strip containing $[-\rho\lambda,\lambda]$ in the complex plane. So we can just study the stability of the weighted subgraph-world partition function on the reduced graph.

  %Unlike in~\cite{GLL20}, 
%compared to the real-valued case in~\cite{GLL20}.
 %This is because the leaf-contraction operation is a linear operation on the external fields.

We give our reduction from general graph $G$ to the leaf-reduced graph $G'$. 
We first provide some intuition of the proof idea and then give the formal proof.
Let $v$ be a leaf (degree-one vertex) with unique neighbor $u$ and leaf edge $e=\{u,v\}$ in graph $G$.  
We call $u$ the \emph{parent} of $v$.
Suppose the current external field at the leaf is $\lambda_v $ and the current external field at the parent is $\lambda_u$.  
Assume that the external field $\lambda_v \notin \mathcal S_{\delta_1}(I_-)$, which violates \Cref{condition:Glambda}.
We will use a process to delete this bad leaf vertex $v$, and absorb the deleted part into the external field at the parent $u$.

Fix a configuration $S \subseteq E$ on $G$ such that $e \in S$. Let $S'=S\setminus \{e\}$. We write $t=(-\rho)^{\deg_{S'}(u)}$, where $\deg_{S'}(u)$ is the degree of the vertex $u$ in the subgraph $(V,S')$. 
For a configuration $S$ and $S'$,
the edge $e$ together with the two endpoints of the leaf edge $e$ contribute the following factor to the partition function polynomial~\eqref{eq:partition-function-polynomial-general}: 
\begin{align*}
\begin{cases}
    (1+\lambda_v)(1+\lambda_u t) 
    \qquad&\text{for configuration } S', \\
    z_e(1-\rho \lambda_v)(1-\rho \lambda_u t)
    \qquad&\text{for configuration } S.
\end{cases}
\end{align*}
The partition function is a sum over all configurations. Both $S$ and $S'$ are valid configurations. If we add the above  contributions of $u$ and $v$ in both cases together, we get
\begin{align*}
 &(1+\lambda_v)(1+\lambda_ut)+z_e(1-\rho \lambda_v)(1-\rho \lambda_u t)  \\
 &\quad=\bigl[(1+\lambda_v)+z_e(1-\rho \lambda_v)\bigr]
       +\lambda_ut\bigl[(1+\lambda_v)-\rho z_e(1-\rho \lambda_v)\bigr].
\end{align*}
Define the following two functions (one may view $a$ as $\lambda_v$ and $z$ as $z_e$):
\begin{equation}\label{eq:C-Phi}
  C_z(a)=(1+a)+z(1-\rho a),
  \qquad
  \Phi_z(a)=\frac{(1+a)-\rho z(1-\rho a)}{(1+a)+z(1-\rho a)}.
\end{equation}
Then, whenever $C_{z_e}(\lambda_v)\ne0$, it holds that
\begin{equation}\label{eq:leaf-factor}
 (1+\lambda_v)(1+\lambda_u t)+z_e(1-\rho \lambda_v)(1-\rho \lambda_u t)
 =C_{z_e}(\lambda_v)\left(1+\lambda_u t \Phi_{z_e}(\lambda_v)\right).
\end{equation}
Consider the subgraph $G'=G-\{v,e\}$ obtained from $G$ by removing the leaf vertex $v$. We can view the factor $\left(1+\lambda_u t \Phi_{z_e}(\lambda_v)\right)$ in the RHS of~\eqref{eq:leaf-factor} as the new factor contributed at the parent $u$ in the subgraph $G'$.
Therefore, we can write the partition function in the following way:
\begin{equation}\label{eq:leaf-contraction}
Z_G(\bm{z}; \bm{\lambda})=C_{z_e}(\lambda_v)Z_{G'}(\bm{z}; \bm{\lambda}'),
\end{equation}
where for all vertices in graph $G'$, the external fields are updated as follows:
\begin{align*}
\begin{cases}
    \lambda'_u=\lambda_u\Phi_{z_e}(\lambda_v)
    \qquad&\text{for parent } u, \\
    \lambda'_w=\lambda_w \quad &\text{for other vertices } w\neq u.
\end{cases}
\end{align*}
%$\lambda'_u=\lambda_u\Phi_{z_e}(\lambda_v)$ and $\lambda'_w=\lambda_w$ for other nodes $w\neq u$. 
Hence, we transfer the stability of $Z_G(\bm{z}; \bm{\lambda})$ to the stability of $Z_{G'}(\bm{z}; \bm{\lambda}')$ if $C_{z_e}(\lambda_v)\ne0$.

Our proof of \Cref{lem:reduction-lemma} is to apply the above deletion process in a systematic way. 
We introduce the following closure lemma, which plays a crucial role in our proof.
%Define the closure of a open region $\mathcal U$ as $\cl{\mathcal U}$, which is the smallest closed set containing $\mathcal U$.
Set
\[
I_+=[0,\lambda],
\qquad
I_-=[-\rho\lambda,0],
\qquad
I=I_-\cup I_+.
\]

\begin{lemma}\label{lemma:leaf-contraction}
For any constant $\delta_1 > 0$, there exists a constant $\delta_2>0 = \delta_2(\beta,\gamma,\lambda,\delta_1)$, together with two regions $\mathcal U^+,\mathcal U^-\subset\mathbb C$ such that the following hold. Let $\mathcal U=\mathcal U^+\cup\mathcal U^-$.

\begin{enumerate}[label=\textnormal{(\roman*)},leftmargin=2.2em]
\item \label{item:region}  
% $I_+\subset \mathcal U^+$, $I_-\subset \mathcal U^-$, $\cl{\mathcal U}\subset \mathcal{S}_{\delta_1}(I)$, and $-1\notin \cl{\mathcal U}$.
$I \subset \+U \subset \mathcal{S}_{\delta_1}(I)$ and $\+U^- \subset \mathcal{S}_{\delta_1}(I_-)$.
 %Moreover, for any $a\in\cl{\mathcal U}$ and $z\in U_{\delta_2}(\theta)$, then $C_z(a)\ne0$, where $\theta=\frac{(\gamma-1)^2}{\beta\gamma-1}$.
\item \label{item:C_z} For any $a \in \+U^+$ and any $z \in U_{\delta_2}(\theta)$, $C_z(a)\ne0$.
\item \label{item:product} For any integer $k\geq 1$, any sequence $a_i \in {\mathcal U^+}$ and $z_i\in U_{\delta_2}(\theta)$ with $1\leq i\leq k$, it holds that
\[
  \forall c \in I, \quad c\prod_{i=1}^{k}{\Phi_{z_i}(a_i)}\in \mathcal U.
\]
%\[
%  c\prod_{i=1}^{k}{\Phi_{z_i}(a_i)}\in \mathcal U^+
%  \qquad\text{for every }c\in I_+,
%\]
%and
%\[
%  c\prod_{i=1}^{k}{\Phi_{z_i}(a_i)}\in \mathcal U^-
%  \qquad\text{for every }c\in I_-.
%\]
\end{enumerate}
\end{lemma}

\Cref{lemma:leaf-contraction} is proved in \Cref{sec:leaf-contraction-proof}. Assuming the lemma, we define the following leaf-deletion process.
Note that $U^- \subset  \mathcal{S}_{\delta_1}(I_-)$ is a safe region for \Cref{condition:Glambda}. We need to keep deleting leaf vertices with external fields in $U^+$. Formally, consider the following process.

Let $G=(V,E)$ be a graph and $\boldsymbol{\lambda} = (\lambda_v)_{v \in V} \in \mathbb{R}^V$ be a vector of external fields in \Cref{lem:reduction-lemma}.
%Next we repeat the following leaf-contraction process to reduce all the leaf vertices in the graph $G$.
\begin{itemize}
\item Set $\lambda'_v \gets \lambda_v$ for all $v \in V$.
\item Repeat the following process until $G$ does not have leaf (degree-one) vertex $v$ with $\lambda_v' \in \+U^+$:
\begin{itemize}
  \item Choose the leaf vertex $v$ with the smallest index such that $\lambda_v' \in \+U^+$ and denote its parent by $u$.
  \item Delete the edge $e=(u,v)$ and the leaf $v$, 
  %absorb the deleted leaf $v$ into the external field at the parent $u$, 
  and update the external field $\lambda_u'\leftarrow\lambda_u'\Phi_{z_e}(\lambda'_v)$.
\end{itemize}
\end{itemize}

After the whole process terminates,
we denote the partition function of $G'$ by $Z_{G'}(\bm{z},\bm{\lambda}')$. We refer to \Cref{fig:leaf-contraction-overview} for an illustration of the process.
Now, we are ready to prove \Cref{lem:reduction-lemma}.
%After the leaf-contraction process, we get a new graph $G'=(V',E')$ with no leaf vertices. Then by the updating rule of the external field, it holds that 
%\begin{align*}
%\forall u \in V', \quad \lambda'_u=\lambda_u\prod_{(u,v)\in E\setminus E'} \Phi_{z_e}(\lambda_v'),
%\end{align*}
%the structure of the external field of $G'$ will be for every $u\in V'$, 
%where $\lambda_v'$ denotes the external field at the leaf vertex $v$ at the moment when $v$ is deleted. Note that $\lambda_v'$ may be different from the original external field $\lambda_v$.

\begin{proof}[Proof of \Cref{lem:reduction-lemma}]
Fix $\delta_1>0$, and let $\delta_2>0$ and $\mathcal U=\mathcal U^+\cup\mathcal U^-$ be given by \Cref{lemma:leaf-contraction}. 
Also fix $\boldsymbol z$ with $z_e\in U_{\delta_2}(\theta)$ for all $e\in E$.
We run the deterministic leaf-deletion procedure described above, always deleting the leaf of smallest index among the leaves whose \emph{current} external field $\lambda_v' \in \mathcal U^+$. This determines a subgraph $G'=(V',E')$ of $G$. 
%By construction, every degree-one vertex that remains in $G'$ has current external field in $\mathcal U^-$.

We track the external fields and the partition function throughout the process.  During the process, let $\lambda^{(t)}_v$ denote the current external field of a vertex $v$ after $t$ leaf deletions. By \Cref{lemma:leaf-contraction} \ref{item:region}, $\lambda^{(0)}_v=\lambda_v\in I\subset \mathcal U$. We prove by induction on $t$ that the following two properties hold. 
\begin{itemize}
  \item The first property is that for every current vertex $u$, its current field has the form
  \begin{align}\label{eq:field-product-form}
  \lambda^{(t)}_u
  =
  \lambda_u\prod_{e=(u,v)\in D_t(u)}
  \Phi_{z_e}(a_{e}),
  \end{align}
  where $D_t(u)$ is the set of deleted edges that have been absorbed into $u$, and each $a_e$ is the current field of the deleted leaf neighbor $v$ at the moment when $e = \{u,v\}$ was deleted. Moreover, %each such 
  for any $e \in D_t(u)$, $a_e \in \mathcal U^+$, and $\lambda^{(t)}_u\in\mathcal U$.
  \item The second property is that for any $G_t$ and $\boldsymbol\lambda^{(t)}$, the partition function satisfies:
  \begin{align}\label{eq:reduction-induction-partition}
  Z_G(\boldsymbol z;\boldsymbol\lambda)
  =
  \left(\prod_{e \in E\setminus E(G_t)} C_{z_e}(a_e)\right)
  Z_{G_t}(\boldsymbol z_{E(G_t)};\boldsymbol\lambda^{(t)}).
  \end{align}
  where $E(G_t)$ is the set of edges in $G_t$.
\end{itemize}

Both claims are immediate at $t=0$. Suppose they hold after $t$ deletions, and let the next deleted leaf be $v$ with parent $u$ and leaf edge $e=(u,v)$. By the choice of the deletion step, the current field $\lambda^{(t)}_v$ lies in $\mathcal U^+$. Hence \Cref{lemma:leaf-contraction} \ref{item:C_z} gives
$C_{z_e}(\lambda^{(t)}_v)\ne0$. The one-step identity~\eqref{eq:leaf-contraction} gives
\[
Z_{G_t}(\boldsymbol z_{E(G_t)};\boldsymbol\lambda^{(t)})
=
C_{z_e}(\lambda^{(t)}_v)
Z_{G_{t+1}}(\boldsymbol z_{E(G_{t+1})};\boldsymbol\lambda^{(t+1)}),
\]
where the only changed field is
\[
\lambda^{(t+1)}_u
=
\lambda^{(t)}_u\Phi_{z_e}(\lambda^{(t)}_v).
\]
This proves the partition-function identity~\eqref{eq:reduction-induction-partition} at time $t+1$ and updates the field as in~\eqref{eq:field-product-form}.

It remains to check that the updated field $\lambda^{(t+1)}_u$ stays in $\mathcal U$. By the induction hypothesis, $\lambda^{(t)}_u$ can be written as
$
\lambda^{(t)}_u
=
\lambda_u\prod_{i=1}^{k}\Phi_{z_i}(a_i),
$
where $\lambda_u\in I$, each $a_i\in {\mathcal U^+}$, and each $z_i\in U_{\delta_2}(\theta)$. The new update appends one more factor with
$a_{k+1}=\lambda^{(t)}_v\in\mathcal U^+$ and $z_{k+1}=z_e\in U_{\delta_2}(\theta)$. \Cref{lemma:leaf-contraction} \ref{item:product} implies
\[
\lambda^{(t+1)}_u
=
\lambda_u\prod_{i=1}^{k+1}\Phi_{z_i}(a_i)
\in {\mathcal U},\quad \text{because } \lambda_u\in I.
\]
%\[
%\lambda^{(t+1)}_u
%=
%\lambda_u\prod_{i=1}^{k+1}\Phi_{z_i}(a_i)
%\in {\mathcal U^+},\quad \text{if } %\lambda_u\in I_+.
%\]
%\[
%\lambda^{(t+1)}_u
%=
%\lambda_u\prod_{i=1}^{k+1}\Phi_{z_i}(a_i)
%\in {\mathcal U^-},\quad \text{if } \lambda_u\in I_-.
%\]
All other current fields are unchanged, so the induction is complete.

At the end of the process, let $G'=G_T = (V',E')$ and let $\boldsymbol\lambda'=\boldsymbol\lambda^{(T)}$, where $T$ is the total number of leaf deletions. Using \Cref{lemma:leaf-contraction} \ref{item:region}, the induction gives $\lambda'_v\in\mathcal U\subset \mathcal S_{\delta_1}(I)$ for every $v\in V'$. Moreover, by the stopping rule, every degree-one vertex of $G'$ has external field in $\mathcal U^-\subset \mathcal S_{\delta_1}(I_-)$. It also gives
\[
Z_G(\boldsymbol z;\boldsymbol\lambda)
=
\prod_{e\in E\setminus E'} C_{z_e}(a_e)\,
Z_{G'}(\boldsymbol z_{E'};\boldsymbol\lambda'),
%\qquad
%C(\boldsymbol z;\boldsymbol\lambda)
%=
%\prod_{e\in E\setminus E'} C_{z_e}(A_e),
\]
where each $C_{z_e}(a_e)$ is nonzero by the above analysis.
Hence %$C(\boldsymbol z;\boldsymbol\lambda)\ne0$, and this implies that 
$Z_{G'}(\boldsymbol z_{E'};\boldsymbol\lambda')\ne0$ if and only if
$Z_G(\boldsymbol z;\boldsymbol\lambda)\ne0$. This proves the lemma.
\end{proof}

\subsection{Stability of the bad-leaf-reduced core}\label{subsec:stability-of-local-polynomial}

We now explain the strategy for proving \Cref{thm:core-zero-free}. In the proof, we use the Holant representation of the weighted subgraph model as in \Cref{sec:WSG}. 
Let $G = (V,E)$ be a graph. Consider the bipartite graph $H_G = (V_1 \uplus V_2, E_H)$ defined in \Cref{def:graph-HG}, where $V_1 = V$ and $V_2 = E$. For each edge $e = \{u,v\} \in E$, we add two half-edges $e_u = \{u,e\}$ and $e_v = \{v,e\}$ to $H_G$; we call $e_u$ and $e_v$ the half-edges of $e$. To simulate the WSG model, the signature function $f_v$ for $v \in V_1$ is defined as 
\begin{align}\label{eq:def-f-v}
f_v = [1+\lambda_v,1+\lambda_v(-\rho), \ldots, 1 + \lambda_v(-1)^{d_v}],
\end{align}
where $d_v = \deg_G(v)$ is the degree of $v$ in $G$. The signature function $g_e$ for $e \in V_2$ is defined as
\begin{align}\label{eq:def-g-e}
g_e = [1,0,0,z_e],
\end{align}
where $z_e$ is the variable of the edge $e$. 
Let $\+F = \{f_v\}_{v \in V_1}$ and $\+G = \{g_e\}_{e \in V_2}$. Then, it holds that
\begin{align}\label{eq:def-ZG-HG}
Z_G(\boldsymbol z;\boldsymbol\lambda) = \operatorname{Holant}(H_G; \+F \mid \+G)
\end{align}

We use the a contraction procedure to show the stability of the partition function for the WSG model. The contraction procedure is as follows:

First, we split every edge into two half-edges, so that each vertex $v$ contributes a local polynomial defined over a star graph. Specifically, for each vertex $v\in V$, let $S_v$ be the star graph rooted at $v$ with one leaf for each edge incident to $v$.  We define $H_0$ as the disjoint union
\begin{align}\label{eq:def-H0}
        H_0=\bigcup_{v\in V} S_v.
\end{align}
More formally, the initial bipartite graph $H_0$ is defined as follows.
\begin{definition}[initial bipartite graph]
Let $G=(V,E)$ be a graph. The left part of the initial bipartite graph $H_0$ is the vertex set $V_0=V$, and the right part is the set of half-edges $E_0=\{e_u,e_v \mid e = \{u,v\} \in E\}$. For each edge $e \in E$ in graph $G$, we add two edges $e_u = \{u,e_u\}$ and $e_v = \{v,e_v\}$ to $H_0$.
\end{definition}

By the definition, the left part of $H_0$ is the vertex set $V$, and the right part contains $2|E|$ vertices.
We then recover the original edge model by \emph{contracting} the half edges. 
Fix an ordering $e_1,\ldots,e_m$ of the edges of $G$.
Given the graph $H_{i-1}$, the next graph $H_{i}$ is defined by applying the Asano contraction \cite{Asa70}, defined as follows.

\begin{definition}[Asano contraction]\label{def:contraction-operation}
Let $e = e_{i}=\{u,v\}$ be the next edge to be contracted. The contraction operation is defined as follows: find two edges $\{u,e_u\}$ and $\{v,e_v\}$ in the bipartite graph $H_{i-1}$, and then contract the two vertices $e_u$ and $e_v$ to be one vertex $e$. The resulting graph is $H_{i}$.
\end{definition}

\begin{figure}[t]
  \centering
  \begin{tikzpicture}[
      x=1cm,y=1cm,
      >=Stealth,
      font=\small,
      leftv/.style={
          circle, draw=blue!75!black, fill=blue!12,
          very thick, minimum size=8mm, inner sep=0pt
      },
      rightv/.style={
          circle, draw=orange!85!black, fill=orange!15,
          very thick, minimum size=8mm, inner sep=0pt
      },
      rightvfade/.style={
          circle, draw=gray!55, fill=gray!8,
          thick, minimum size=7.5mm, inner sep=0pt
      },
      actedge/.style={draw=orange!85!black, line width=1.1pt},
      fadeedge/.style={draw=gray!55, line width=0.8pt},
      panel/.style={
          rounded corners=4pt,
          draw=gray!30,
          fill=gray!4,
          inner sep=8pt
      },
      lab/.style={font=\small},
      eqtxt/.style={font=\small},
      qlab/.style={font=\small, text=orange!85!black}
  ]
  
  %====================================================
  % Left panel: H_i
  %====================================================
  
  % left part
  \node[leftv] (u) at (0.8,  1.0) {$u$};
  \node[leftv] (v) at (0.8, -1.0) {$v$};
  
  % right part
  \node[rightvfade] (r1) at (3.2,  2.2) {};
  \node at (3.2,1.65) {$\vdots$};
  
  \node[rightv] (eu) at (3.2,  1.0) {$e_u$};
  \node[rightv] (ev) at (3.2, -1.0) {$e_v$};
  
  \node at (3.2,-1.45) {$\vdots$};
  \node[rightvfade] (r2) at (3.2, -2.2) {};
  
  % edges
  \draw[fadeedge] (u) -- (r1);
  \draw[actedge]  (u) -- (eu);
  \draw[actedge]  (v) -- (ev);
  \draw[fadeedge] (v) -- (r2);
  
  % q-labels near nodes (not on edges)
  \node[qlab, right=2pt of eu] {$q_{u,e}$};
  \node[qlab, right=2pt of ev] {$q_{v,e}$};
  
  % top/bottom labels
  \node[lab] at (2.0,3.2) {$H_i$};
  \node[eqtxt] at (2.0,-3.5)
  {$\begin{aligned}
    &P(q_{v,e},q_{u,e})=\\
    &A+Bq_{v,e}+Cq_{u,e}+Dq_{v,e}q_{u,e}
  \end{aligned}$};
  
  % fit box
  \begin{scope}[on background layer]
    \node[panel, fit=(u)(v)(r1)(r2)(eu)(ev)] (Lbox) {};
  \end{scope}
  
  %====================================================
  % Right panel: H_{i+1}
  %====================================================
  
  % left part
  \node[leftv] (uu) at (7.2,  1.0) {$u$};
  \node[leftv] (vv) at (7.2, -1.0) {$v$};
  
  % right part
  \node[rightvfade] (s1) at (9.6,  2.2) {};
  \node at (9.6,1.65) {$\vdots$};
  
  \node[rightv] (e) at (9.6, 0.0) {$e$};
  
  \node at (9.6,-1.45) {$\vdots$};
  \node[rightvfade] (s2) at (9.6, -2.2) {};
  
  % edges
  \draw[fadeedge] (uu) -- (s1);
  \draw[actedge]  (uu) -- (e);
  \draw[actedge]  (vv) -- (e);
  \draw[fadeedge] (vv) -- (s2);
  
  % q-label near node
  \node[qlab, right=3pt of e] {$q_e$};
  
  % top/bottom labels
  \node[lab] at (8.4,3.2) {$H_{i+1}$};
  \node[eqtxt] at (8.4,-3.2) {$Q(q_e)=A+Dq_e$};
  
  % fit box
  \begin{scope}[on background layer]
    \node[panel, fit=(uu)(vv)(s1)(s2)(e)] (Rbox) {};
  \end{scope}
  
  %====================================================
  % Middle arrow
  %====================================================
  
  \path let \p1 = (Lbox.east), \p2 = (Rbox.west) in
        coordinate (Mid) at ($(\p1)!0.5!(\p2)$);
  
  \draw[->, very thick]
      ($(Mid)+(-1.0,0)$) -- ($(Mid)+(1.0,0)$);
  
  \node[lab] at ($(Mid)+(0,0.65)$) {contract};
  \node[lab] at ($(Mid)+(0,0.25)$) {$e_u,e_v \mapsto e$};
  \node[lab] at ($(Mid)+(0,-0.25)$) {\Cref{def:contraction-operation}};
  
  \end{tikzpicture}
  
  \caption{One-step contraction from $H_i$ to $H_{i+1}$, by \Cref{def:contraction-operation} and \Cref{lem:contraction-one-step}.}
  \label{fig:contraction-Hi-Hi1}
  \end{figure}
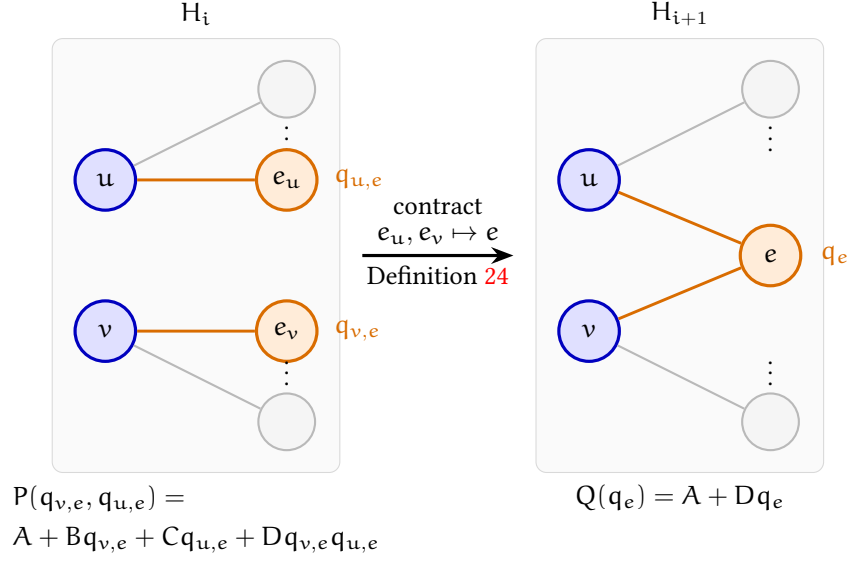

Using the above contraction operation, we construct a sequence of bipartite graphs
\[
H_0,H_1,\ldots,H_m, \quad \text{for } m = |E|,
\]
where $H_0$ is the disjoint union of the vertex stars and $H_m=H_G$ is the edge-incidence graph. 
Note that each bipartite graph $H_i$ in this sequence can be written as $H_i = (V_1^{(i)} \uplus V_2^{(i)}, E_{H}^{(i)})$, where $V_1^{(i)} = V$ and $V_2^{(i)} = \{e_j \mid j \leq i\} \cup \{e_u,e_v \mid e = e_j = \{u,v\} \land j > i\}$.

Given $H_0,H_1,\ldots,H_m$, we can define a sequence of polynomials $Z_{H_0},Z_{H_1},\ldots,Z_{H_m}$. 
However, due to technical reasons, we define a modified version instead as follows.

\begin{definition}[polynomial sequence]\label{def:modified_partition_function}
Let $H_i$ be a bipartite graph in the sequence. For each $v \in V$ in the left part, let $f_v$ in~\eqref{eq:def-f-v} be the signature function of $v$. For each $w \in V_2^{(i)}$ in the right part, define
\[
\widetilde g_w =
\begin{cases}
[q_w,1], & \text{if } w \text{ is an uncontracted half-edge vertex},\\
[q_w,0,0,1], & \text{if } w \text{ is a contracted edge vertex}.
\end{cases}
\]
Then, the \textit{modified partition function} $\widetilde Z_{H_i}$ is defined as follows:
\begin{align}\label{eq:def-ZHi}
\widetilde Z_{H_i} (\bm q; \bm\lambda) = \operatorname{Holant}(H_i; F^{(i)} \mid \widetilde G^{(i)}),
\end{align}
where $\+F^{(i)} = \{f_v\}_{v \in V_1^{(i)}}$ and $\widetilde G^{(i)} = \{\widetilde g_w\}_{w \in V_2^{(i)}}$.
\end{definition}

%At step $i$, if $e_i=\{u,v\}$, we apply the Asano contraction lemma (\Cref{lem:Asano-contraction} for details) to contract the two half-edge variables associated with $u$ and $v$ into a single edge variable. Intuitively, the contraction step is obtained by merging the two half-edge vertex $u_{e_i}$ and $v_{e_i}$ in $H_i$ to be one vertex in $H_{i+1}$.

In the definition above, uncontracted half-edge vertices use the reciprocal unary signature $[q_w,1]$, while contracted edge vertices use the reciprocal binary equality signature $[q_w,0,0,1]$.
Intuitively, this modification is the change of variables $q_w = 1/z_w$ applied to the signatures $[1,z_w]$ and $[1,0,0,z_w]$, respectively.
This is valid because the variables $z_w$ we consider are bounded away from $0$.
Note that $\widetilde Z_{H_m}(q;\lambda)
=
\left(\prod_{e\in E}q_e\right)
Z^{\text{wsg}}_{G}(1/q,\lambda)$ as defined in~\eqref{eq:def-ZG-HG}.
The proof of \Cref{thm:core-zero-free} has two ingredients: the stability of the initial local star polynomials, and the preservation of stability under each contraction operation. As a result, the final $\widetilde Z_{H_m}$ remains stable, as we desired.

Define the local polynomial of vertex $v$ with degree $d_v$ and external field $\lambda_v$. Here the variables $z_1,\ldots,z_{d_v}$ represent half-edge variables:
\begin{equation}\label{eq:local-poly}
    \begin{aligned}
    P_{v,d_v}(z_1,\ldots,z_{d_v})
    &\defeq \sum_{S\subseteq [d_v]} (1+\lambda_v (-\rho)^{|S|}) \prod_{i\in S} z_i\\
    &=
    \prod_{i=1}^{d_v}(1+z_i)
    +
    \lambda_v\prod_{i=1}^{d_v}(1-\rho z_i)
    \end{aligned}
  \end{equation}
Instead of studying $P_{v,d_v}(z_1,\ldots,z_{d_v})$, we consider a modified local polynomial by setting $q_i= \frac{1}{z_i}$. 
Define the modified local polynomial:
\begin{equation}\label{eq:reciprocal-local-poly}
  \widetilde{P}_{v,d_v}(q_1,\ldots,q_{d_v})
  \defeq
  \left(\prod_{i=1}^{d_v} q_i\right)
  P_{v,d_v}(1/q_1,\ldots,1/q_{d_v}) = \sum_{S\subseteq [d_v]} (1+\lambda_v (-\rho)^{|S|}) \prod_{i\notin S} q_i.
\end{equation}
%Equivalently,
%\begin{equation}\label{eq:reciprocal-local-poly-explicit}
%  \begin{aligned}
%  \widetilde P_{v,d_v}(q_1,\ldots,q_{d_v})
%  &= \sum_{S\subseteq [d_v]} (1+\lambda_v (-\rho)^{|S|}) \prod_{i\notin S} q_i\\
%  &=
%  \prod_{i=1}^{d_v}(q_i+1)
%  +
%  \lambda_v\prod_{i=1}^{d_v}(q_i-\rho)\\
%  &=\prod_{i=1}^{d_v}(q_i-\rho)\left(\prod_{i=1}^{d_v}\left(\frac{1+q_i}{q_i-\rho}\right)+\lambda_v\right),\qquad \text{assuming } q_i\neq \rho \text{ for all }i.
%  \end{aligned}
%\end{equation}
The modified partition function of $H_0$ is as given by \Cref{def:modified_partition_function}:

\begin{align*}
        \widetilde Z_{H_0}(\bm q;\bm\lambda)
        =
        \prod_{v\in V}\widetilde P_{v,d_v}\bigl((q_{v,e})_{e\ni v}\bigr) =\sum_{S\subseteq E_0} \prod_{e\notin S} q_e \prod_{v\in V} (1+\lambda_v (-\rho)^{\deg_S(v)}).
\end{align*}
Here $q_{v,e}$ denotes the half-edge variable at the copy of edge $e$ incident to $v$. As we discussed before, we first study the stability of $\widetilde Z_{H_0}(\bm q;\bm\lambda)$. Since it is a product of local polynomials, it suffices to establish stability of $\widetilde P_{v,d_v}(q_{v,1},\ldots,q_{v,d_v})$ for each vertex $v$.

We first motivate some definitions that will be essential for our proof. Assume for the moment that $q_{v,i}\neq \rho$ for all $1\leq i\leq d_v$. Then the local polynomial can be written as
\begin{align}
  \widetilde P_{v,d_v}(q_{v,1},\ldots,q_{v,d_v})
  &= \prod_{i=1}^{d_v}(q_{v,i}+1)
  +
  \lambda_v\prod_{i=1}^{d_v}(q_{v,i}-\rho)\label{eq:first-product}\\
  \text{(as we assume $q_{v,i}\neq \rho$ for all $i$)}\quad\quad&=\prod_{i=1}^{d_v}(q_{v,i}-\rho)\left(\prod_{i=1}^{d_v}\left(\frac{1+q_{v,i}}{q_{v,i}-\rho}\right)+\lambda_v\right).\label{eq:reciprocal-local-poly-explicit}
\end{align}
The ratio $\frac{1+q_{v,i}}{q_{v,i}-\rho}$ motivates the following M\"obius change of variables:
\begin{equation}\label{eq:psi-q}
  u=\psi(q):=-\frac{1+q}{q-\rho}.
\end{equation}
Under this assumption, if $\widetilde P_{v,d_v}(q_{v,1},\ldots,q_{v,d_v})=0$, then
$\prod_{i=1}^{d_v}u_{v,i}=(-1)^{d_v+1}\lambda_v$, where $u_{v,i}=\psi(q_{v,i})$.
Thus, to prove stability (zero-freeness) of the modified local polynomial, it suffices to rule out this product identity when the variables lie in the desired domain. We will choose a \emph{bad region} for the $u$-variables so that the identity can hold only when some $u_{v,i}$ lies inside this region. More precisely, the bad region is the open disk
\[
        D_\epsilon=D(c^\star,r^\star-\epsilon),
\]
where $\epsilon>0$ will be chosen later. The constants $c^\star$ and $r^\star$ are from~\cite{GLL20}:
\begin{equation}\label{eq:cstar-rstar}
        c^\star\defeq\frac{\gamma \log (\sqrt{\gamma/\beta})}{\sqrt{\gamma\beta-1}\arctan \sqrt{\gamma\beta-1}-\log (\sqrt{\gamma/\beta})},
        \qquad
        r^\star\defeq\sqrt{\left(c^\star+\frac{1}{\beta}\right)^2+\frac{\beta\gamma-1}{\beta^2}}.
  \end{equation}
Pulling this disk back under the map $\psi$, we define the corresponding bad region for the $q$-variables by
\begin{equation}\label{eq:Gamma-epsilon}
  \Gamma_\epsilon
  :=
  \{q\in\mathbb C:\ \psi(q)\in D(c^\star,r^\star-\epsilon)\}.
\end{equation}
The following lemma gives the desired stability of each local polynomial.

\begin{lemma}\label{lem:stability-of-local-polynomial}
There exist constants $\delta_1,\epsilon>0$, depending only on $\gamma,\beta,\lambda$, such that the following holds.
Let $G = (V,E)$ and $\bm{\lambda} \in \mathbb{C}^V$ satisfy \Cref{condition:Glambda} with constant $\delta_1$.
For any $v \in V$, the polynomial $\widetilde P_{v,d_v}(q_1,\ldots,q_{d_v})\neq0$ whenever $q_i\in\mathbb C\setminus \Gamma_\epsilon$ for all $i\in [d_v]$.
  \end{lemma}

The proof of \Cref{lem:stability-of-local-polynomial} is deferred to \Cref{sec:local-stability}.
Although the definition of $\Gamma_\epsilon$ was motivated by the case $q_{v,i}\neq \rho$ for all $i$, the lemma itself also covers the general case in which some $q_{v,i}$ may equal $\rho$.
Since $\widetilde Z_{H_0}(\bm q;\bm\lambda)$ is a product of local polynomials, \Cref{lem:stability-of-local-polynomial} implies the stability of $\widetilde Z_{H_0}(\bm q;\bm\lambda)$.

\iffalse
We next turn to the contraction step. %where we use the Asano contraction lemma (\Cref{lem:Asano-contraction}) to prove \Cref{thm:core-zero-free}. 
%Recall that we construct a sequence of bipartite graphs $H_0,H_1,\ldots,H_m$, where each $H_{i+1}$ is derived by forcing the two half-edge vertices $v_{e_i}$ and $u_{e_i}$ in $H_i$ to be one vertex in $H_{i+1}$. 

As we discussed above, the modified partition function of initial graph
\[
        \widetilde Z_{H_0}(\bm q;\bm\lambda)
        =
        \prod_{v\in V}\widetilde P_{v,d_v}\bigl((q_{v,e})_{e\ni v}\bigr).
\]
is already stable for each $q_{v,e}\in \mathbb{C}\setminus \Gamma_\epsilon$ with $v\in V$ and $e\in E_v$.  The final graph $H_m$ is the edge-incidence bipartite graph $H$ from \Cref{def:graph-HG}; its right-side variable corresponding to an edge $e$ will be denoted by $q_e$.  By \Cref{prop:spin-holant} and the change of variables $q_e=1/z_e$, the modified partition function on $H$ satisfies
\[
        \widetilde Z_H(\bm q;\bm\lambda)
        =
        \left(\prod_{e\in E}\frac1{z_e}\right)
        Z_G^{\text{wsg}}(\bm z,\bm\lambda).
\]
Thus stability of $\widetilde Z_H$ near $q_e=1/\theta$ implies stability of $Z_G^{\text{wsg}}$ near $z_e=\theta$, since $\theta\neq0$.

\fi
We next turn to the contraction step. 
For each polynomial $Z_{H_i}$ in the sequence, similarly, we define the modified partition function $\widetilde Z_{H_i}$ by replacing each variable $z$ with $q = 1/z$.
%We now construct $H$ from $H_0$ by contracting one original edge at a time.  Fix an ordering $e_1,e_2,\ldots,e_m$ of the edges of the leaf reduced core $G'$.  
Suppose that at step $i+1$ the edge is $e_{i+1}=\{u,v\}$.  The graph $H_i$ still has two half-edge variables $q_{u,e_{i+1}}$ and $q_{v,e_{i+1}}$ corresponding to the two copies of $e_{i+1}$.  The graph $H_{i+1}$ is obtained by replacing these two half-edge variables with a single edge variable $q_{e_{i+1}}$.  Equivalently, if $\bm q'$ denotes all variables not involved in this step, then we define two auxiliary polynomials $P$ and $Q$ as follows:
\begin{align*}
   P(q_{u,e_{i+1}},q_{v,e_{i+1}}) & \defeq
  \widetilde Z_{H_i}(\bm q', q_{u,e_{i+1}},q_{v,e_{i+1}};\bm\lambda);\\
  Q(q_{e_{i+1}})
  &\defeq
  \widetilde Z_{H_{i+1}}(\bm q',q_{e_{i+1}}; \bm\lambda),
\end{align*}
where $Q$ is the contracted version of $P$.
%From an algebraic point of view, this step is the Asano contraction of the two half-edge variables associated with $e_{i+1}$.
%we force the two half-edge variables associated with $e_{i+1}$ to be equal to derive $\widetilde Z_{H_{i+1}}(\bm q',\bm\lambda;q_{e_{i+1}})$.
 The next lemma gives a algebraic description of this one-edge contraction and shows that it preserves the required stability region.

\begin{lemma}\label{lem:contraction-one-step}
  Fix $\bm{q}'$ and $\bm{\lambda}$, and define coefficients $A,B,C,D$ depending on $\bm{q}'$ and $\bm{\lambda}$ such that
  \[
  {P}(q_{v_e}, q_{u_e}) = \widetilde Z_{H_i}(\bm{q}',q_{v_e}, q_{u_e}; \bm{\lambda})= A+Bq_{v_e}+Cq_{u_e}+Dq_{v_e}q_{u_e}.
  \]
  For every \(0<\epsilon<r^\star/2\), if ${P}(q_{v_e}, q_{u_e})\neq 0$ for $q_{v_e}, q_{u_e}\in \mathbb{C}\setminus \Gamma_\epsilon$, then there exists $\delta_3>0$, depending only on $\beta,\gamma,\epsilon$, such that 
  \[Q(q_e) = \widetilde Z_{H_{i+1}}(\bm{q}', q_e;\bm{\lambda})=A+Dq_e\]
   is zero-free for $q_e\in U_{\delta_3}(1/\theta)$.
\end{lemma}

We explain why $Q(q_e)$ can be written as $A+Dq_e$ in \Cref{lem:contraction-one-step}.
In the graph $H_i$, before contracting the edge $e=\{u,v\}$, the two half-edge variables $q_{v_e}$ and $q_{u_e}$ are independent. Hence, after fixing all other variables $\bm q'$ and $\bm\lambda$, the polynomial has the form
\[
  P(q_{v_e},q_{u_e})
  =
  A+Bq_{v_e}+Cq_{u_e}+Dq_{v_e}q_{u_e}.
\]
The four coefficients correspond to the four possible states of the two half-edges in reciprocal coordinates: $Dq_{v_e}q_{u_e}$ is the contribution where neither half-edge is selected, $Bq_{v_e}$ and $Cq_{u_e}$ are the two mixed contributions, and $A$ is the contribution where both half-edges are selected. The combinatorial contraction in \Cref{def:contraction-operation} merges the two half-edge vertices into one vertex $e$, whose equality signature allows only the two consistent states: both half-edges are not selected or both are selected. Therefore the two mixed contributions, represented algebraically by $Bq_{v_e}$ and $Cq_{u_e}$, disappear after contraction. 
The contracted signature $[q_e,0,0,1]$ assigns weight $q_e$ to the state where neither half-edge is selected and weight $1$ to the state where both half-edges are selected.
Thus the contracted polynomial is
\[
  Q(q_e)=A+Dq_e.
\]
The proof of the stability of $Q(q_e)$ in \Cref{lem:contraction-one-step} is deferred to \Cref{sec:contraction-stable}.
It relies essentially on an argument first given by Ruelle \cite{Rue71}.

Given \Cref{lem:stability-of-local-polynomial} and \Cref{lem:contraction-one-step}, we are now ready to prove \Cref{thm:core-zero-free}.

\begin{proof}[Proof of \Cref{thm:core-zero-free}.]
  Let \(\epsilon\) be the constant from \Cref{lem:stability-of-local-polynomial}.
  By shrinking \(\epsilon\) if necessary, we may assume that
  \(0<\epsilon<r^\star/2\); this only shrinks the stability domain
  \(\mathbb C\setminus\Gamma_\epsilon\). By \Cref{lem:stability-of-local-polynomial}, the modified partition function of the initial graph $H_0$
  \[
        \widetilde Z_{H_0}(\bm q;\bm\lambda)
        =
        \prod_{v\in V}\widetilde P_{v,d_v}\bigl((q_{v,e})_{e\ni v}\bigr).
\]
  is stable for each half-edge variable $q_{v,e}\in \mathbb{C}\setminus \Gamma_\epsilon$ for $v\in V$, $e\in E_v$. Now we order the edges of \(G\) as $e_1,e_2,\ldots,e_m$.
  Starting from \(\widetilde Z_{H_0}\), apply Asano contraction successively from $H_i$ to $H_{i+1}$. \Cref{lem:contraction-one-step}, applied with this \(\epsilon\), shows that each contraction preserves
  stability for the newly created reciprocal edge variable $q_e$ in the disk $U_{\delta_3}(1/\theta)$, where \(\delta_3=\delta_3(\beta,\gamma,\epsilon)\). Since \(\epsilon\) depends only on \(\beta,\gamma,\lambda\), so does \(\delta_3\). Therefore, after all contractions, the resulting polynomial $\widetilde Z_{H_m}(\bm q;\bm\lambda)$ is stable whenever $q_e\in U_{\delta_3}(1/\theta)$ for every $e\in E$. Finally, by the construction of the $\bm q$-variables, the fully
  contracted polynomial satisfies
  \[
          \widetilde Z_{H_m}(\bm q;\bm\lambda)
          =
          \left(\prod_{e\in E}q_e\right)
          Z^{\text{wsg}}_{G}(1/\bm q,\bm\lambda).
  \]
  The factor \(\prod_{e\in E}q_e\) is nonzero in a small neighborhood of \(1/\theta\) since $\theta\neq 0$. Hence $\widetilde Z_{H_m}(\bm q;\bm\lambda)\neq 0$ implies $Z_{G}^{\text{wsg}}(\bm{z},\bm{\lambda})\neq 0$ for $z_e=1/q_e$. Shrinking the neighborhood if necessary, this gives
  stability of \(Z^{\text{wsg}}_{G}(\bm z,\bm\lambda)\) whenever
  $z_e\in U_\delta(\theta)$ for every $e\in E$ for some $\delta>0$ depending only on $\beta,\gamma,\lambda$. This proves the theorem.
\end{proof}

%\iffalse
\subsection{Proof of the closure lemma}\label{sec:leaf-contraction-proof}

This section is devoted to the proof of \Cref{lemma:leaf-contraction}. Recall that by our analysis in \Cref{subsec:leaf-contraction-reduction}, we have that
\[
  C_z(a)=(1+a)+z(1-\rho a),
  \qquad
  \Phi_z(a)=\frac{(1+a)-\rho z(1-\rho a)}{(1+a)+z(1-\rho a)}.
\]
and the partition function after contracting one leaf edge $e=\{u,v\}$ with leaf vertex $v$ satisfies: 
\[
Z_G(\bm{z}; \bm{\lambda})=C_{z_e}(\lambda_v)Z_{G\setminus\{v,e\}}(\bm{z}; \bm{\lambda}'),
\]
where $\lambda'_u=\lambda_u\Phi_{z_e}(\lambda_v)$ and $\lambda'_w=\lambda_w$ for other nodes $w\neq u$. 

Fix a constant $\delta_1 > 0$ in \Cref{lemma:leaf-contraction}. 
First, we prove \Cref{lemma:leaf-contraction} \ref{item:region}. To this end, we will define the region $\+U^+$ and $\+U^-$ satisfying the conditions in the Lemma. Let $I_+=[0,\lambda]$ and $I_-=[-\rho\lambda,0]$, where $0<\lambda<\lambda^\star$. We use a potential function from~\cite{GuoL18}:
\[
  h_+(x)=\max\left\{t,\ x\log\frac{\zeta}{x}\right\},
  \qquad 0\le x\le\lambda.
\]
The specific choice of $\zeta,t$ can be found in~\cite{GuoL18}.
In particular, the following holds.

\begin{lemma}[Equation (10) of \cite{GuoL18}\protect\footnote{The results in \cite{GuoL18} actually work for $\lambda$ up to the threshold $\lambda_c>\lambda^\star$.}]\label{lem:guo-lu-one-variable}
  There exist constants $\zeta \in (\lambda,\lambda^\star)$ and $t>0$, depending only on $\beta,\gamma,\lambda$, such that the following holds for a constant $\kappa\in(0,1)$:
  \begin{equation}\label{eq:guo-lu-one-variable}
    \frac{(\beta\gamma-1)h_+(x)}
    {(1+\beta x)(\gamma+x)}
    \le
    \kappa\log\frac{\gamma+x}{1+\beta x}
    \qquad \forall\, 0\le x\le\lambda.
  \end{equation}
\end{lemma}

Next we define the region $\+U$ in the Lemma. For $\eta>0$, define the positive tube and the negative tube by
\begin{align}\label{eq:U-eta-decomposition}
  \mathcal U_\eta^+
  :=
  \bigcup_{x\in[0,\lambda]}D(x,\eta h_+(x)) \quad\text{ and }\quad
   \mathcal U_\eta^-
  :=
  \bigcup_{x\in[0,\rho\lambda]}D(-x,\eta h_+(x)),
\end{align}
where $D(x,r)$ is the open disk centered at $x$ with radius $r$.
See \Cref{fig:tube-illustration} for an illustration.
We use $\+U^+ = \mathcal U_\eta^+$ and $\+U^- = \mathcal U_\eta^-$ for some small enough constant $\eta = \eta(\beta,\gamma,\lambda,\delta_1)>0$. The value of the constant $\eta$ will be determined later in the proof.

By choosing $\eta$ small enough with respect to $\delta_1$, the region $\mathcal U_\eta=\mathcal U_\eta^+\cup\mathcal U_\eta^-$ satisfies
\[
  I_+\subset\mathcal U_\eta^+,\qquad
  I_-\subset\mathcal U_\eta^-,
  \qquad
  {\mathcal U_\eta}\subset\mathcal S_{\delta_1}(I).
\]
This verifies the first item of the Lemma.

\begin{figure}[t]
  \centering
  \includegraphics[width=0.65\textwidth]{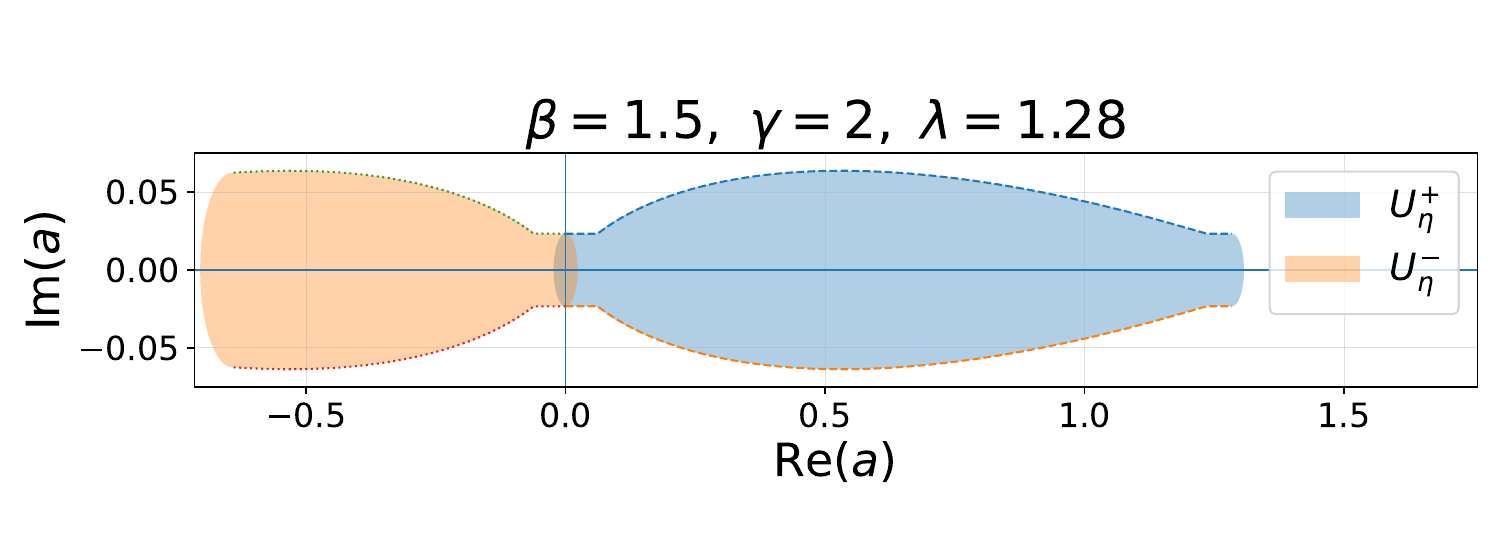}
  \caption{An illustration of the region $\mathcal U_\eta$ used in the closure lemma.}
  \label{fig:tube-illustration}
\end{figure}

Next, we prove \Cref{lemma:leaf-contraction} \ref{item:C_z}. This is guaranteed by continuity of $C_z(x)$ with respect to $z$ and $x$. When choosing $z=\theta$, we have that for any $0\le x\le\lambda$,
\[
  C_\theta(x)=(1+x)+\theta(1-\rho x)>C,
\]
for some constant $C>0$. This is because $1-\rho\theta=\frac{\beta+\gamma-2}{\beta\gamma-1}>0$. (Recall that $\gamma > \beta > 1$, $\rho=\frac{\beta-1}{\gamma-1}$, and $\theta=\frac{(\gamma-1)^2}{\beta\gamma-1}$.) Then compactness gives, after shrinking $\eta$ if necessary and taking a sufficiently small constant $\delta_0>0$, that
\[
  C_z(a)\ne0
  \qquad
  \text{for every }a\in {\mathcal U_\eta^+},\ z\in U_{\delta_0}(\theta).
\]
This verifies the second item of the Lemma.

Finally, we prove \Cref{lemma:leaf-contraction} \ref{item:product}. We introduce an auxiliary Lemma to prove this item. 

\begin{lemma}\label{lem:dist_separation}
 For all sufficiently small $\eta>0$, it holds that for any $c\in I$, $k\geq 1$, and $a_i\in \mathcal U_\eta^+$ with $1\le i\le k$,
  \[
  \textnormal{dist}\left(c\prod_{i=1}^k \Phi_\theta(a_i), \partial \mathcal U_\eta\right) \ge \Omega_{\beta,\gamma,\lambda}(\eta),
  \]
  where $\partial \mathcal U_\eta$ is the boundary of $\mathcal U_\eta$.
\end{lemma}

We first prove \Cref{lemma:leaf-contraction} \ref{item:product} via \Cref{lem:dist_separation} and then prove \Cref{lem:dist_separation}. Again, the property \ref{item:product} is guaranteed by continuity of $\Phi_z(a)$. For notational convenience, we define a function \( M(a):=\Phi_\theta(a)\). We now simplify this expression.  Recall that
$\rho=\frac{\beta-1}{\gamma-1}$ and $\theta=\frac{(\gamma-1)^2}{\beta\gamma-1}.$
Then, the numerator and denominator of $M(a)$ are given by:
\begin{align*}
  (1+a)-\rho\theta(1-\rho a)
  &=(1-\rho\theta)+(1+\rho^2\theta)a =\frac{\beta+\gamma-2}{\beta\gamma-1}(1+\beta a),\\
  (1+a)+\theta(1-\rho a)&=(1+\theta)+(1-\rho\theta)a =\frac{\beta+\gamma-2}{\beta\gamma-1}(\gamma+a).
\end{align*}
The common factor \((\beta+\gamma-2)/(\beta\gamma-1)\) is nonzero under our
assumption \(\gamma > \beta > 1\).  Hence
\[
  M(a)=\frac{1+\beta a}{\gamma+a}.
\]
We remark that $M(a)$ has the same form as the factor on one edge $e$ contributed by the tree recursion in~\cite{GuoL18}.
In particular, on the real interval \(a\in[-1,1/\rho]\), the denominator is
nonzero since \(\gamma+a\ge\gamma-1>0\).
We have the following properties of $M(a)$:
\[
M(-1)=-\rho\geq -1,\quad M(0)=\frac{1}{\gamma}>0,\quad M(1/\rho)=1,
\]
\[
\text{and} \quad M'(a)=\frac{\beta\gamma-1}{(\gamma+a)^2}>0 \quad \text{for } a\in (-1,1/\rho).\]
By the fact that $\lambda<\lambda^\star<1/\rho$, there is a constant $q<1$ such that
\[
  0<M(x)\le q<1
  \qquad\text{for every }x\in[0,\lambda].
\]
By the definition of $\Phi_z(a)$, we have that $\Phi_z(a)$ is analytic for $a\in \mathcal U_\eta^+$ and for $z$ near $\theta$. Thus $\Phi_z(a)$ is uniformly close to $M(a)$:
\[
  \Phi_z(a)=M(a)+r_z,\quad |r_z|\le C|z-\theta|\text{ for a constant }C>0.
\]
Moreover, for sufficiently small $\eta$ and $\delta$, there exists $q_1<1$ such that
\[
  |\Phi_z(a)|\le q_1\quad\text{and}\quad |M(a)|\le q_1
  \qquad
  (a\in {\mathcal U_\eta^+},\ z\in U_{\delta}(\theta)).
\]
For each $i$, write
\[
  \Phi_{z_i}(a_i)=M(a_i)+r_{z_i},
  \qquad |r_{z_i}|\le C\delta.
\]
Then the difference of the two products can be expanded as
\[
  \prod_{i=1}^k\Phi_{z_i}(a_i)-\prod_{i=1}^kM(a_i)
  =
  \sum_{i=1}^k
  \left(\prod_{j<i}\Phi_{z_j}(a_j)\right)
  r_{z_i}
  \left(\prod_{j>i}M(a_j)\right).
\]
Therefore, uniformly in $k$,
\[
  \left|
  \prod_{i=1}^k\Phi_{z_i}(a_i)-\prod_{i=1}^kM(a_i)
  \right|
  \le
  C\delta\sum_{i=1}^k q_1^{k-1}
  \le
  C\delta\, kq_1^{k-1}
  \le C'\delta,
\]
Since $\sup_k kq_1^{k-1}<\infty$ and $h_+$ is lower bounded by constant, by \Cref{lem:dist_separation}, choosing $\delta\ll\eta$ makes this perturbation smaller than the strict margin obtained above.  Thus the conclusions for $z_i=\theta$ persist for all $z_i\in U_\delta(\theta)$.  Taking $\mathcal U^+=\mathcal U_\eta^+$, $\mathcal U^-=\mathcal U_\eta^-$, and $\delta_2=\delta$ proves \Cref{lemma:leaf-contraction}  \ref{item:product}.

The remaining part of this section is devoted to the proof of \Cref{lem:dist_separation}. 
We write each $a_i$ in $\mathcal U_{\eta}^+$ as 
\begin{align}\label{eq:a-i-decomposition}
  a_i=x_i+\xi_i\in\mathcal U_\eta^+,
  \qquad
  x_i\in[0,\lambda],
  \qquad
  |\xi_i|< \eta h_+(x_i).
\end{align}
Since $M$ is analytic in a complex neighborhood of $[0,\lambda]$, Taylor expansion gives, uniformly in $i$ and $x_i$,
\[
  M(a_i)=M(x_i)+M'(x_i)\xi_i+r_i, \quad \text{where } |r_i| = O_{\beta,\gamma,\lambda}(|\xi_i|^2).
\]
Since $0<M(x_i)<1$ on $[0,\lambda]$, expanding the product gives
\begin{align}\label{eq:product-expansion}
  \prod_{i=1}^kM(a_i)
  =
  \prod_{i=1}^kM(x_i)
  +
  \sum_{i=1}^k(M'(x_i)\xi_i)\prod_{j\ne i}M(x_j)
  +
  E_k,
\end{align}
where $E_k$ is the error term. We need the following lemma to bound the first-order term.

\begin{lemma}\label{lem:positive-potential-contraction}
  There exists $\kappa\in (0,1)$ such that, for every $k\ge1$, every $c\in[0,\lambda]$, and every $x_i\in[0,\lambda]$,
  \begin{equation}\label{eq:positive-potential-contraction}
    c\sum_{i=1}^k
    M'(x_i)h_+(x_i)\prod_{j\ne i}M(x_j)
    \le
    \kappa\,
    h_+\!\left(c\prod_{i=1}^kM(x_i)\right).
  \end{equation}
  \end{lemma}

  \begin{proof}
  Since $0 < M(x_i) < 1$ and $0 \leq c \leq \lambda$, we have $c \prod_{i=1}^kM(x_i) \in [0, \lambda]$, which is in the domain of $h_+$.
  If $c=0$, then \eqref{eq:positive-potential-contraction} is trivial.  Hence assume $c>0$.  Let
  \[
    F=c\prod_{i=1}^kM(x_i).
  \]
  We remark that $F$ is exactly the tree recursion function appeared in~\cite{GuoL18}.
  By $M(x)\in (0,1)$ for $x\in [0,\lambda]$, it holds that $0<F\le c\le\lambda$ and
  \[
    \frac{M'(x)}{M(x)}
    =
    \frac{\beta\gamma-1}{(1+\beta x)(\gamma+x)}.
  \]
  By \Cref{lem:guo-lu-one-variable},
  \[
    \frac{M'(x_i)}{M(x_i)}h_+(x_i)
    \le
    \kappa\log\frac1{M(x_i)}.
  \]
  Therefore
  \[
    c\sum_iM'(x_i)h_+(x_i)\prod_{j\ne i}M(x_j)
    =
    F\sum_i\frac{M'(x_i)}{M(x_i)}h_+(x_i)
    \le
    \kappa F\sum_i\log\frac1{M(x_i)}
    =
    \kappa F\log\frac{c}{F}.
  \]
  Since $h_+(y)\ge y\log(\zeta/y)$ and $0<F\le c\le\lambda<\zeta$, we have
  \[
    F\log\frac{c}{F}\le h_+(F). \qedhere
  \]
  %This proves the Lemma.
  \end{proof}

We next bound the error term in the following lemma.

\begin{lemma}\label{lem:positive-product-error}
There exists a constant $C_1$ such that, for any integer $k\ge1$, any $c\in[0,\lambda]$, it holds that
\[
  c|E_k|
  \le
  C_1\eta^2 h_+\!\left(c\prod_{i=1}^kM(x_i)\right),
\]
where $E_k$ is the error term defined in \eqref{eq:product-expansion}.
\end{lemma}

\begin{proof}
To justify this bound, write
\[
  M(a_i)=m_i+\ell_i+r_i,
  \qquad
  m_i=M(x_i),\quad \ell_i=M'(x_i)\xi_i,
\]
with $|r_i|\le C|\xi_i|^2$ for some constant $C>0$. Since $h_+$ is bounded from above on $[0,\lambda]$, we have
$|\ell_i|\le C\eta$ and $|r_i|\le C\eta^2$ by letting constant $C$ be sufficiently large. Also,
$0<m_i\le q<1$ uniformly on $[0,\lambda]$ for some $q<1$. Choose $\eta$ small enough so that
$q+C\eta<q_1<1$ for some $q_1<1$.   We view the product
$
  \prod_{i=1}^k(m_i+\ell_i+r_i)
$
as a polynomial in $\ell_i$ and $r_i$ with coefficients depending on $m_i$.
Expanding the product and subtracting the constant term and the terms linear in $\ell_i$'s, the remaining terms are of two types.  All the terms with one $r_i$ have total absolute value at most
\[
  C\eta^2k q_1^{k-1}\le C'\eta^2, \text{ because } q_1 < 1.
\]
 Similarly, all the terms with at least two factors from the collection $\{\ell_i,r_i\}$ have total absolute value at most
  $\sum_{r=2}^k \binom{k}{r}(C\eta)^r q^{k-r}$.  
Indeed, for \(r\ge2\),
 $ \binom{k}{r}
  =
  \frac{k(k-1)}{r(r-1)}
  \binom{k-2}{r-2}
  \le
  k^2\binom{k-2}{r-2}$.
Therefore, using \(q+C\eta<q_1\),
\[
\begin{aligned}
  \sum_{r=2}^k \binom{k}{r}(C\eta)^r q^{k-r}
  &\le
  C^2\eta^2 k^2
  \sum_{r=2}^k
  \binom{k-2}{r-2}(C\eta)^{r-2}q^{k-r}  \\
  &=
  C^2\eta^2 k^2(q+C\eta)^{k-2}
  \le
  C^2\eta^2 k^2q_1^{k-2}
  \le
  C''\eta^2,
\end{aligned}
\]
where the last inequality follows from $q_1 < 1$.
Hence $|E_k|\le C_0\eta^2$ uniformly in $k$. Since $c\le\lambda$ and $h_+$ is lower bounded by a positive constant on $[0,\lambda]$, this uniform bound implies \Cref{lem:positive-product-error}.
\end{proof}

Finally, we use \Cref{lem:positive-potential-contraction} and \Cref{lem:positive-product-error} to prove \Cref{lem:dist_separation}. First suppose $c\in I_+$. Define
\[
 F_0= c \prod_{i=1}^kM(x_i)\in I_+,
\]
where each $a_i = x_i + \xi_i$ is as in \eqref{eq:a-i-decomposition}.
Using \Cref{lem:positive-potential-contraction} and $|\xi_i|\le \eta h_+(x_i)$, we get
\[
  \left|
  c\sum_{i=1}^k
  M'(x_i)\xi_i\prod_{j\ne i}M(x_j)
  \right|
  \le
  \eta\,c\sum_{i=1}^k
  M'(x_i)h_+(x_i)\prod_{j\ne i}M(x_j)
  \le
  \kappa\eta h_+(F_0).
\]
Together with \Cref{lem:positive-product-error}, this gives
\begin{equation}\label{eq:positive-product-error}
  \left|
  c\prod_{i=1}^kM(a_i)-F_0
  \right|
  \le
  (\kappa+C_1\eta)\eta h_+(F_0).
\end{equation}
Note that $\kappa < 1$.
Choosing $\eta$ small enough so that $\kappa+C_1\eta<1$, we obtain
\[
  c\prod_{i=1}^kM(a_i)\in D(F_0,\eta h_+(F_0)).
\]
Since $F_0\in[0,\lambda]$, this disk is contained in the union definition $\mathcal U_\eta^+$ of \eqref{eq:U-eta-decomposition}. Hence
\(
  c\prod_{i=1}^kM(a_i)\in\mathcal U_\eta^+
\).
In addition, by \eqref{eq:positive-product-error}, we have $\text{dist}(c\prod_{i=1}^kM(a_i),\partial \mathcal U_\eta^+)\ge \Omega(\eta)$ uniformly for $a_i\in\mathcal U_\eta^+$, since $h_+(F_0)$ is uniformly lower bounded by constant $t$. 
Similarly, for a negative field $-c\in I_-$ with $c\in[0,\rho\lambda]\subset[0,\lambda]$, the same estimate gives
\[
  \left|
  -c\prod_{i=1}^kM(a_i)-(-F_0)
  \right|
  \le
  (\kappa+C_1\eta)\eta h_+(F_0),
\]
where now $F_0=c\prod_iM(x_i)\in[0,\rho\lambda]$. Thus \(-c\prod_{i=1}^kM(a_i)\in\mathcal U_\eta^-\). Similarly, by the above inequality, we have $\text{dist}(-c\prod_{i=1}^kM(a_i),\partial \mathcal U_\eta^-)\ge \Omega(\eta)$. This finishes the proof of \Cref{lem:dist_separation}.

\subsection{Local stability}\label{sec:local-stability}

We prove  \Cref{lem:stability-of-local-polynomial} in this subsection.
%
%This subsection is devoted to the proof of \Cref{lem:stability-of-local-polynomial} and \Cref{lem:contraction-one-step}. First, we prove. 
The following two geometric facts about the disk will be used in our analysis; their proofs are deferred to Section~\ref{sec:other-proofs}. 
We define the closure of a open region $\mathcal U$ as $\cl{\mathcal U}$, which is the smallest closed set containing $\mathcal U$.

\begin{claim}\label{claim:cstar-rstar}
When $\gamma > \beta > 1$, it holds that $c^\star>0$ and $r^\star-c^\star>1$. Therefore, $D(c^\star,r^\star)$ contains $[-1,c^\star]$.
\end{claim}

\begin{claim}\label{claim:one-over-rho}
  It holds that $\frac1\rho \notin \cl{D(c^\star,r^\star)}$.
\end{claim}

By \Cref{claim:one-over-rho} and $D(c^\star,r^\star)$ being an open set, $\psi(0)=1/\rho\notin D(c^\star,r^\star-\epsilon)$ for any $\epsilon>0$. Hence $0\notin \Gamma_\epsilon$. 
Let  $K_\epsilon\defeq\mathbb C\setminus D(c^\star,r^\star-\epsilon)$, and then we have the following technical Lemma. 
Define the $k$th-power of a set $S \subseteq \mathbb{C}$ as $S^k:=\{\prod_{i=1}^k s_i: s_i\in S\}$.

\begin{lemma}\label{lem:lambda_avoid}
  Fix $0<\lambda<\lambda^\star$, and let $I=[-\rho\lambda,\lambda]$.
  There exist $\epsilon>0$ and $\delta'>0$, depending only on $\beta,\gamma,\lambda$, such that for every integer $d\ge 2$, it holds that
      \[
              (-1)^{d+1}K_\epsilon^{d}\cap \mathcal S_{\delta'}(I)
              =
              \emptyset.
      \]
\end{lemma}
Note that \Cref{lem:lambda_avoid} does not hold for $d=1$ and that is why we contract bad leaves first.

We first use \Cref{lem:lambda_avoid} to prove \Cref{lem:stability-of-local-polynomial} and then we prove \Cref{lem:lambda_avoid}.

\begin{proof}[Proof of \Cref{lem:stability-of-local-polynomial}]
 Let $\epsilon_0$ and $\delta_0$ be the two constants $\epsilon$ and $\delta'$ from \Cref{lem:lambda_avoid}.
 We will show \Cref{lem:stability-of-local-polynomial} for some constants $(\epsilon,\delta_1)$ such that $0 < \epsilon < \epsilon_0$ and $0 < \delta_1 < \delta_0$.
 We remark that since \Cref{lem:lambda_avoid} holds for  $\epsilon_0$ and $\delta_0$, it also holds for smaller constants $\epsilon$ and $\delta_1$.

  Assume first that $q_i\ne \rho$ for every $i$.  Then, by the definition $u_i=\psi(q_i)$, we have
  \begin{align*}
          \widetilde P_{v,d_v}(q_1,\ldots,q_{d_v})
          &=
          \prod_{i=1}^{d_v}(q_i-\rho)
          \left((-1)^{d_v}\prod_{i=1}^{d_v}u_i+\lambda_v\right).
  \end{align*}
  Thus a zero can occur only if $(-1)^{d_v+1}\prod_{i=1}^{d_v}u_i=\lambda_v$.
  Since $q_i\notin \Gamma_\epsilon$, we have $u_i\in K_\epsilon$ for every $i$. For $d_v\ge 2$, this contradicts \Cref{lem:lambda_avoid} because $\lambda_v \in \mathcal S_{\delta_1}(I)$. For $d_v=1$, by \Cref{condition:Glambda}, we have $\lambda_v\in \mathcal S_{\delta_1}((-\rho\lambda,0))$, and if zero can occur, we must have $u_1\in \mathcal S_{\delta_1}((-\rho\lambda,0))$. We show this cannot happen for sufficiently small $\epsilon$ and $\delta_1$. Note that the following two facts hold simultaneously.
  \begin{itemize}
    \item  $u_1\notin D(c^\star,r^\star-\epsilon)$ because $q_1 \notin \Gamma_\epsilon$.
    \item  By \Cref{claim:cstar-rstar}, $[-1,c^\star]$ is contained in $D(c^\star,r^\star)$, where $c^\star,r^\star > 0$ are constants. 
    Recall that $\lambda < \lambda^\star$. By \Cref{lem:lambda-v-tau}, $-\rho\lambda > -\rho \lambda^\star > -1$. Hence, both $-\rho \lambda$ and $0$ are inner points of  $D(c^\star,r^\star)$. For sufficiently small constants $\epsilon$ and $\delta_1$, it holds that $S_{\delta_1}((-\rho\lambda,0)) \subset D(c^\star,r^\star - \epsilon)$.
  \end{itemize}
  Combining these two facts, we deduce that $u_1\in \mathcal S_{\delta_1}((-\rho\lambda,0))$ is impossible.

  It remains to handle the case where $q_i=\rho$ for some $i$.  Then the second product in the RHS of \eqref{eq:first-product} vanishes. (Recall that~\eqref{eq:first-product} still holds for this case.)  The first product can vanish only if $q_j=-1$ for some $j$.  But $q_j=-1$ gives $\psi(q_j)=0$. Note that by \Cref{claim:cstar-rstar}, $0\in D(c^\star,r^\star-\epsilon)$ for sufficiently small constant $\epsilon$, and hence $q_j\in \Gamma_\epsilon$, contradicting the assumption $q_j \notin \Gamma_\epsilon$ in \Cref{lem:stability-of-local-polynomial}.  Therefore no zero occurs in this case either.
  \end{proof}

Next we prove \Cref{lem:lambda_avoid}. We need the following lemma from~\cite{GLL20}.
Recall that $K_\epsilon=\mathbb C\setminus D(c^\star,r^\star-\epsilon)$. 
Let $K \defeq \mathbb C\setminus D(c^\star,r^\star) \subset K_\epsilon$.

\begin{lemma}[\text{\cite[Lemma 8]{GLL20}}]\label{lem:lambda_GLL}
    For any $d\geq 2$ and any $0<\lambda<\lambda^\star$, 
        $(-1)^{d+1}K^{d}\cap [0,\lambda]=\emptyset.$
\end{lemma}

It is straightforward to extend \Cref{lem:lambda_GLL} to the following corollary.

\begin{corollary}\label{cor:lambda_GLL_strip}
    For any $d\geq 2$ and any $0<\lambda<\lambda^\star$, 
        $(-1)^{d+1}K^{d}\cap I =\emptyset$ for $I=[-\rho\lambda,\lambda]$.
\end{corollary}
\begin{proof}
    By \Cref{lem:lambda_GLL}, we only need to show that $(-1)^{d+1}K^{d}\cap [-\rho\lambda,0]=\emptyset$. For $a\in[-\rho\lambda,0]$, we have $|a|\le \rho\lambda<1$ due to \Cref{lem:lambda-v-tau}. By \Cref{claim:cstar-rstar}, we have $r^\star-c^\star>1$. Therefore, every $u\in K$ satisfies $|u|>1$. Every product in $K^d$ has modulus greater than $1$, which cannot equal any point of $(-1)^{d+1}[-\rho\lambda,0]$.
\end{proof}

%Using \Cref{lem:lambda_GLL} together with \Cref{claim:cstar-rstar}, \Cref{claim:one-over-rho}, we can prove that the region obtained after adding a small perturbation $\epsilon$ avoids $\mathcal S_{\delta'}(I)$ for some $\delta'>0$, i.e. \Cref{lem:lambda_avoid}. 
%Before we prove \Cref{lem:lambda_avoid}, we need some notations. We say that compact sets $A_\epsilon$ converge to $A$ in Hausdorff distance if $\operatorname{dist}_{\mathrm H}(A_\epsilon,A)\to 0$ as $\epsilon\to 0$.

Now we can prove \Cref{lem:lambda_avoid}.
Comparing to \Cref{cor:lambda_GLL_strip}, the difference is that we need to consider $K_{\epsilon}$, which is slightly larger than $K$, and we want to avoid a strip around the interval $I$.
This strengthening is necessary for \Cref{lem:contraction-one-step} later.
\begin{proof}[Proof of \Cref{lem:lambda_avoid}]
By \Cref{claim:cstar-rstar}, $r^\star-c^\star>1$.  Set constants
\[
        M:=\max\{\lambda,\rho\lambda\},\qquad
        m_0:=\frac{r^\star-c^\star+1}{2}>1,
        \qquad
        d_0:=\max\{2,\lceil \log_{m_0}(M+1)\rceil\}.
\]

Choose $\epsilon_0>0$ small enough that
$r^\star-c^\star-\epsilon_0>m_0 > 1$.  Then, for every
$0<\epsilon\le \epsilon_0$ and every $u\in K_\epsilon$,
$|u|\ge |u - c^\star| - c^\star \geq r^\star-c^\star-\epsilon\ge m_0$.
In particular, if $u_1,\ldots,u_d\in K_\epsilon$, then
$|u_1\cdots u_d|\ge m_0^d$.
If $d>d_0$, then $m_0^d>M+1$.  
Since $\mathcal S_1(I)\subseteq\{a\in\mathbb C: |a|<M+1\}$,
we have $(-1)^{d+1}K_\epsilon^d\cap \mathcal S_{1}(I)=\emptyset$ for $\delta_0 = 1$, any $d > d_0$, and any $0<\epsilon\le\epsilon_0$.

It remains to handle degrees $d$ between $2$ and $d_0$.
Let
\[
        R:=\max\{r^\star-c^\star,M+1\},
        \qquad
        L_\epsilon:=K_\epsilon\cap \cl{D(0,R)}.
\]
For $0<\delta\le1$, avoiding $\mathcal S_\delta(I)$ for products from
$L_\epsilon$ is enough.  Indeed, if $u_1,\ldots,u_d\in K_\epsilon$ and
some $u_i\notin L_\epsilon$, then $|u_i|>R$, while
$|u_j|\ge m_0 > 1$ for $j\ne i$.  Hence
$|\prod_{j=1}^d u_j|>R m_0^{d-1}\ge M+1$.
Thus $(-1)^{d+1}\prod_j u_j\notin \mathcal S_\delta(I)$, because
every point of $\mathcal S_\delta(I)$ has modulus less than $M+1$.
Consequently,
\begin{equation}\label{eq:lambda-avoid-2}
        (-1)^{d+1}L_\epsilon^d\cap \mathcal S_\delta(I)=\emptyset
        \quad\Longrightarrow\quad
        (-1)^{d+1}K_\epsilon^d\cap \mathcal S_\delta(I)=\emptyset .
\end{equation}

We now prove the left-hand side of \eqref{eq:lambda-avoid-2} for each
fixed $2\le d\le d_0$ and some $\delta_d$ depending on $d$. 
Note that the set $L_\epsilon$ is compact. We need some standard notations.
For a point $a\in\mathbb C$ and a nonempty set $S\subseteq\mathbb C$, write $\operatorname{dist}(a,S):=\inf_{z\in S}|a-z|$. For nonempty compact sets $A,B\subseteq\mathbb C$, the Hausdorff distance is defined by
        \[
        \operatorname{dist}_{\mathrm H}(A,B)
            :=
                \max\Bigl\{\sup_{a\in A}\operatorname{dist}(a,B),\ \sup_{b\in B}\operatorname{dist}(b,A)\Bigr\}.
        \]
The compact sets $L_\epsilon$ converge, as
$\epsilon\to0$, in Hausdorff distance to
$L_0:=K\cap \cl{D(0,R)}.$
Since multiplication is uniformly continuous on the fixed compact set
$\cl{D(0,R)}^d$, the product sets
$
        L_\epsilon^d=\{\prod_{i=1}^d u_i: u_i\in L_\epsilon\}
$
also converge in Hausdorff distance to $L_0^d$.  Moreover,
$L_0\subseteq K$, so \Cref{cor:lambda_GLL_strip} gives
$
        (-1)^{d+1}L_0^d\cap I=\emptyset .
$
Both sets are compact, and thus their distance is positive.  Write
\[
        \eta_d:=\dist\bigl(I,(-1)^{d+1}L_0^d\bigr)>0,
\]
where $\eta_d > 0$ is a positive constant depending on $\beta,\gamma,\lambda,d$.
Since $L_\epsilon^d$ converges in Hausdorff distance to $L_0^d$, there exists $\epsilon_d>0$ such that the following holds:
\[
        \operatorname{dist}_{\mathrm H}
        \bigl((-1)^{d+1}L_\epsilon^d,(-1)^{d+1}L_0^d\bigr)
        <\eta_d/3
        \qquad\text{whenever }0<\epsilon\le\epsilon_d .
\]
Taking $\delta_d:=\eta_d/3$, we get
\[
        (-1)^{d+1}L_\epsilon^d\cap \mathcal S_{\delta_d}(I)=\emptyset
        \qquad\text{for all }0<\epsilon\le\epsilon_d .
\]

Finally choose constants $\epsilon$ and $\delta'$ as follows:
$\epsilon:=\min\{\epsilon_0,\epsilon_2,\ldots,\epsilon_{d_0}\}$ and $\delta':=\min\{\delta_0,\delta_2,\ldots,\delta_{d_0}\}.$
The large-degree argument covers all $d>d_0$, while the finite-degree
argument and \eqref{eq:lambda-avoid-2} cover all $2\le d\le d_0$.
Therefore
        $(-1)^{d+1}K_\epsilon^d\cap \mathcal S_{\delta'}(I)=\emptyset$ for every $d\ge2$,
as required.
\end{proof}

\subsection{Contraction preserves stability}
\label{sec:contraction-stable}

Next we prove \Cref{lem:contraction-one-step}. We first state several auxiliary lemmas needed for the proof. Consider the single-variable polynomial $P(x)=A+2Bx+Cx^2$ with parameters $A,B,C\in \mathbb{C}$. Its corresponding polar form in two variables $x_1,x_2$ is defined as
\[
  \widehat{P}(x_1,x_2)=A+B(x_1+x_2)+Cx_1x_2.
\]
Let $\Gamma$ be an open connected region in $\mathbb{C}$. We say that $\Gamma$ is a circular region if it is a disk, or the complement of a disk in $\mathbb{C}$. The classical Grace--Szeg\H{o}--Walsh coincidence theorem has the following immediate consequence.

\begin{proposition}[\text{Grace--Szeg\H{o}--Walsh coincidence theorem \cite[Theorem 11]{BorceaBranden2009}}]\label{lem:GSW}
Let $\Gamma$ be a circular region and $C>0$ in $\widehat{P}(x_1,x_2)$. The polynomial $\widehat{P}(x_1,x_2)$ is $\Gamma \times \Gamma$-stable if and only if $P(x)$ is $\Gamma$-stable.
\end{proposition}

The next ingredient is the classical Asano contraction lemma \cite{Asa70,Rue71}. Here we use a simplified version for quadratic polynomials.

\begin{lemma}[Asano contraction]\label{lem:Asano-contraction}
Let $\Gamma\subset \mathbb{C}$ with $0\notin \Gamma$. If the multivariate polynomial
\[
P(x_1,x_2)=A+Bx_1+Cx_2+Dx_1x_2
\]
is stable for $(x_1,x_2)\in (\mathbb{C}\setminus \Gamma)\times (\mathbb{C}\setminus \Gamma)$, then
\[
Q(x):=A+Dx
\]
is stable for $x\in \mathbb{C}\setminus \left(-\Gamma^2\right)$, where $\Gamma^2:=\{x\cdot y: x,y\in \Gamma\}$.
\end{lemma}

\begin{proof}
  Since $0\notin \Gamma$, we have $A=P(0,0)\neq 0$. Without loss of generality, we assume $D\neq 0$, since otherwise the polynomial $Q$ is a constant and thus always stable. By the assumption on $P$, we have $\widetilde{P}(x)=A+(B+C)x+Dx^2$ is stable for $x\in \mathbb{C}\setminus \Gamma$. Thus $\widetilde{P}(x)=0$ can occur only when $x\in \Gamma$. Suppose the two roots of $\widetilde{P}(x)=0$ are $\zeta_1, \zeta_2 \in \Gamma$. By Vieta's formula, we have $\zeta_1\zeta_2=A/D$. On the other hand, $Q(x)=0$ can occur only for $x=-A/D\in -\Gamma^2$. Therefore, $Q(x)$ is stable for $x\in \mathbb{C}\setminus \left(-\Gamma^2\right)$.
\end{proof}

For our analysis, we need an auxiliary property of the disk $D(c^\star,r^\star-\epsilon)$, as well as an inequality for the parameters. The proofs of the following two results are deferred to \Cref{sec:other-proofs}.

\begin{lemma}\label{lem:disk-separation}
  Let $\zeta_1,\zeta_2$ be the two roots of $\beta x^2+2x+\gamma=0$. For any $0<\epsilon<r^\star/2$, it holds that $$\cl {D(c^\star,r^\star-\epsilon)}\cap \{\zeta_1, \zeta_2\}=\emptyset.$$
\end{lemma}

\begin{claim}\label{claim:rho-theta}
  It holds that $\rho \theta<1$.
\end{claim}

Finally, we are ready to prove \Cref{lem:contraction-one-step}.

\begin{proof}[Proof of \Cref{lem:contraction-one-step}.]
  Fix \(0<\epsilon<r^\star/2\). 
  Applying \Cref{lem:Asano-contraction} with the region \(\Gamma_\epsilon\), it
  suffices to show that a sufficiently small neighborhood of \(1/\theta\) is
  contained in \(\mathbb C\setminus(-\Gamma_\epsilon^2)\).

  By \Cref{claim:one-over-rho} and the continuity of \(\psi\) at \(0\), we have
  \(0\notin\cl{\Gamma_\epsilon}\), where $\cl{\cdot}$ denotes the closure of a set. Hence there is a constant \(a_0>0\) such
  that \(\Gamma_\epsilon\cap D(0,a_0)=\emptyset\). Set
  \[
        B_0:=\left|1/\theta\right|+1,
        \qquad
        R_0:=\frac{B_0}{a_0},
        \qquad
        \Gamma_{\epsilon,R}:=\Gamma_\epsilon\cap\cl{D(0,R_0)}.
  \]
  If \(q_1,q_2\in\Gamma_\epsilon\) and
  \(-q_1q_2\in U_1(1/\theta)\), then \(|q_1q_2|<B_0\), while
  \(|q_i|\ge a_0\) for \(i=1,2\). Thus \(|q_i|<R_0\) for \(i=1,2\), and hence
  \(q_1,q_2\in \Gamma_{\epsilon,R}\). Therefore, to find a neighborhood of
  \(1/\theta\) avoiding \(-\Gamma_\epsilon^2\), it suffices to show that
  $\frac1\theta\notin \cl{-\Gamma_{\epsilon,R}^2}$.
  Now, we are working with a bounded set $\Gamma_{\epsilon,R}$.

  Suppose, for contradiction, that
  \(1/\theta\in \cl{-\Gamma_{\epsilon,R}^2}\). Since
  \(\Gamma_{\epsilon,R}\) is bounded, the continuity of multiplication gives
  \(q_1,q_2\in\cl{\Gamma_{\epsilon,R}}\subseteq \cl{\Gamma_\epsilon} \) such that
  $
        1/\theta=-q_1q_2.
  $
  Since \(q_i\in\cl{\Gamma_\epsilon}\), we have \(q_i\ne\rho\), and the
 Möbius variables \(u_i:=\psi(q_i)\) lie in
  \(\cl{D(c^\star,r^\star-\epsilon)}\). Moreover \(u_i\ne -1\), since
  \(\psi(q)=-1\) has no finite solution. Solving
  \(u_i=-(1+q_i)/(q_i-\rho)\), we obtain for $i=1,2$ that
  $
        q_i=\frac{\rho u_i-1}{1+u_i}.
  $
Using \(\frac1\theta=-q_1q_2\), we get the following equation:
  \[
        \frac1\theta
        =
        -\frac{(1-\rho u_1)(1-\rho u_2)}
        {(1+u_1)(1+u_2)}.
  \]
Equivalently,
\[
        (1+u_1)(1+u_2)+\theta(1-\rho u_1)(1-\rho u_2)
        =
        0.
\]
Expanding and using the definitions of \(\theta = (\gamma-1)^2/(\beta\gamma-1)\) and \(\rho = (\beta-1)/(\gamma-1)\), this becomes
\[
        (1-\rho\theta)(\beta u_1u_2+u_1+u_2+\gamma)=0.
\]
By \Cref{claim:rho-theta}, \(1-\rho\theta\ne0\), and therefore
$
        \beta u_1u_2+u_1+u_2+\gamma=0.
$
Since \(u_1,u_2\in\cl{D(c^\star,r^\star-\epsilon)}
\subset D(c^\star,r^\star-\epsilon/2)\), \Cref{lem:GSW} applied to the circular
region \(D(c^\star,r^\star-\epsilon/2)\) implies that
$
        \beta x^2+2x+\gamma=0
$
for some \(x\in D(c^\star,r^\star-\epsilon/2)\). This contradicts
\Cref{lem:disk-separation}, applied with \(\epsilon/2\).

Thus \(1/\theta\notin \cl{-\Gamma_{\epsilon,R}^2}\). Hence there exists
\(0<\delta_3\le 1\) such that
$ U_{\delta_3}(1/\theta)\subseteq \mathbb C\setminus\cl{-\Gamma_{\epsilon,R}^2}$.
To finish the proof, we still need to go back to $\Gamma_\epsilon$ . 
Take any $z\in U_{\delta_3}(1/\theta)$. 
If $z=-q_1q_2$ for some $q_1,q_2\in\Gamma_\epsilon$, 
then $z\in U_1(1/\theta)$, since $\delta_3\le1$. 
By the choice of $R_0$, 
this forces $q_1,q_2\in\Gamma_{\epsilon,R}$, 
and hence $z\in-\Gamma_{\epsilon,R}^2$. 
This contradicts the choice of $\delta_3$. 
Therefore, $U_{\delta_3}(1/\theta)\subseteq\mathbb C\setminus(-\Gamma_\epsilon^2)$.
The construction of $\delta_3$ depends only on $\beta,\gamma,\epsilon$. 
This proves the lemma.
\end{proof}

\section{Simulated annealing reductions}\label{sec:simulated-annealing}

The sampling algorithm in \Cref{thm:sampling} implies an FPRAS for
$Z_G^{\text{spin}}(\beta,\gamma,\lambda)$ by annealing.
We could apply the adaptive procedure of \cite{SVV09}, but here a simple non-adaptive annealing for the external field $\lambda$ in the original spin system will do. %This avoids annealing
%the weighted subgraph model and gives a shorter second-moment calculation.
We include a proof for completeness.

Fix constants $\beta,\gamma,\lambda > 0$ such that $\gamma>\beta>1$ and
$\lambda<\lambda^\star$. Let $G=(V,E)$ and $m=|E|$. 
Assume $0<\epsilon\le 1$. Set
\[
    q=1-\frac{1}{5n}.
\]
Choose $\ell=O_\lambda(n\log(n/\epsilon))$ large enough so that, for the
sequence
\[
    \lambda_\ell=\lambda,\qquad
    \lambda_{i-1}=q\lambda_i \quad \text{for } i=1,\ldots,\ell,
\]
we have $\lambda_0\le \epsilon/(20n)$. For notational convenience, set
$\lambda_{\ell+1}=\lambda_\ell/q$. For $0\le i\le \ell+1$, write
\[
    Z_i=Z_G^{\text{spin}}(\beta,\gamma,\lambda_i),
    \qquad
    \mu_i=\mu_G^{\text{spin}}(\beta,\gamma,\lambda_i)
    \quad (0\le i\le \ell).
\]
The target partition function is $Z_\ell$, and
\[
    Z_\ell
    =
    \frac{Z_\ell}{Z_{\ell-1}}\cdot
    \frac{Z_{\ell-1}}{Z_{\ell-2}}\cdots
    \frac{Z_1}{Z_0}\cdot Z_0 .
\]

\begin{lemma}\label{lem:initial-partition-function}
For the above choice of $\ell$,
\[
    \gamma^m \le Z_0 \le e^{\epsilon/10}\gamma^m .
\]
\end{lemma}
\begin{proof}
The all-$1$ configuration contributes exactly $\gamma^m$, so $Z_0\ge \gamma^m$.
For any configuration $\sigma$ with $k=n_0(\sigma)$ vertices in spin $0$,
the edge contribution is at most $\gamma^m$, because $\gamma>\beta>1$.
Therefore
\[
    Z_0
    \le
    \gamma^m \sum_{k=0}^n \binom{n}{k}\lambda_0^k
    =
    \gamma^m(1+\lambda_0)^n
    \le
    e^{\epsilon/10}\gamma^m ,
\]
where the last inequality follows from $\lambda_0\le \epsilon/(20n)$.
%The lower bound $e^{-\epsilon/10}\gamma^m\le Z_0$ is immediate from
%$Z_0\ge \gamma^m$.
\end{proof}

To estimate the telescoping product, consider the following estimator:
\[
    W_i(\sigma)
    =
    \frac{
        w^{\text{spin}}_{G}(\sigma;\beta,\gamma,\lambda_i)
    }{
        w^{\text{spin}}_{G}(\sigma;\beta,\gamma,\lambda_{i-1})
    }
    =
    \left(\frac{\lambda_i}{\lambda_{i-1}}\right)^{n_0(\sigma)}
    \quad \text{where } \sigma\sim\mu_{i-1}.
\]
Let $W_1,W_2,\ldots,W_\ell$ be independent random variables, with $W_i$
sampled using $\mu_{i-1}$, and define $W=\prod_{i=1}^{\ell}W_i$.
A simple calculation shows that
\begin{align}\label{eqn:W-exp}
\Ex[W_i] = \frac{Z_i}{Z_{i-1}} \quad \text{and} \quad \Ex[W] = \frac{Z_\ell}{Z_0}.
\end{align}

The following lemma bounds the variance of $W_i$ and $W$.
\begin{lemma}\label{lem:variance-one-step}
For every $1 \leq i \leq \ell$,
\[
    \Var{}{W_i}
    \leq \Ex[W_i^2]
    =
    \frac{Z_{i+1}}{Z_{i-1}}.
\]
As a consequence,
\begin{align*}
\Var{}{W} \leq \Ex[W^2]
\leq O_{\beta,\gamma,\lambda}(1)\left(\frac{Z_\ell}{Z_0}\right)^2.
\end{align*}
\end{lemma}
\begin{proof}
Since the dependence on $\lambda$ in
$w^{\text{spin}}_{G}(\sigma;\beta,\gamma,\lambda)$ is exactly
$\lambda^{n_0(\sigma)}$, and since
$\lambda_{i-1}\lambda_{i+1}=\lambda_i^2$, we have
\[
    \frac{
        w^{\text{spin}}_{G}(\sigma;\beta,\gamma,\lambda_i)^2
    }{
        w^{\text{spin}}_{G}(\sigma;\beta,\gamma,\lambda_{i-1})
    }
    =
    w^{\text{spin}}_{G}(\sigma;\beta,\gamma,\lambda_{i+1}).
\]
Hence
\[
    \Ex[W_i^2]
    =
    \frac{1}{Z_{i-1}}
    \sum_{\sigma:V\to\{0,1\}}
    w^{\text{spin}}_{G}(\sigma;\beta,\gamma,\lambda_{i+1})
    =
    \frac{Z_{i+1}}{Z_{i-1}}.
\]
Independence gives
\[
    \Ex[W^2]
    =
    \prod_{i=1}^{\ell}\frac{Z_{i+1}}{Z_{i-1}}
    =
    \frac{Z_{\ell+1}Z_\ell}{Z_1Z_0}.
\]
For every adjacent pair $\lambda_j=q^{-1}\lambda_{j-1}$,
\[
    1
    \le
    \frac{Z_j}{Z_{j-1}}
    =
    \Ex_{\mu_{j-1}}\left[
        \left(\frac{\lambda_j}{\lambda_{j-1}}\right)^{n_0(\sigma)}
    \right]
    \le q^{-n}
    =O(1).
\]
Thus $Z_0/Z_1\le 1$ and $Z_{\ell+1}/Z_\ell=O(1)$, which implies
\[
    \Ex[W^2]
    \le
    O(1)\left(\frac{Z_\ell}{Z_0}\right)^2 .
\]
The variance bounds follow from $\Var{}{X}\le \Ex[X^2]$.
\end{proof}

\begin{lemma}\label{lem:estimator-spin}
Assume that we can draw independent samples from $\mu_0,\mu_1,\ldots,\mu_{\ell-1}$. There is an estimator $\widehat{Z}_\ell$ such that, using $N=O_{\beta,\gamma,\lambda}(\epsilon^{-2})$ independent samples of the random variable $W$, 
\[
    \Pr\left[\widehat{Z}_\ell \in (1\pm \epsilon) Z_\ell\right] \geq \frac{9}{10}.
\]
\end{lemma}
\begin{proof}
For each $r=1,\ldots,N$, independently sample $\sigma_i^{(r)}\sim \mu_{i-1}$ for every $i=1,\ldots,\ell$, and define
\[
    W_i^{(r)}
    =
    \left(\frac{\lambda_i}{\lambda_{i-1}}\right)^{n_0(\sigma_i^{(r)})},
    \qquad
    W^{(r)}=\prod_{i=1}^{\ell} W_i^{(r)}.
\]
Let
\[
    \overline{W}=\frac{1}{N}\sum_{r=1}^{N} W^{(r)}
    \qquad\text{and}\qquad
    \widehat{Z}_\ell=\gamma^m\overline{W}.
\]
By \eqref{eqn:W-exp}, $\Ex[W^{(r)}]=Z_\ell/Z_0$. By \Cref{lem:variance-one-step},
\[
    \Var{}{W^{(r)}}\leq O_{\beta,\gamma,\lambda}(1)\left(\frac{Z_\ell}{Z_0}\right)^2.
\]
Since $(W^{(r)})_r$ are independent,
\[
    \Var{}{\overline{W}}
    \leq
    \frac{O_{\beta,\gamma,\lambda}(1)}{N}\left(\frac{Z_\ell}{Z_0}\right)^2.
\]
Taking $N=O_{\beta,\gamma,\lambda}(\epsilon^{-2})$ sufficiently large and applying Chebyshev's inequality gives
\[
    \Pr\left[
    \overline{W}\in \left(1\pm \frac{\epsilon}{3}\right)\frac{Z_\ell}{Z_0}
    \right]\geq \frac{9}{10}.
\]
On this event,
$\frac{\widehat{Z}_\ell}{Z_\ell}
    =
    \frac{\gamma^m}{Z_0}
    \cdot
    \frac{\overline{W}}{Z_\ell/Z_0}$.
By the estimate for $Z_0$ in \Cref{lem:initial-partition-function},
$e^{-\epsilon/10}\leq \frac{\gamma^m}{Z_0}\leq e^{\epsilon/10}$.
Thus, for $0<\epsilon\leq 1$, we have the following bound
\[
    \frac{\widehat{Z}_\ell}{Z_\ell}
    \in
    \left[e^{-\epsilon/10}\left(1-\frac{\epsilon}{3}\right),
    e^{\epsilon/10}\left(1+\frac{\epsilon}{3}\right)\right]
    \subseteq [1-\epsilon,1+\epsilon].
\]
This proves the claim. %The success probability can be amplified in the standard way by taking independent repetitions and returning the median.
\end{proof}

\begin{proof}[Proof of \Cref{thm:fpras-spin}]
We implement the estimator in \Cref{lem:estimator-spin} using the sampler from
\Cref{thm:sampling}. Since all annealing fields satisfy
$\lambda_i\leq \lambda<\lambda^\star$ for all $i=0,\ldots,\ell-1$, the constant in the sampling bound can
be chosen uniformly as $C=C(\beta,\gamma,\lambda)$.

Let $N=O_{\beta,\gamma,\lambda}(\epsilon^{-2})$ be the number of independent
copies of $W$ used by the estimator. Each copy of $W$ requires one sample from
each of $\mu_0,\mu_1,\ldots,\mu_{\ell-1}$, so the total number of samples
needed is $N\ell$. For each required sample, run the sampler of
\Cref{thm:sampling} with total variation error at most $1/(100N\ell)$. Its
running time is
\[
    O\left(\Delta^C n \log(100N\ell n)\right).
\]
By the union bound and the standard coupling interpretation of total variation
distances, all approximate samples can be coupled with exact independent
samples with probability at least $99/100$. Conditional on the success of the
coupling, \Cref{lem:estimator-spin} gives a $(1\pm\epsilon)$ approximation to
$Z_\ell=Z_G^{\text{spin}}(\beta,\gamma,\lambda)$ with probability at least
$9/10$. Hence the overall success probability is at least
$9/10-1/100>2/3$.

The number of annealing stages is
$\ell=O_{\beta,\gamma,\lambda}(n\log(n/\epsilon))$, and each stage requires $N=O_{\beta,\gamma,\lambda}(\epsilon^{-2})$ samples.
Also $\log(100N\ell n)=O_{\beta,\gamma,\lambda}(\log(n/\epsilon))$. 
Therefore the overall running time is
\[
    O\left(
        \Delta^C \cdot \frac{n^2}{\epsilon^2}\cdot
        \log^2\frac{n}{\epsilon}
    \right),
\]
after adjusting the constant $C=C(\beta,\gamma,\lambda)$. This proves the
claimed FPRAS.
\end{proof}

\section{Missing proofs} \label{sec:missing-proofs}

In this section we collate deferred proofs.

\subsection{Validity of the Gibbs distribution}
\label{sec:validify-gibbs}

The positivity result in \Cref{thm:mixing-time-sgw} is a consequence of \Cref{lem:lambda-v-tau}.
\begin{proof}[Proof of \Cref{lem:lambda-v-tau}]
    
  By $\rho<1$, we only need to prove $\rho\lambda^\star<1$. Set parameters
    \[
            t=\sqrt{\beta\gamma},
            \qquad
            h=\sqrt{\gamma/\beta},
            \qquad
            \alpha=\arctan \sqrt{\beta\gamma-1}.
    \]
    Then $t=\sec\alpha$ for $1<h<t$, where \(h>1\) follows from \(\gamma>\beta\), and \(h<t\) follows from
    \(\beta=t/h>1\). Moreover, by definition of $\lambda^\star$, we have
    \[
            \lambda^\star
            =
            \left(\frac{\gamma}{\beta}\right)^{\pi/(2\alpha)}
            =
            h^{\pi/\alpha},
    \]
    and
    \[
            \frac1\rho
            =
            \frac{\gamma-1}{\beta-1}
            =
            \frac{th-1}{t/h-1}.
    \]
    Thus it is enough to prove \(h^{\pi/\alpha}<\frac{th-1}{t/h-1}\).
    Let \(y=\log h\), so \(0<y<\log t\), and define
    \[
            F(y)
            =
            \log\frac{te^y-1}{te^{-y}-1}
            -
            \frac{\pi}{\alpha}y.
    \]
    Then \(F(0)=0\).  We prove \(F'(y)>0\) for
    \(0<y<\log t\). Then it holds that $F(y)>0$. In particular, for $y=\log h$, we have $\log\frac{th-1}{t/h-1}>\frac{\pi}{\alpha}\log h$. So $\frac{1}{\rho}=\frac{th-1}{t/h-1}>h^{\pi/\alpha}=\lambda^\star$. Indeed,
    \[
            F'(y)
            =
            \frac{te^y}{te^y-1}
            +
            \frac{te^{-y}}{te^{-y}-1}
            -
            \frac{\pi}{\alpha}.
    \]
    To prove $F'(y)>0$, we let
    \[
            S(y)
            =
            \frac{te^y}{te^y-1}
            +
            \frac{te^{-y}}{te^{-y}-1}.
    \]
    Since the function \(r\mapsto r/(r-1)^2\) is decreasing for \(r>1\),
    and since \(te^{-y}\le te^y\), we have
    \[
            S'(y)
            =
            -\frac{te^y}{(te^y-1)^2}
            +
            \frac{te^{-y}}{(te^{-y}-1)^2}
            \ge 0.
    \]
    Therefore \(S(y)\ge S(0)=\frac{2t}{t-1}\). Using \(t=\sec\alpha\), this gives \(S(y)\ge \frac{2}{1-\cos\alpha}\).
    Since \(0<\alpha<\pi/2\), the concavity of \(\cos\) on
    \([0,\pi/2]\) implies
    \[
            \cos\alpha>1-\frac{2\alpha}{\pi} \quad\Longrightarrow\quad \frac{2}{1-\cos\alpha}>\frac{\pi}{\alpha}.
    \]
    Thus
    \[
            F'(y)=S(y)-\frac{\pi}{\alpha}>0.
    \]
    Since \(F(0)=0\), we obtain \(F(y)>0\) for every \(y>0\). 
    Equivalently, $\rho\lambda^\star<1$.
\end{proof}

\subsection{Proof of the marginal lower bound}
\label{sec:lb}

\begin{proof}[Proof of \Cref{lem:marginal-lower-bound}]
Fix an edge $e = \{u,v\} \in E$. To prove the marginal lower bound, it suffices to consider the worst pinning $\tau \in \{0,1\}^{E \setminus \{e\}}$ for all edges except $e$. Suppose $d_u$ edges incident to $u$ are pinned to $1$ and $d_v$ edges incident to $v$ are pinned to $1$. Then, the conditional distribution $\mu^\tau_{e}$ satisfies
\begin{align*}
\mu^\tau_{e}(1) &= \frac{\theta (1 + \lambda (-\rho)^{d_u+1}) (1 + \lambda (-\rho)^{d_v+1}) }{\theta (1 + \lambda (-\rho)^{d_u+1}) (1 + \lambda (-\rho)^{d_v+1}) + (1 + \lambda (-\rho)^{d_u}) (1 + \lambda (-\rho)^{d_v})}\\
&= \frac{1}{1 + \frac{(1 + \lambda (-\rho)^{d_u}) (1 + \lambda (-\rho)^{d_v})}{\theta(1 + \lambda (-\rho)^{d_u+1}) (1 + \lambda (-\rho)^{d_v+1})}} \geq \frac{1}{1+ \frac{(1+\lambda)^2}{\theta (1-\lambda \rho)^2} } = \Omega_{\theta, \lambda, \rho}(1).
\end{align*}
The last inequality holds because $\lambda \rho < 1$ by \Cref{lem:lambda-v-tau} and $\rho < 1$ so that $(1 + \lambda (-\rho)^{d_u}) (1 + \lambda (-\rho)^{d_v}) \leq (1 + \lambda)^2$ and $\theta (1 + \lambda (-\rho)^{d_u+1}) (1 + \lambda (-\rho)^{d_v+1}) \geq \theta(1 - \lambda \rho)^2$.
Similarly,
\begin{align*}
\mu^\tau_{e}(0) &= \frac{ (1 + \lambda (-\rho)^{d_u}) (1 + \lambda (-\rho)^{d_v})}{\theta (1 + \lambda (-\rho)^{d_u+1}) (1 + \lambda (-\rho)^{d_v+1}) + (1 + \lambda (-\rho)^{d_u}) (1 + \lambda (-\rho)^{d_v})} \\
&= \frac{1}{1 + \frac{\theta(1 + \lambda (-\rho)^{d_u+1}) (1 + \lambda (-\rho)^{d_v+1})}{(1 + \lambda (-\rho)^{d_u}) (1 + \lambda (-\rho)^{d_v})}} \geq \frac{1}{1+ \frac{\theta (1+\lambda)^2}{(1-\lambda \rho)^2} } = \Omega_{\theta, \lambda, \rho}(1).
\end{align*}
By taking the minimum of the two values, we obtain the marginal lower bound $b = \Omega_{\theta, \lambda, \rho}(1)$.
\end{proof}

\subsection{Other missing proofs in \texorpdfstring{\Cref{sec:stability-of-polynomial}}{Section 4}.}
\label{sec:other-proofs}

\begin{proof}[Proof of \Cref{claim:cstar-rstar}]
    
  First we prove $c^\star>0$. To see this, by \Cref{eq:cstar-rstar}, we write $c^\star=\gamma \log (\sqrt{\gamma/\beta})/\Phi$ with $\Phi=\sqrt{\gamma\beta-1}\arctan \sqrt{\gamma\beta-1}-\log (\sqrt{\gamma/\beta})$. We prove $\Phi>0$. We let
  \[
    s=\sqrt{\gamma\beta-1}\quad\Longrightarrow\quad 1+s^2=\gamma\beta
  \]

  We consider
  \[
  h(t)=t\arctan t-\frac{1}{2}\log(1+t^2)\quad\text{for }t>0.
  \]
  \[
  h'(t)=\arctan t+\frac{t}{1+t^2}-\frac{t}{1+t^2}=\arctan t>0\quad\text{for }t>0.
  \]

  We note that $h(0)=0$, so $h(s)>0$ for $s>0$. Therefore, $s\arctan s>\frac{1}{2}\log(1+s^2)=\frac{1}{2}\log(\gamma\beta)$ for $s>0$. By $\beta>1$, it holds that $\gamma\beta>\frac{\gamma}{\beta}$. Therefore,
  \[
  \frac{1}{2}\log(\gamma\beta)>\frac{1}{2}\log\left(\frac{\gamma}{\beta}\right)=\log \left(\sqrt{\frac{\gamma}{\beta}}\right).
  \]
  Therefore, $\Phi=\sqrt{\gamma\beta-1}\arctan \sqrt{\gamma\beta-1}-\log (\sqrt{\gamma/\beta})>0$, which implies $c^\star>0$.

  Next we prove $r^\star-c^\star>1$. By \Cref{eq:cstar-rstar},
    \[
      (r^\star)^2=(c^\star)^2+\frac{2c^\star}{\beta}+\frac{\gamma}{\beta}.
    \]
    Since $c^\star>0$, proving $r^\star-c^\star>1$ is equivalent to prove $(r^\star)^2>(c^\star+1)^2$, which is equivalent to show
    \[
      \frac{2c^\star}{\beta}+\frac{\gamma}{\beta}>2c^\star+1.
    \]
    Rearranging gives
    \[
      c^\star<\frac{\gamma-
    \beta}{2(\beta-1)}.
    \]
    For notational simplicity, we substitute $c^\star=\gamma L/(s\alpha-L)$ with parameters $L=\log \sqrt{\frac{\gamma}{\beta}}$, $s=\sqrt{\gamma\beta-1}$, $\alpha=\arctan s$, we reduce this to show
    \begin{equation}\label{eq:Hineq}
      (\gamma-
    \beta)s\alpha>(2\beta\gamma-
    \beta-
    \gamma)L.
    \end{equation}
    To verify \eqref{eq:Hineq}, we write $q=\gamma/\beta>1$ and define
    \[
      H_\beta(q)=(q-1)s\arctan s-
      \frac{(2\beta-1)q-1}{2}\log q,
      \qquad s=\sqrt{\beta^2q-1}.
    \]
    Then \eqref{eq:Hineq} is equivalent to $H_\beta(q)>0$.  At $\beta=1$,
    \[
      H_1(q)=(q-1)\left(s\arctan s-
      \frac12\log(1+s^2)\right),
      \qquad s=\sqrt{q-1},
    \]
    By former computation, we already prove $H_1(q)=h(q)>0$.  For $\beta>1$, we prove $H_\beta(q)$ increases as $\beta$ increases. We prove $\partial H_\beta(q)/\partial \beta>0$. By
    \[
    s=\sqrt{\beta^2q-1}\quad\Longrightarrow\quad \frac{\partial s}{\partial \beta}=\frac{q\beta}{s}.
    \]
    Then it holds that
    \[
    \frac{\partial H_\beta(q)}{\partial \beta}=(q-1)\left(\frac{q\beta}{s}\arctan s+\frac{1}{\beta}\right)-q\log q.
    \]
    
    Using $\arctan s>\frac{s}{\sqrt{1+s^2}}$, we have
    \[
    \frac{q\beta}{s}\arctan s>\frac{q\beta}{\sqrt{1+s^2}}=\frac{q\beta}{\sqrt{q}\beta}=\sqrt{q}.
    \]
    Then it holds that
    \begin{align*}
    \frac{\partial H_\beta(q)}{\partial \beta}&>(q-1)\sqrt{q}-q\log q\\
    &=(t^2-1)t-2t^2\log t,\quad \text{where  } t=\sqrt{q}\\
    &=t^2\left(t-\frac{1}{t}-2\log t\right).
    \end{align*}

    Let $f(t)=t^2(t-\frac{1}{t}-2\log t)$. By direct computation, we have $f(1)=0$ and $f'(t)>0$ for $t>1$. Therefore, $f(t)>0$ for $t>1$. Therefore, $\partial H_\beta(q)/\partial \beta>0$. Therefore, $H_\beta(q)>H_1(q)>0$. So we have $c^\star<\frac{\gamma-\beta}{2(\beta-1)}$, which is equavalent to $r^\star-c^\star>1$.
  \end{proof}

\begin{proof}[Proof of \Cref{claim:one-over-rho}]
    We prove that \(1/\rho \notin \operatorname{cl}(D(c^\star,r^\star))\).
    Since \(1/\rho\) and \(c^\star\) are real, it is enough to show
    \[
    \left(\frac1\rho-c^\star\right)^2>(r^\star)^2 .
    \]
    Recall that
    \[
    \rho=\frac{\beta-1}{\gamma-1},
    \qquad
    \frac1\rho=\frac{\gamma-1}{\beta-1},
    \]
    and
    \[
    (r^\star)^2=\left(c^\star+\frac1\beta\right)^2
    +\frac{\beta\gamma-1}{\beta^2}.
    \]
    Hence
    \[
    \left(\frac1\rho-c^\star\right)^2-(r^\star)^2
    =
    \left(\frac1\rho\right)^2
    -\frac{\gamma}{\beta}
    -2c^\star\left(\frac1\rho+\frac1\beta\right).
    \]
    A direct simplification gives
    \[
    \left(\frac1\rho\right)^2-\frac{\gamma}{\beta}
    =
    \frac{(\gamma-\beta)(\beta\gamma-1)}
    {\beta(\beta-1)^2},
    \]
    and
    \[
    \frac1\rho+\frac1\beta
    =
    \frac{\beta\gamma-1}{\beta(\beta-1)}.
    \]
    Therefore
    \[
    \left(\frac1\rho-c^\star\right)^2-(r^\star)^2
    =
    \frac{\beta\gamma-1}{\beta(\beta-1)}
    \left(
    \frac{\gamma-\beta}{\beta-1}-2c^\star
    \right).
    \]
    Thus it suffices to prove
    \[
    2c^\star<\frac{\gamma-\beta}{\beta-1}.
    \]
    which is already justified in the proof of \Cref{claim:cstar-rstar}.
\end{proof}

\begin{proof}[Proof of \Cref{lem:disk-separation}]
    The two roots of $\beta x^2+2x+\gamma=0$ are $\zeta_\pm=\frac{-1\pm i\sqrt{\beta\gamma-1}}{\beta}$.
    By the definition of \(r^\star\), we have 
    \[
    (r^\star)^2=\left(c^\star+\frac1\beta\right)^2+\frac{\beta\gamma-1}{\beta^2}=|\zeta_\pm-c^\star|^2.
    \] 
    Therefore both roots lie on the boundary of the disk \(D(c^\star,r^\star)\). Since \(\epsilon>0\), we have
    \[
    |\zeta_\pm-c^\star|=r^\star>r^\star-\epsilon. 
    \]
    Thus neither root belongs to the closed disk \(\cl{D(c^\star,r^\star-\epsilon)}\). Equivalently, \(\cl{D(c^\star,r^\star-\epsilon)}\cap\{\zeta_1,\zeta_2\}=\emptyset\).
\end{proof}

\begin{proof}[Proof of \Cref{claim:rho-theta}]
    Recall that
    \[
            \rho=\frac{\beta-1}{\gamma-1},
            \qquad
            \theta=\frac{(\gamma-1)^2}{\beta\gamma-1}.
    \]
    Therefore
    \[
            \rho\theta
            =
            \frac{\beta-1}{\gamma-1}\cdot
            \frac{(\gamma-1)^2}{\beta\gamma-1}
            =
            \frac{(\beta-1)(\gamma-1)}{\beta\gamma-1}.
    \]
    Since \(\beta>1\) and \(\gamma>1\), we have
    \[
            \beta\gamma-1-(\beta-1)(\gamma-1)
            =
            \beta+\gamma-2
            >
            0.
    \]
    Hence
    \[
            (\beta-1)(\gamma-1)<\beta\gamma-1,
    \]
    and consequently $\rho\theta<1$.
\end{proof}

\ifthenelse{\boolean{doubleblind}}
{
}
{
    \section*{Acknowledgments} Weiming Feng acknowledges the support of ECS grant 27202725 from Hong Kong RGC. Heng Guo has received funding from the European Research Council (ERC) under the European Union's Horizon 2020 research and innovation programme (grant agreement No.~947778).
}

\bibliographystyle{alpha}
\bibliography{ref}

\newcommand{\etalchar}[1]{$^{#1}$}
\begin{thebibliography}{DGGJ04}

\bibitem[AFF{\etalchar{+}}25]{AFFGW25}
Konrad Anand, Weiming Feng, Graham Freifeld, Heng Guo, and Jiaheng Wang.
\newblock Approximate counting for spin systems in sub-quadratic time.
\newblock {\em TheoretiCS}, 4, 2025.

\bibitem[ALO24]{ALO24}
Nima Anari, Kuikui Liu, and Shayan {Oveis Gharan}.
\newblock Spectral independence in high-dimensional expanders and applications
  to the hardcore model.
\newblock {\em {SIAM} J. Comput.}, 53(6):S20--1, 2024.

\bibitem[Asa70]{Asa70}
Taro Asano.
\newblock Theorems on the partition functions of the {H}eisenberg ferromagnets.
\newblock {\em J. Phys. Soc. Japan}, 29:350--359, 1970.

\bibitem[Bar16]{Bar16}
Alexander~I. Barvinok.
\newblock {\em Combinatorics and Complexity of Partition Functions}, volume~30
  of {\em Algorithms and combinatorics}.
\newblock Springer, 2016.

\bibitem[BB09]{BorceaBranden2009}
Julius Borcea and Petter Br{\"a}nd{\'e}n.
\newblock P{\'o}lya--{S}chur master theorems for circular domains and their
  boundaries.
\newblock {\em Ann. of Math.}, 170(1):465--492, 2009.

\bibitem[CCLZ26]{CCLZ26}
Xiaoyu Chen, Zongchen Chen, Kuikui Liu, and Xinyuan Zhang.
\newblock Subquadratic counting via perfect marginal sampling.
\newblock {\em arXiv}, abs/2604.02235, 2026.

\bibitem[CCYZ25]{CCYZ25}
Xiaoyu Chen, Zongchen Chen, Yitong Yin, and Xinyuan Zhang.
\newblock Rapid mixing at the uniqueness threshold.
\newblock In {\em {STOC}}, pages 879--890. {ACM}, 2025.

\bibitem[CG24]{CG24}
Zongchen Chen and Yuzhou Gu.
\newblock Fast sampling of $b$-matchings and $b$-edge covers.
\newblock In {\em {SODA}}, pages 4972--4987. {SIAM}, 2024.

\bibitem[CLV21]{CLV21}
Zongchen Chen, Kuikui Liu, and Eric Vigoda.
\newblock Optimal mixing of {G}lauber dynamics: entropy factorization via
  high-dimensional expansion.
\newblock In {\em {STOC}}, pages 1537--1550. {ACM}, 2021.

\bibitem[CLV23]{CLV23}
Zongchen Chen, Kuikui Liu, and Eric Vigoda.
\newblock Rapid mixing of {G}lauber dynamics up to uniqueness via contraction.
\newblock {\em {SIAM} J. Comput.}, 52(1):196--237, 2023.

\bibitem[CLV24]{CLV24}
Zongchen Chen, Kuikui Liu, and Eric Vigoda.
\newblock Spectral independence via stability and applications to {H}olant-type
  problems.
\newblock {\em TheoretiCS}, 3, 2024.

\bibitem[CZ23]{CZ23}
Xiaoyu Chen and Xinyuan Zhang.
\newblock A near-linear time sampler for the {I}sing model with external field.
\newblock In {\em {SODA}}, pages 4478--4503. {SIAM}, 2023.

\bibitem[DGGJ04]{DGGJ04}
Martin~E. Dyer, Leslie~Ann Goldberg, Catherine~S. Greenhill, and Mark Jerrum.
\newblock The relative complexity of approximate counting problems.
\newblock {\em Algorithmica}, 38(3):471--500, 2004.

\bibitem[ES88]{ES88}
Robert~G. Edwards and Alan~D. Sokal.
\newblock Generalization of the {F}ortuin-{K}asteleyn-{S}wendsen-{W}ang
  representation and {M}onte {C}arlo algorithm.
\newblock {\em Phys. Rev. D (3)}, 38(6):2009--2012, 1988.

\bibitem[FGW23]{FGW23}
Weiming Feng, Heng Guo, and Jiaheng Wang.
\newblock Swendsen-{W}ang dynamics for the ferromagnetic {I}sing model with
  external fields.
\newblock {\em Inf. Comput.}, 294:105066, 2023.

\bibitem[FGY26]{feng2026rapid}
Weiming Feng, Heng Guo, and Minji Yang.
\newblock Rapid mixing in positively weighted restricted {B}oltzmann machines.
\newblock In {\em {FOCS}}, 2026.
\newblock to appear, available at \url{abs/2604.00963}.

\bibitem[FK72]{FK72}
Cees~M. Fortuin and Piet~W. Kasteleyn.
\newblock On the random-cluster model. {I}. {I}ntroduction and relation to
  other models.
\newblock {\em Physica}, 57:536--564, 1972.

\bibitem[GJ09]{GJ09}
Geoffrey~R. Grimmett and Svante Janson.
\newblock Random even graphs.
\newblock {\em Electron. J. Comb.}, 16(1), 2009.

\bibitem[GJ18]{GuoJ18}
Heng Guo and Mark Jerrum.
\newblock Random cluster dynamics for the {I}sing model is rapidly mixing.
\newblock {\em Ann. Appl. Probab.}, 28(2):1292--1313, 2018.

\bibitem[GJP03]{GoldbergJP03}
Leslie~Ann Goldberg, Mark Jerrum, and Mike Paterson.
\newblock The computational complexity of two-state spin systems.
\newblock {\em Random Struct. Algorithms}, 23(2):133--154, 2003.

\bibitem[GL18]{GuoL18}
Heng Guo and Pinyan Lu.
\newblock Uniqueness, spatial mixing, and approximation for ferromagnetic
  2-spin systems.
\newblock {\em {ACM} Trans. Comput. Theory}, 10(4):17:1--17:25, 2018.

\bibitem[GLL20]{GLL20}
Heng Guo, Jingcheng Liu, and Pinyan Lu.
\newblock Zeros of ferromagnetic 2-spin systems.
\newblock In {\em {SODA}}, pages 181--192. {SIAM}, 2020.

\bibitem[GLLZ21]{GLLZ21}
Heng Guo, Chao Liao, Pinyan Lu, and Chihao Zhang.
\newblock Zeros of {H}olant problems: Locations and algorithms.
\newblock {\em {ACM} Trans. Algorithms}, 17(1):4:1--4:25, 2021.

\bibitem[G{\v{S}}V16]{GSV16}
Andreas Galanis, Daniel {\v{S}}tefankovi\v{c}, and Eric Vigoda.
\newblock Inapproximability of the partition function for the antiferromagnetic
  {I}sing and hard-core models.
\newblock {\em Combin. Probab. Comput.}, 25(4):500--559, 2016.

\bibitem[JS93]{JerrumS93}
Mark Jerrum and Alistair Sinclair.
\newblock Polynomial-time approximation algorithms for the {I}sing model.
\newblock {\em {SIAM} J. Comput.}, 22(5):1087--1116, 1993.

\bibitem[Kol18]{Kol18}
Vladimir Kolmogorov.
\newblock A faster approximation algorithm for the {G}ibbs partition function.
\newblock In {\em {COLT}}, volume~75 of {\em Proceedings of Machine Learning
  Research}, pages 228--249. {PMLR}, 2018.

\bibitem[LLY13]{LLY13}
Liang Li, Pinyan Lu, and Yitong Yin.
\newblock Correlation decay up to uniqueness in spin systems.
\newblock In {\em {SODA}}, pages 67--84. {SIAM}, 2013.

\bibitem[LLZ14]{LiuLZ14}
Jingcheng Liu, Pinyan Lu, and Chihao Zhang.
\newblock The complexity of ferromagnetic two-spin systems with external
  fields.
\newblock In {\em {RANDOM}}, LIPIcs, pages 843--856. Schloss Dagstuhl -
  Leibniz-Zentrum f{\"{u}}r Informatik, 2014.

\bibitem[PR17]{PR17}
Viresh Patel and Guus Regts.
\newblock Deterministic polynomial-time approximation algorithms for partition
  functions and graph polynomials.
\newblock {\em {SIAM} J. Comput.}, 46(6):1893--1919, 2017.

\bibitem[Rue71]{Rue71}
David Ruelle.
\newblock Extension of the {L}ee-{Y}ang circle theorem.
\newblock {\em Phys. Rev. Lett.}, 26:303--304, 1971.

\bibitem[Sly10]{Sly10}
Allan Sly.
\newblock Computational transition at the uniqueness threshold.
\newblock In {\em {FOCS}}, pages 287--296. {IEEE} Computer Society, 2010.

\bibitem[SS14]{SS14}
Allan Sly and Nike Sun.
\newblock Counting in two-spin models on {$d$}-regular graphs.
\newblock {\em Ann. Probab.}, 42(6):2383--2416, 2014.

\bibitem[SST14]{SST14}
Alistair Sinclair, Piyush Srivastava, and Marc Thurley.
\newblock Approximation algorithms for two-state anti-ferromagnetic spin
  systems on bounded degree graphs.
\newblock {\em J. Stat. Phys.}, 155(4):666--686, 2014.

\bibitem[SVV09]{SVV09}
Daniel Stefankovic, Santosh~S. Vempala, and Eric Vigoda.
\newblock Adaptive simulated annealing: {A} near-optimal connection between
  sampling and counting.
\newblock {\em J. {ACM}}, 56(3):18:1--18:36, 2009.

\bibitem[Val08]{Valiant08}
Leslie~G. Valiant.
\newblock Holographic algorithms.
\newblock {\em {SIAM} J. Comput.}, 37(5):1565--1594, 2008.

\bibitem[vdW41]{Wae41}
Bartel~L. van~der Waerden.
\newblock Die lange {R}eichweite der regelm{\"a}{\ss}igen {A}tomanordnung in
  {M}ischkristallen.
\newblock {\em Zeitschrift f{\"u}r Physik}, 118(7):473--488, 1941.

\bibitem[Wei06]{weitz2006counting}
Dror Weitz.
\newblock Counting independent sets up to the tree threshold.
\newblock In {\em {STOC}}, pages 140--149. {ACM}, 2006.

\end{thebibliography}

\end{document}